\documentclass[11pt, oneside]{article}
\usepackage{amsmath, amsthm, amssymb, calrsfs, wasysym, verbatim, bbm, color, graphics, geometry, amsbsy, amsthm, amsfonts, latexsym, graphicx, xifthen, manfnt, ifthen, threeparttable, booktabs, multirow, mdsymbol, placeins, caption, subfig, mathrsfs}
\usepackage[colorlinks=true, allcolors=blue]{hyperref}
\usepackage{enumerate}
\usepackage{natbib}
\usepackage{breakurl}  
\usepackage{url}

\def\bse{\begin{eqnarray*}}
\def\ese{\end{eqnarray*}}
\def\be{\begin{eqnarray}}
\def\ee{\end{eqnarray}}
\def\bq{\begin{equation}}
\def\eq{\end{equation}}

\DeclareMathOperator*{\argmax}{\arg\!\max}
\DeclareMathOperator*{\argmin}{\arg\!\min}

\newtheorem{theorem}{Theorem}
\newtheorem{lemma}[theorem]{Lemma}

\newtheorem{condition}{Condition}

\renewenvironment{proof}{\noindent{\bf Proof}\hspace*{1em}}{\qed\bigskip\\}
\newenvironment{proof-sketch}{\noindent{\bf Sketch of Proof}
  \hspace*{1em}}{\qed\bigskip\\}
\newenvironment{proof-idea}{\noindent{\bf Proof Idea}
  \hspace*{1em}}{\qed\bigskip\\}
\newenvironment{proof-of-lemma}[1][{}]{\noindent{\bf Proof of Lemma {#1}}
  \hspace*{1em}}{\qed\bigskip\\}
\newenvironment{proof-of-proposition}[1][{}]{\noindent{\bf
    Proof of Proposition {#1}}
  \hspace*{1em}}{\qed\bigskip\\}
\newenvironment{proof-of-theorem}[1][{}]{\noindent{\bf Proof of Theorem {#1}}
  \hspace*{1em}}{\qed\bigskip\\}
\newenvironment{inner-proof}{\noindent{\bf Proof}\hspace{1em}}{
  $\bigtriangledown$\medskip\\}
\newenvironment{proof-attempt}{\noindent{\bf Proof Attempt}
  \hspace*{1em}}{\qed\bigskip\\}

\newenvironment{remark}{\noindent{\bf Remark}
  \hspace*{1em}}{\bigskip}

\newcounter{example}

\newenvironment{example*}[1][]{
  \ifthenelse{\isempty{#1}}{%
    \noindent \textbf{Example:}\hspace*{.05em}
  }{%
    \noindent \textbf{Example} ({#1})\textbf{:}\hspace*{.05em}
  }
}{%
  $\clubsuit$ \bigskip
}

\title{Post-Selection Inference for the Cox Model with Interval-Censored Data}
  \author{Jianrui Zhang
  \\
    Department of Statistics and Probability, Michigan State University,\\
   Chenxi Li \\
    Department of Epidemiology and Biostatistics, Michigan State University,\\
   Haolei Weng\\
    Department of Statistics and Probability, Michigan State University
  }
\date{}

\begin{document}

\maketitle
\vspace{.25in}

\begin{abstract}
We develop a post-selection inference method for the Cox proportional hazards model with interval-censored data, which provides asymptotically valid p-values and confidence intervals conditional on the model selected by lasso. The method is based on a pivotal quantity that is shown to converge to a uniform distribution under local alternatives. The proof can be adapted to many other regression models, which is illustrated by the extension to generalized linear models and the Cox model with right-censored data. Our method involves estimation of the efficient information matrix, for which several approaches are proposed with proof of their consistency. Thorough simulation studies show that our method has satisfactory performance in samples of modest sizes. The utility of the method is illustrated via an application to an Alzheimer's disease study.
\end{abstract}

\section{Introduction}
\label{section.intro.intro}
Variable selection has drawn tremendous interest from statistical researchers and applied statisticians in the past thirty years. Traditional methods such as best subset selection, forward selection, backward elimination and stepwise regression, are discrete processes and so cause high variance to the selected model. By contrast, penalized variable selection, namely minimizing a penalized risk function with a sparsity inducing penalty, is a more stable variable selection framework, and the resulting parameter estimators' properties are easier to study \citep{fan2006statistical, fan2010selective}. Lasso, proposed by \citet{tibshirani1996regression}, is the most well-known penalized variable selection method and has been proven successful in many areas of application \citep{hastie2015statistical}. It considers a $L_{1}$ penalty function, which leads to both sparsity in parameter estimators and convexity of the penalized risk function, thus attaining variable selection in a short computation time.

Despite the good performance of lasso in variable selection, the asymptotic distribution of the lasso estimator is complicated even in linear regression \citep{fu2000asymptotics}, which becomes an obstacle to hypothesis testing and  confidence interval construction regarding regression parameters. To overcome this issue, one active line of research builds on debiased lasso estimators to obtain tractable inference procedures \citep{zhang2014confidence, van2014asymptotically, javanmard2014confidence}, which will not be the focus of this paper. Alternatively, to do lasso-based inference, one may think of using lasso to select variables and then performing inference using standard methods (e.g., likelihood method) for the selected model. This procedure brings other issues, one of which is that the target of inference could change with the selected model \citep{berk2013valid}. Besides, the $p$-values and confidence intervals from classical inference procedures (e.g., Wald tests) based on the selected model do not enjoy the classic frequentist property (see numerical results in \citet{taylor2015statistical, lee2016exact, taylor2018post}), as the stochastic aspect of model selection is not accounted for.

The aforementioned problem of inference after model selection is known as post-selection inference. Problems related to post-selection inference have been recognized and studied for a long time \citep{hotelling1940selection, olshen1973conditional, potscher1991effects, leeb2005model, leeb2006can}. In the past decade, significant progress has been made, and statistical inference methods accounting for various data-driven model selection procedures have been emerging rapidly. Despite the fast-growing literature in this area, the majority of the works focus on linear regression, generalized linear models (GLM), or the Cox model for right censored data \citep{berk2013valid, fithian2014optimal, lee2016exact, tibshirani2016exact, taylor2018post, bachoc2019valid, bachoc2020uniformly}. We refer to the review paper \citep{kuchibhotla2022post} and references therein for many other important works. In clinical and epidemiological studies of chronic diseases, investigators often want to know, among a large number of considered predictors, which are independently associated with a time-to-event outcome subject to interval censoring and quantify such associations. For example, with the data of the Alzheimer\textquoteright s Disease Neuroimaging Initiative (ADNI), researchers tried to develop clinical, imaging, genetic, and biochemical biomarkers as predictors of time to the onset of Alzheimer\textquoteright s disease (AD)  \citep[e.g.,][]{li2017prediction, spencer2019combined,li2020penalized, wu2020variable,du2022variable}. Since the ADNI participants were only examined periodically, the time at AD onset could not be observed exactly and was only known to lie between two consecutive examination times. How to identify the important biomarkers and assess the strength of their associations while accounting for the variable selection process and interval censoring is still an open statistical problem.

In this work, we develop a post-lasso inference procedure for the Cox model with interval-censored data under the framework of conditional inference \citep{fithian2014optimal,lee2016exact,taylor2018post}. We first describe some basic principles of the framework here, and then provide more details in Section \ref{review:psi}. Consider a regression problem with $p$ covariates. For a subset $M\subset\left\{ 1,2,\dots,p\right\}$ corresponding to indices of covariates which is regarded as a submodel, the target of the conditional inference is $\beta^M \in \mathbb{R}^{|M|}$, the population regression parameter of the submodel $M$, defined as the minimizer of a population risk function depending on the regression problem, e.g., the mean squared error for linear regression and the Kullback-Leibler divergence for generalized linear models. Let $\hat{M}\subset\left\{ 1,2,\dots,p\right\} $ be the submodel selected by a model selection procedure. The conditional inference aims to construct a confidence interval $C_{j}^{M}$ (or equivalently, conduct hypothesis testing) for $\beta_{j}^{M}$, conditioning on the selection event $\left\{ \hat{M}=M\right\} $. Namely, given a significance level $\alpha$, we hope to find an interval $C_{j}^{M}$ such that
\begin{align}
\label{psi:goal}
\text{P}\left(\beta_{j}^{M}\in C_{j}^{M}\mid\hat{M}=M\right)\ge 1-\alpha, \quad \forall j\in M
\end{align}
 at least asymptotically. Conditioning on the selection event can be justified by the Fisherian proposition of relevance, which requires appropriate conditioning of the hypothetical data samples to ensure relevance of the inference to actual data under study \citep{kuffner2018principled, kuchibhotla2022post}. Moreover,  \citet{lee2016exact} establishes that conditional coverage is a sensible criterion to consider in post-selection inference, by proving that conditional coverage implies the control of the (positive) false coverage-statement rate \citep{benjamini2005false,storey2003positive}.

 In the context of using the lasso as a model selection method, to achieve \eqref{psi:goal} it is computationally more efficient to further condition on the signs of the lasso estimator $\hat{\beta}$:
\begin{align}
\label{psi:goal2}
\text{P}\left(\beta_{j}^{M}\in C_{j}^{M}\mid\hat{M}=M, \hat{s}_{\hat{M}}=s_M\right)\ge 1-\alpha,
\end{align}
where $\hat{M}=\{j:\hat{\beta}_j\neq 0\}$ and $\hat{s}_{\hat{M}}\in \mathbb{R}^{|\hat{M}|}$ is a vector consisting of signs of the non-zeros $\{j\in \hat{M}: {\rm sign}(\hat{\beta}_j)\}$. It is straightforward to confirm that \eqref{psi:goal2} yields \eqref{psi:goal} by marginalizing over all possible sign lists $s_{M}$
associated with $M$. The exact post-selection conditional inference under the lasso has been rigorously studied by \citet{lee2016exact} for Gaussian linear models. See also \citet{markovic2017unifying, mccloskey2020hybrid} for generalizations to random design and non-Gaussian settings. \citet{taylor2018post} extended \citeauthor{lee2016exact}'s \citeyearpar{lee2016exact} method to generalized linear models and the Cox model for right-censored data, but they only provided heuristic justification for the methods and the justification was specific to only generalized linear models, where the risk function in the lasso estimation can be written as a sum of independent and identically distributed (i.i.d.) quantities.

In this paper, we extend \citeauthor{taylor2018post}'s \citeyearpar{taylor2018post} methods to the Cox model for interval-censored data.  We develop the post-lasso conditional inference method without the assumption that the lasso has screened successfully, i.e., the set $M$ does not necessarily include the indices of all the non-zero components of the underlying true regression coefficient $\beta^{\ast}$. This assumption was made by \citet{taylor2018post}. Our work has four contributions. First, we rigorously prove the asymptotic validity of our post-selection inference procedure  for the Cox model with interval-censored data. Second, we establish a general framework to prove the convergence in distribution conditional on the selection event under local alternatives. Our approach is based on Le Cam's third lemma and the asymptotic continuity of empirical process indexed by a Donsker class of functions \citep[Lemma 19.24,][]{van2000asymptotic}, which is different from previous works such as \citet{tian2018selective}. Our framework can be applied to a wide range of parametric and semiparametric models. Besides the Cox model under interval censoring, we apply our framework to GLM and the Cox model for right-censored data to rigorously prove the asymptotic validity of the post-selection inference methods proposed by \citet{taylor2018post}. Third, we establish the asymptotic normality and semiparametric efficiency of a one-step estimator of regression coefficients (analogous to the one in \citet{taylor2018post} to construct post-selection confidence intervals) in the Cox model with interval-censored data. Finally, compared with \citet{taylor2018post}, we use new methods for estimating the efficient information matrix, and prove their consistency.

The rest of the paper is organized as follows. We review the post-selection conditional inference in linear regression and generalized regression models (including GLM and the Cox model), in particular \cite{lee2016exact} and \cite{taylor2018post}, in Sections \ref{section.review.linear} and \ref{section.review.glm} respectively. In Section \ref{section.method}, we propose a one-step estimator for the regression coefficients of the selected model and an associated pivot function to construct post-selection confidence intervals. In Section \ref{information_est}, we give several methods for estimating the efficient information matrix, which is involved in the one-step estimator and the pivot function. The asymptotic results and the idea of our proof are presented in Section \ref{section.result}, with detailed proof deferred to the appendix. Section \ref{section.simulation} shows an extensive simulation study. We illustrate the method via an application to the ADNI data in Section \ref{section.real.data}.

We collect the notations used throughout the paper here for convenience. For a subset $M\subset \{1,2,\ldots, p\}$, a vector $v\in \mathbb{R}^{p}$ and a matrix $X\in \mathbb{R}^{n\times p}$, $|M|$ is the cardinality of $M$, $v_M\in \mathbb{R}^{|M|}$ denotes the subvector of $v$ comprising the elements with indices in $M$, $v_{-M}$ represents the subvector of $v$ with indices in the complement of $M$, $X_M\in \mathbb{R}^{n\times |M|}$ denotes the submatrix of $X$ with columns indexed by $M$, and $X_{-M}$ is the submatrix of $X$ with columns indexed by the complement of $M$. Similarly, for $A\subset \{1,\ldots, n\}$ and $B\subset\{1,\ldots, p\}$, $X_{A,B}$ represents the submatrix of $X$ whose rows and columns are indexed by $A$ and $B$ respectively. We use $\text{diag}(s)$ to denote the diagonal matrix whose diagonal elements are from the vector $s$, $\mathbf{1}$ for a vector of ones, $I_m$ for an $m\times m$ identity matrix, and $I(C)$ to represent the indicator function of the set $C$. For i.i.d. samples $\{Z_i\}_{i=1}^n$, we write $Pf$ for the expectation $Ef(Z_1)$ and abbreviate the average $\sum_{i=1}^nf(Z_i)/n$ to $\mathbb{P}_nf$. For a given vector $v=(v_1,\ldots, v_p)^T\in \mathbb{R}^p$, $v^{\otimes2}=vv^{\top}$, $\|v\|_{\infty}=\max_{i}|v_i|$ and $\|v\|_q=(\sum_{i=1}^p|v_i|^q)^{1/q}, q\in (0,\infty)$. In the context of GLM, $l(\beta)$ represents the log likelihood for a single subject evaluated at $\beta$. In semiparametric models, $\ell(\beta,\Lambda)$ denotes the log likelihood for a single subject evaluated at $(\beta, \Lambda)$. $\tilde{\ell}(\beta^{\ast})$ represents the efficient score function evaluated at true values $(\beta^{\ast}, \Lambda^{\ast})$, and $\mathcal{I}$ denotes the efficient information matrix for $\beta$ at $(\beta^{\ast}, \Lambda^{\ast})$.

\section{Review of Post-Selection Conditional Inference}
\label{review:psi}
\subsection{Review of post-selection conditional inference for linear regression}
\label{section.review.linear}
Consider the lasso for a linear regression model. Let $y \in\mathbb{R}^n$ denote the response vector, $X \in\mathbb{R}^{n\times p}$ denote the design matrix, and $\hat{\beta}$ be the lasso estimator,
\[
\hat{\beta}=\argmin_{\beta\in \mathbb{R}^p}\frac{1}{2}\|y-X\beta\|_{2}^{2}+\lambda\|\beta\|_{1}.
\]
Under a fixed design, \citet{lee2016exact} shows that the event that specifies the model and the signs of the selected variables by lasso
\[
\left\{ y\in\mathbb{R}^{n}:\hat{M}(y)=M,\hat{s}_{\hat{M}}(y)=s_M\right\}
\]
can be represented as a polyhedron $\left\{ y:Ay\le b\right\} $ where
\[
A=\begin{pmatrix}-\text{diag}(s_{M})X_{M}^{+}\\
\frac{1}{\lambda}X_{-M}^{\top}(I-P_{M})\\
-\frac{1}{\lambda}X_{-M}^{\top}(I-P_{M})
\end{pmatrix}\quad \mbox{and} \quad b=\begin{pmatrix}-\lambda\text{diag}(s_{M})(X_{M}^{\top}X_{M})^{-1}s_{M}\\
\mathbf{1}-X_{-M}^{\top}(X_{M}^{+})^{\top}s_{M}\\
\mathbf{1}+X_{-M}^{\top}(X_{M}^{+})^{\top}s_{M}
\end{pmatrix}.
\]
Here, $P_{M}=X_{M}(X_{M}^{\top}X_{M})^{-1}X_{M}^{\top}$ is the projection onto the column span of $X_{M}$ and $X_{M}^{+}=(X_{M}^{\top}X_{M})^{-1}X_{M}^{\top}$ is the pseudo-inverse of $X_{M}$.

Assume that $y\mid X\sim N(\mu,\Sigma)$, where $\mu=X\beta^{\ast}$ and $\Sigma$ is assumed known \citep[post-selection inference when $\Sigma$ is unknown was discussed in Section 8.1 of][]{lee2016exact}. The regression coefficient $\beta_j^{M}$ can be written in the form of $\beta_{j}^{M}=\gamma^{\top}\mu$ with $\gamma=(X_M^{+})^{\top} e_j$, where $e_j$ is the unit vector for the $j$th coordinate. So, a confidence interval for $\beta_{j}^{M}$ with a conditional coverage guarantee can be obtained from the distribution of $\gamma^{\top}y\mid Ay\le b$, which turns out to be closely related to a truncated Gaussian distribution. Specifically, letting $F_{\theta,\sigma^{2}}^{u,v}$ denote the cumulative distribution function (CDF) of $N(\theta,\sigma^{2})$ truncated to the interval $[u,v]$, that is,
\[
F_{\theta,\sigma^{2}}^{u,v}(x)=\frac{\Phi\left\{ (x-\theta)/\sigma\right\} -\Phi\left\{ (u-\theta)/\sigma\right\} }{\Phi\left\{ (v-\theta)/\sigma\right\} -\Phi\left\{ (u-\theta)/\sigma\right\} },
\]
where $\Phi$ is the CDF of $N(0,1)$, \cite{lee2016exact} shows that
\[
F_{\gamma^{\top}\mu,\gamma^{\top}\Sigma\gamma}^{\mathcal{V}^{-},\mathcal{V}^{+}}(\gamma^{\top}y)\mid Ay\le b\sim Unif(0,1),
\]
where
\begin{equation}\label{V-}
    \mathcal{V}^{-}=\max_{j:(Ac)_{j}<0}\frac{b_{j}-(Az)_{j}}{(Ac)_{j}},
\end{equation}
\begin{equation}\label{V+}
    \mathcal{V}^{+}=\min_{j:(Ac)_{j}>0}\frac{b_{j}-(Az)_{j}}{(Ac)_{j}},
\end{equation}
$z=(I_{n}-c\gamma^{\top})y$, and  $c=\Sigma\gamma(\gamma^{\top}\Sigma\gamma)^{-1}$. \citet{lee2016exact} further proves that $F_{\theta,\sigma^{2}}^{a,b}$ is monotone decreasing in $\theta$. Thus a confidence interval $[L,U]$ for $\gamma^{\top}\mu$ with a $1-\alpha$ conditional coverage can be obtained by solving
\[
F_{L,\gamma^{\top}\Sigma\gamma}^{\mathcal{V}^{-},\mathcal{V}^{+}}(\gamma^{\top}y)=1-\frac{\alpha}{2}\quad \mbox{and} \quad F_{U,\gamma^{\top}\Sigma\gamma}^{\mathcal{V}^{-},\mathcal{V}^{+}}(\gamma^{\top}y)=\frac{\alpha}{2}.
\]

\subsection{Review of post-selection conditional inference for generalized linear models and the Cox model}
\label{section.review.glm}
Suppose that the covariates $X_{i}$'s are i.i.d. draws from some distribution $F$ and given $X_i$'s, the response variables $y_{i}$'s independently drawn from some generalized linear model or the Cox model under right censoring, with the true regression parameter being denoted by $\beta^{\ast}$ and of dimension $p$. Let
\[
\hat{\beta}=\argmin_{\beta\in \mathbb{R}^p}-l_n(\beta)+\lambda\|\beta\|_{1}
\]
be the lasso solution, where $l_n(\beta)$ is the log (partial) likelihood for all subjects at $\beta$.
Let $\mathcal{I}$ be the (efficient) information matrix evaluated at $\beta^{\ast}$. We decompose $\mathcal{I}$ as
\[
\mathcal{I}=\begin{pmatrix}\mathcal{I}_{M,M} & \mathcal{I}_{M,-M}\\
\mathcal{I}_{M,-M}^{\top} & \mathcal{I}_{-M,-M}
\end{pmatrix}
\]
so that $\mathcal{I}_{M,M}$ is the information matrix for the submodel $M$ evaluated at $\beta_{M}^{\ast}$.

To perform inference for $\beta^{M}$ conditional on  the selection event $\left\{ \hat{M}=M,\hat{s}_{\hat{M}}=s_M\right\} $, \citet{taylor2018post} assumes that $M$ includes the indices of all the non-zero components of $\beta^{\ast}$ so that $\beta^{M}=\beta^{\ast}_M$, and they use a one-step estimator,
\[
\overline{\beta}^{M}=\hat{\beta}^{M}+\lambda I_{M}(\hat{\beta}^{M})^{-1}s_{M},
\]
where $\hat{\beta}^{M}$ is the solution for $\beta_{M}$ to the following equation,
\[
\frac{\partial l_n}{\partial\beta_{M}}\bigg|_{\beta=(\beta_{M},0)}=\lambda s_{M},
\]
and $I_{M}(\hat{\beta}^{M})$ is the observed information matrix for the submodel $M$ evaluated at $\hat{\beta}^{M}$. Note that
\[\hat{\beta}^{M}=\argmin_{\beta_{M}\in \mathbb{R}^{|M|}}-l_n(\beta_M,0)+\lambda s_{M}^{\top}\beta_{M}\]
and given the selection event $\left\{ \hat{M}=M,\hat{s}_{\hat{M}}=s_M\right\} $, $\hat{\beta}^{M}=\hat{\beta}_{M}$, the lasso estimator for $\beta_M$.

\citet{taylor2018post} argues heuristically that when $\lambda\propto n^{1/2}$,
\begin{equation}\label{asymp_norm_beta_M}
 \overline{\beta}^{M}-\beta^{\ast}_M\approx N(0,n^{-1}\mathcal{I}_{M,M}^{-1})
\end{equation}
and the selection event $\left\{ \hat{M}=M,\hat{s}_{\hat{M}}=s_M\right\} $ is asymptotically equivalent to
an ``active" constraint,
\begin{equation}\label{active_constr}
    \text{diag}(s_M)\left\{\overline{\beta}^{M}-n^{-1}\lambda\mathcal{I}_{M,M}^{-1}s_{M}\right\}>0,
\end{equation}
and an ``inactive" constraint in terms of a random vector that they claimed to be asymptotically independent of $\overline{\beta}^{M}-\beta^{\ast}_M$.

\citet{taylor2018post} also assumes a local alternative $\beta^{\ast}=n^{-1/2}\theta^{\ast}$ and defines $\overline{\theta}^M=n^{1/2}\overline{\beta}^{M}$. They rewrite \eqref{active_constr} as
\[
-\text{diag}(s_{M})\overline{\theta}^{M}<-(n^{-1/2}\lambda)\text{diag}(s_{M})\mathcal{I}_{M,M}^{-1}s_{M},
\] and claim that under the local alternative, \eqref{asymp_norm_beta_M} can be restated as
\[
\overline{\theta}^M=n^{1/2}\overline{\beta}^{M}\approx N(\theta^{\ast}_M,\mathcal{I}_{M,M}^{-1})
\]
and the ``inactive" constraint is asymptotically independent of $\overline{\beta}^{M}-\beta^{\ast}_{M}$.
Therefore, the conditional distribution
\[
\gamma^{\top}\overline{\theta}^M\mid \left\{ \hat{M}=M,\hat{s}_{\hat{M}}=s_M\right\}
\]
is (asymptotically) the same as the conditional distribution
\[
\gamma^{\top}\overline{\theta}^M\mid\left\{-\text{diag}(s_{M})\overline{\theta}^{M}<-(n^{-1/2}\lambda)\text{diag}(s_{M})\mathcal{I}_{M,M}^{-1}s_{M}\right\},
\] and it allows post-selection conditional inference for linear functionals of $\theta^{\ast}_M$, $\gamma^{\top}\theta^{\ast}_M$, and thus also for linear functionals of $\beta^{\ast}_{M}$.

According to Theorem 5.2 in \citet{lee2016exact}, \citet{taylor2018post} claims that
$
F_{\gamma^{\top}\theta^{\ast}_M,\gamma^{\top}\mathcal{I}_{M,M}^{-1}\gamma}^{\mathcal{V}^{-},\mathcal{V}^{+}}(\gamma^{\top}\overline{\theta}^{M})
$
is a conditional pivotal quantity satisfying
\[
F_{\gamma^{\top}\theta^{\ast}_M,\gamma^{\top}\mathcal{I}_{M,M}^{-1}\gamma}^{\mathcal{V}^{-},\mathcal{V}^{+}}(\gamma^{\top}\overline{\theta}^{M})\mid \left\{-\text{diag}(s_{M})\overline{\theta}^{M}<-(n^{-1/2}\lambda)\text{diag}(s_{M})\mathcal{I}_{M,M}^{-1}s_{M}\right\}\stackrel{d}{\approx} Unif(0,1),
\]
 where the truncation limits $\mathcal{V}^{-}$ and $\mathcal{V}^{+}$ are computed as in \eqref{V-} and \eqref{V+} with
\[
A=-\text{diag}(s_{M}),\quad b=-(n^{-1/2}\lambda)\text{diag}(s_{M})\mathcal{I}_{M,M}^{-1}s_{M},\quad c=\mathcal{I}_{M,M}^{-1}\gamma(\gamma^{\top}\mathcal{I}_{M,M}^{-1}\gamma)^{-1},\quad \mbox{and} \quad
z=(I_{|M|}-c\gamma^{\top})\overline{\theta}^M.
\] Based on this pivotal quantity, a post-selection confidence interval for $\gamma^{\top}\theta^{\ast}_M$ can be constructed the same way as in Section \ref{section.review.linear}.
Though $\mathcal{I}_{M,M}$ is unknown, \citet{taylor2018post} proposes using $I_{M}(\hat{\beta}_{M})/n$ as an estimator for it.

\section{Methodology}
\label{section.method}
\subsection{Set-up}
We consider the Cox proportional hazards model with interval-censored data. Suppose that we have a random sample of $n$ independent subjects. Denote the failure time of interest by $T$, and let $X$ be a $p$-dimensional vector of covariates under consideration. For each subject, there exists a random sequence of inspection times $\overrightarrow{U}=(U_{0},U_{1},\dots,U_{K},U_{K+1})$, where $U_{0}=0$, $U_{K+1}=\infty$ and $K$ is a random positive integer. Define $\overrightarrow{\Delta}=(\Delta_{0},\Delta_{1},\dots,\Delta_{K})$ where $\Delta_{k}=I(U_{k}<T\le U_{k+1})$ for $k=0,\dots,K$. The observed data consist of $\left\{ (\overrightarrow{U_{i}},\overrightarrow{\Delta_{i}},X_{i})\right\} _{i=1}^{n}$.

We assume that the inspection process is independent of the failure time conditional on the covariates, i.e.,  $T\bot(K,\overrightarrow{U})\mid X$. Also, we assume that the cumulative hazard function of $T$ conditional on $X$ takes the form
\[
\Lambda(t|X)=\Lambda(t)\exp(\beta^{\top}X),
\]
where $\Lambda(t)$ denotes an unknown baseline cumulative hazard function and $\beta$ is a $p$-dimensional vector of unknown regression coefficients.

\subsection{Post-lasso conditional inference}
\label{section.lasso.one.step}
Let $L_{i}$ and $R_{i}$ denote the two successive inspection times bracketing $T_{i}$. The log likelihood
of the observed data concerning $(\beta,\Lambda)$ is
\begin{equation}\label{loglikIC}
\sum_{i=1}^{n}\log\left[\exp\left\{ -\Lambda(L_{i})\exp(\beta^{\top}X_{i})\right\} -\exp\left\{ -\Lambda(R_{i})\exp(\beta^{\top}X_{i})\right\} \right].
\end{equation}
We use $\ell(\beta,\Lambda)$ to denote the log likelihood for a single subject and $\mathbb{P}_nf$ to denote the expectation of $f$ with respect to the empirical measure based on $n$ independent subjects. Then the log likelihood \eqref{loglikIC} can be written as $n\mathbb{P}_n\ell(\beta,\Lambda)$. The lasso estimator for $\beta$ is given by
\[
(\hat{\beta},\hat{\Lambda})=\argmax_{(\beta,\Lambda)}\mathbb{P}_n\ell(\beta,\Lambda)-\frac{\lambda}{n}\|\beta\|_{1}.
\]
As discussed in \cite{li2020adaptive}, the lasso estimation can be carried out via an EM algorithm.

Let $(\beta^{\ast},\Lambda^{\ast})$ denote the true values of $(\beta,\Lambda)$ and $\mathcal{I}$ denote the efficient information matrix for $\beta$ evaluated at $(\beta^{\ast},\Lambda^{\ast})$. We  decompose $\mathcal{I}$ as
\[
\mathcal{I}=\begin{pmatrix}\mathcal{I}_{M,M} & \mathcal{I}_{M,-M}\\
\mathcal{I}_{M,-M}^{\top} & \mathcal{I}_{-M,-M}
\end{pmatrix}
\]
such that $\mathcal{I}_{M,M}$ is the efficient information matrix for the submodel $M$ evaluated at $(\beta^{\ast},\Lambda^{\ast})$. Following \cite{taylor2018post}, we consider a one-step estimator based on the lasso estimator $\hat{\beta}$ to perform post-selection inference. Since the observed information matrix has no explicit form in the Cox model with interval-censored data, our one-step estimator takes the form of
\begin{equation} \label{eq.one.step.def.1}
\overline{\beta}^{M}=\hat{\beta}^{M}+\lambda\hat{I}_{M,n}^{-1}s_{M},
\end{equation}
where $\hat{I}_{M,n}=(\hat{I_{n}})_{M,M}$,  $\hat{I}_{n}/n$ is a consistent estimator of $\mathcal{I}$, $\hat{\beta}^{M}$ is the solution for $\beta_M$ to the equation,
\begin{equation}\label{submodel_lasso_eq}
 \frac{\partial}{\partial \beta_M}\mathbb{P}_n\ell((\beta_M,0),\hat{\Lambda}_{(\beta_M,0)})=\frac{\lambda}{n}s_M,
\end{equation}
and $\hat{\Lambda}_{\beta}  =\argmax_{\Lambda}\mathbb{P}_n\ell(\beta,\Lambda)$. In the sequel, we let $\hat{\beta}^M_{\textrm{full}}=(\hat{\beta}^M,0)$ and $\hat{\Lambda}^M=\hat{\Lambda}_{\hat{\beta}^M_{\textrm{full}}}$. Note that $\hat{\beta}^M=\hat{\beta}_M$, the lasso estimator for $\beta_M$, conditional on the selection event $\left\{ \hat{M}=M,\hat{s}_{\hat{M}}=s_M\right\}$, and thus we don't have to solve equation \eqref{submodel_lasso_eq} to get $\overline{\beta}^M$ in performing the post-selection conditional inference. Our estimator includes \citeauthor{taylor2018post}'s (\citeyear{taylor2018post}) as a special case, since in GLM, the scaled observed information matrix $I(\hat{\beta})/n$ is consistent for $\mathcal{I}$, which becomes the Fisher information matrix in GLM. Several methods for obtaining $\hat{I}_{n}$ are discussed in the following section, and the consistency of those estimators is proved in Section \ref{section.result.interval}.

Similar to Section \ref{section.review.glm}, we perform post-selection inference under a local alternative $\beta^{\ast}_n=n^{-1/2}\theta^{\ast}$\footnote{The subscript $n$ is added to the true parameter to clarify that we are considering local asymptotics.} where $\theta^{\ast}$ contains some non-zero elements. Let $\overline{\theta}^{M}=n^{1/2}\overline{\beta}^{M}$. We show in Section \ref{section.result.interval} that  \begin{equation}\label{pivot} F_{\gamma^{\top}\widetilde{\theta}^{M},\gamma^{\top}\mathcal{I}_{M,M}^{-1}\gamma}^{\mathcal{V}^{-},\mathcal{V}^{+}}(\gamma^{\top}\overline{\theta}^{M})\mid \left\{ \hat{M}=M,\hat{s}_{\hat{M}}=s_M\right\}\underset{\beta_{n}^{\ast}}{\overset{d}{\longrightsquigarrow}}Unif(0,1),
\end{equation}
where $\mathcal{V}^{-}$ and $\mathcal{V}^{+}$ are computed as in Section \ref{section.review.glm} except that $\mathcal{I}_{M,M}$ is estimated using the methods in Section \ref{information_est} instead of the observed information matrix. Here $\widetilde{\theta}^{M}$ is defined as
\begin{equation} \label{eq.new.target}
\widetilde{\theta}^{M}=\theta_{M}^{\ast}+\mathcal{I}_{M,M}^{-1}\mathcal{I}_{M,-M}\theta_{-M}^{\ast}.
\end{equation}
Then the post-selection confidence intervals for $\gamma^{\top}\widetilde{\theta}^{M}$ can be obtained the same way as before.

It is imperative to elucidate the connections between the true scaled regression coefficient $\theta^{\ast}_M=n^{1/2}\beta_{M,n}^{\ast}$ and our new inference target $\widetilde{\theta}^{M}$ before proceeding to the next subsection. When $M$ contain all the signal covariates, $\theta_{-M}^{\ast}=0$ and so $\widetilde{\theta}^{M}$ recovers exactly $\theta_{M}^{\ast}$. To tackle the case where some signal covariates are not included in $M$, we let $(\beta_{n}^{M},\Lambda_{n}^{M})$ be the minimizer of the population log likelihood under the local alternative $(\beta_{n}^{\ast},\Lambda^{\ast})$ restricted to the submodel $M$, namely,
\begin{equation} \label{eq.new.target.2}
(\beta_{n}^{M},\Lambda_{n}^{M})=\argmax_{(\beta_M,\Lambda)}P_{\beta_{n}^{\ast},\Lambda^{\ast}}\ell((\beta_M,0),\Lambda).
\end{equation}
Let $\theta_{n}^{M}=n^{1/2}\beta_{n}^{M}$ be the scaled minimizer and it is a common inference target to consider under model misspecification \citep[see][]{kuchibhotla2022post}. We show in Section \ref{section.interpretation} that $\widetilde{\theta}^{M}=\theta_{n}^{M}+o(1)$. This means that our method provides asymptotically valid confidence intervals for $\theta_{n}^{M}$ even when some true predictors are not selected, though this is uncommon in our simulation studies (see Section \ref{section.simulation}).

\subsection{Information estimation}\label{information_est}
\subsubsection{The profile likelihood method}
Recall that
$
\hat{\Lambda}_{\beta}  =\argmax_{\Lambda}\mathbb{P}_n\ell(\beta,\Lambda)$ and we let $l(\beta)=\ell(\beta,\hat{\Lambda}_{\beta})
$
so that $l$ is the log profile likelihood (PL) for one subject. \citet{murphy2000profile} proposes using a second-order numerical difference to estimate the efficient information matrix, namely,
\[
(\hat{I}_{n})_{i,j}=-\frac{n\mathbb{P}_n\left\{ l(\hat{\beta}+\epsilon_{n}e_{i}+\epsilon_{n}e_{j})-l(\hat{\beta}+\epsilon_{n}e_{i})-l(\hat{\beta}+\epsilon_{n}e_{j})+l(\hat{\beta})\right\} }{\epsilon_{n}^{2}},
\]
where $e_{i}$ is the $i$th unit vector in $\mathbb{R}^{p}$ and $\epsilon_{n}=O(n^{-1/2})$. The estimator
$\hat{\Lambda}_{\beta}$ can be calculated using the EM algorithm in \cite{li2020adaptive} with $\beta$ being fixed, and so  $\mathbb{P}_nl(\beta)=\mathbb{P}_n\ell(\beta,\hat{\Lambda}_{\beta})$ can be evaluated. \cite{murphy2000profile} proves the consistency
of $\hat{I}_{n}/n$ for $\mathcal{I}$ under mild regularity conditions. However, the PL method is computationally intensive compared to the methods discussed later, and is sensitive to the increment $\epsilon_{n}$ \citep{xu2014standard}. Hence, we do not include the PL method in our simulations and real data analysis.

\subsubsection{The least squares approach}
\label{section.LS}
 Let $(l_{1},u_{1}],\dots,(l_{m},u_{m}]$ denote the maximal intersections that characterize the support of the penalized nonparametric maximum likelihood estimator $\hat{\Lambda}$  \citep{li2020adaptive}. The quantity $\hat{\Lambda}$ can be represented as a linear combination of $I(t\ge u_{k})$ ($k=1,\dots,m$) \citep{li2020adaptive}. Inspired by \cite{huang2012least}, we first obtain a least squares (LS) estimator
\[
\hat{g}_{n}\in\argmin_{g\in\mathcal{G}_n^p}\mathbb{P}_{n}\bigg\|\ell_{\beta}(\hat{\beta},\hat{\Lambda})-\dot{\ell}_{\Lambda}(\hat{\beta},\hat{\Lambda})(g)\bigg\|_2^{2},
\]
where $\mathcal{G}_{n}$ is the linear space generated by $I(t\ge u_{k})$ ($k=1,\dots,m$), $\mathcal{G}_{n}^{p}$ is the space of $p$-dimensional vectors with components in $\mathcal{G}_{n}$, $\ell_{\beta}(\beta,\Lambda)$ is the score function of $\ell(\beta,\Lambda)$ with respect to $\beta$, and $\dot{\ell}_{\Lambda}(\beta,\Lambda)(f)$ is the Fisher score of the one-dimensional parametric submodel $\ell(\beta,\Lambda_{s,f}) $ evaluated at $s=0$ where $\Lambda_{0,f}=\Lambda$ and $\partial\Lambda_{s,f}/\partial s|_{s=0}=f$. Then we estimate the $n\mathcal{I}$ by
\[
\hat{I}_{n}=n\mathbb{P}_{n}\left[\left\{ \ell_{\beta}(\hat{\beta},\hat{\Lambda})-\dot{\ell}_{\Lambda}(\hat{\beta},\hat{\Lambda})(\hat{g}_{n})\right\} ^{\otimes2}\right],
\]
where $v^{\otimes2}=vv^{\top}$ for a vector $v$.
Following \citet{huang2012least}, it can be shown that an equivalent expression of $\hat{I}_n$ is
\[
\hat{I}_{n}=n(A_{11}-A_{12}A_{22}^{-}A_{21}),
\]
where $\tilde{g}=\left(I(t\ge u_{1}),I(t\ge u_{2}),\dots,I(t\ge u_{m})\right)^{\top}$ and
\begin{align*}
A_{11} & =\mathbb{P}_n\left\{ \ell_{\beta}(\hat{\beta},\hat{\Lambda})\right\} ^{\otimes2},\quad A_{12}=\mathbb{P}_n\left\{ \ell_{\beta}(\hat{\beta},\hat{\Lambda})\dot{\ell}_{\Lambda}^{\top}(\hat{\beta},\hat{\Lambda})(\tilde{g})\right\} ,\\
A_{21} & =A_{12}^{\top},\quad A_{22}=\mathbb{P}_n\left\{ \dot{\ell}_{\Lambda}(\hat{\beta},\hat{\Lambda})(\tilde{g})\right\} ^{\otimes2}.
\end{align*}
Here $A^{-}$ denotes the generalized inverse of $A$.

\subsubsection{The PRES method}
\label{section.PRES}
\citet{xu2014standard} studied standard error estimation for the semiparametric models fitted by the EM algorithm. They proposed  a so-called PRES method for estimating the efficient information matrix, which differentiates the score of the expected complete-data log likelihood using Richardson extrapolation with a profile technique. We adapt this method to our setting. Recall that when obtaining the lasso estimator via the EM algorithm in \citet{li2020adaptive}, $\beta$ is updated by maximizing
\[
Q(\beta,\hat{\Lambda}({\beta})\mid\tilde{\beta},\tilde{\Lambda})-\lambda\|\beta\|_{1},
\]
where $(\tilde{\beta},\tilde{\Lambda})$ is the current update of $(\beta,\Lambda)$, $Q$ denotes the expected complete-data log likelihood given $(\tilde{\beta},\tilde{\Lambda})$, and $\hat{\Lambda}({\beta})=\argmax_{\Lambda}Q(\beta,\Lambda\mid\tilde{\beta},\tilde{\Lambda})$. The quantity $Q(\beta,\Lambda\mid\tilde{\beta},\tilde{\Lambda})$ is an explicit function of $(\beta, \Lambda)$, and $\hat{\Lambda}({\beta})$ has an explicit form as well \citep{li2020adaptive}. Define
\[
S(\tilde{\beta})=\frac{\partial Q(\beta,\Lambda|\tilde{\beta},\hat{\Lambda}_{\tilde{\beta}})}{\partial\beta}\bigg|_{\beta=\tilde{\beta},\Lambda=\hat{\Lambda}_{\tilde{\beta}}}.
\]
Standard calculus and probabilistic arguments can show that \[S(\tilde{\beta})=n\mathbb{P}_n l_{\beta}(\tilde{\beta}),\]
the score of the profile log likelihood. Thus we numerically differentiate $S(\tilde{\beta})$ using Richardson extrapolation to estimate $-n\mathbb{P}_n l_{\beta\beta}(\hat{\beta})$, which is used as the estimator $\hat{I}_n$ for $n\mathcal{I}$. Specifically, let
\[
\tilde{\beta}_{1}^{(i)}=\hat{\beta}+\epsilon e_{i},\quad\tilde{\beta}_{2}^{(i)}=\hat{\beta}-\epsilon e_{i},\quad\tilde{\beta}_{3}^{(i)}=\hat{\beta}+2\epsilon e_{i},\quad \text{and}\quad \tilde{\beta}_{4}^{(i)}=\hat{\beta}-2\epsilon e_{i}
\]
for a small $\epsilon$. The $i$th row of $\hat{I}_{n}$ is calculated by
\[
(\hat{I}_{n})_{i,}=-\frac{S\left(\tilde{\beta}_{4}^{(i)}\right)-8\cdot S\left(\tilde{\beta}_{2}^{(i)}\right)+8S\left(\tilde{\beta}_{1}^{(i)}\right)-S\left(\tilde{\beta}_{3}^{(i)}\right)}{12\epsilon}.
\]
The PRES algorithm saves computation time and appears to be more stable to the increment $\epsilon$ compared to the PL method \citep{xu2014standard}.

\subsubsection{The sPRES method} \label{section.sPRES}
As discussed above, the $S(\tilde{\beta})$ in PRES is equal to the score of the profile likelihood evaluated at $\tilde{\beta}$. Note that
\[
n\mathbb{P}_nl_{\beta}(\tilde{\beta})=n\mathbb{P}_n\left\{ \ell_{\beta}(\beta,\hat{\Lambda}_{\beta})+\frac{\partial\hat{\Lambda}_{\beta}^{\top}}{\partial\beta}\frac{\partial\ell(\beta,\hat{\Lambda}_{\beta})}{\partial\hat{\Lambda}_{\beta}}\right\} \bigg|_{\beta=\tilde{\beta}}=n\mathbb{P}_n\ell_{\beta}(\beta,\hat{\Lambda}_{\beta})\bigg|_{\beta=\tilde{\beta}}=n\mathbb{P}_n\ell_{\beta}(\tilde{\beta},\hat{\Lambda}_{\tilde{\beta}}),
\]
where $\partial\ell(\beta,\hat{\Lambda}_{\beta})/\partial\hat{\Lambda}_{\beta}$ is the derivative of the log likelihood with respect to $\hat{\Lambda}_{\beta}$, which can be viewed as a vector of its increments over the maximum intersections \citep{li2020adaptive}. So we have $S(\tilde{\beta})=n\mathbb{P}_n \ell_{\beta}(\tilde{\beta},\hat{\Lambda}_{\tilde{\beta}})$.
Since $\mathbb{P}_n\ell_{\beta}(\tilde{\beta},\hat{\Lambda}_{\tilde{\beta}})$ has a closed form in terms of $(\tilde{\beta},\hat{\Lambda}_{\tilde{\beta}})$ under the Cox model with interval-censored data, we propose a so-called sPRES (simplified PRES) method for obtaining $\hat{I}_n$. Defining $\tilde{S}(\tilde{\beta})=n\mathbb{P}_n \ell_{\beta}(\tilde{\beta},\hat{\Lambda}_{\tilde{\beta}})$, the sPRES method
calculates the $i$th row of $\hat{I}_{n}$ by 
\[
(\hat{I}_{n})_{i,}=-\frac{\tilde{S}\left(\tilde{\beta}_{4}^{(i)}\right)-8\cdot \tilde{S}\left(\tilde{\beta}_{2}^{(i)}\right)+8\tilde{S}\left(\tilde{\beta}_{1}^{(i)}\right)-\tilde{S}\left(\tilde{\beta}_{3}^{(i)}\right)}{12\epsilon}.
\]

\section{Asymptotic Results}
\label{section.result}
\subsection{Idea of proof}
\label{section.idea}
To prove the asymptotic result \eqref{pivot}, we derive the asymptotic distribution of $\overline{\theta}^{M} = n^{1/2}\overline{\beta}^{M}$ conditional on the selection event $\left\{ \hat{M}=M,\hat{s}_{\hat{M}}=s_M\right\}$ under the local alternative $\beta_n^{\ast}=n^{-1/2}\theta^{\ast}$. In this subsection, we give an outline of the derivation with the formal results and detailed proofs relegated to Section \ref{section.result.interval} and the appendix respectively.  Same as \citet{taylor2018post}, we only consider the case where the dimension of  $\beta$, $p$, is fixed.

Following \citet{zeng2016maximum}, we adopt parametric submodels of $\Lambda$ defined by $d\Lambda_{\epsilon,h}=(1+\epsilon h)d\Lambda$, where $h$ runs through the space of bounded functions,  to obtain the efficient score function for $\beta$. Note that these submodels of $\Lambda$ are different from the ones used in \citet{huang2012least}. So the two sets of submodels used two different notations for the model parameter, namely $h$ and $g$.  Let $\Lambda^{\ast}$ denote the true $\Lambda$ and define $\beta^{\ast}=0$, the limit of $\beta_n^{\ast}$ as $n\to\infty$. Under certain regularity conditions, the least favorable direction $h^{\ast}$ \citep{murphy2000profile} at $(\beta^{\ast}, \Lambda^{\ast})$ exists \citep[][pp. 269]{zeng2016maximum}. Define
\[
\tilde{\ell}(\beta)=\ell_{\beta}(\beta,\Lambda^{\ast})-\ell_{\Lambda}(\beta,\Lambda^{\ast})(h^{\ast}),
\] where $\ell_{\Lambda}(\beta,\Lambda)(h)=\partial \ell(\beta,\Lambda_{\epsilon,h})/\partial \epsilon |_{\epsilon=0}$ (note that this nuisance score operator $\ell_{\Lambda}(\beta,\Lambda)$, used by \citet{zeng2016maximum}, is different from the one $\dot{\ell}_{\Lambda}(\beta,\Lambda)$ used in \citet{huang2012least}). Then $\tilde{\ell}(\beta^{\ast})$ is the efficient score function for $\beta$ evaluated at $(\beta^{\ast},\Lambda^{\ast})$. As a $p$-dimensional vector, $\tilde{\ell}(\beta)$ can be decomposed as
\[
\tilde{\ell}(\beta)=(\tilde{\ell}_{M}(\beta)^{\top}, \tilde{\ell}_{-M}(\beta)^{\top})^{\top},
\]
where $\tilde{\ell}_{M}(\beta)$ denotes the subvector of $\tilde{\ell}(\beta)$ corresponding to the submodel $M$. Under $P_{\beta^{\ast},\Lambda^{\ast}}$ and certain regularity conditions, the estimators $\hat{\beta}^M_{\textrm{full}}$ and $\hat{\Lambda}^M$ can be shown to satisfy
\begin{equation}\label{full_model_one_step_est_condition}
n^{1/2}\mathbb{P}_{n}\left\{ \ell_{\beta}(\hat{\beta}^{M}_{\textrm{full}},\hat{\Lambda}^{M})-\tilde{\ell}(\beta^{\ast})\right\} =-n^{1/2}\mathcal{I}\left(\hat{\beta}^{M}_{\textrm{full}}-\beta^{\ast}\right)+o_{P}(1)
\end{equation}
(see Lemma \ref{lemma.one.step.interval}), the one-step estimator $\overline{\beta}^{M}$ defined in equation (\ref{eq.one.step.def.1}) can be shown to satisfy
\begin{equation}\label{asymp_onestepest}
 n^{1/2}\left(\overline{\beta}^{M}-\beta_{M}^{\ast}\right)=n^{1/2}\mathcal{I}_{M,M}^{-1}\mathbb{P}_n\tilde{\ell}_{M}(\beta^{\ast})+o_{P}(1)
\end{equation}
(see Lemma \ref{lemma.one.step.interval}), and by the central limit theorem, we have
\begin{equation}\label{asymp_score}
n^{1/2}\mathbb{P}_n\tilde{\ell}(\beta^{\ast})\stackrel{d}{\longrightsquigarrow}N\left(0,\mathcal{I}\right).
\end{equation}

Now let us discuss how to prove the convergence of $\overline{\theta}^{M}$ in distribution conditional on $\left\{ \hat{M}=M,\hat{s}_{\hat{M}}=s_M\right\}$. Suppose we have a sequence of random vectors $Z_{n}$ such that $Z_{n}\stackrel{d}{\longrightsquigarrow}Z$. The convergence in distribution conditional on some event, taking the form of
\[
Z_{n}\mid A_{n}\stackrel{d}{\longrightsquigarrow}Z\mid A,
\]
cannot hold for any $A_{n}$ and $A$. However, if
\[
A_{n}=\left\{ Z_{n}\in B\right\} ,\ A=\left\{ Z\in B\right\} ,\ \mbox{and}\ P\left(Z\in B\right)>0,
\]
then $Z_{n}\mid A_{n}\stackrel{d}{\longrightsquigarrow}Z\mid A$ holds because
\[
\frac{\text{P}\left(Z_{n}\le t,Z_{n}\in B\right)}{\text{P}\left(Z_{n}\in B\right)}\longrightarrow\frac{\text{P}\left(Z\le t,Z\in B\right)}{\text{P}\left(Z\in B\right)},
\]
where the numerator and the denominator on the left converge to their counterparts on the right respectively.

In our case, the problem is slightly more complicated. Suppose we have two sequences of random vectors $Z_{n}$ and $W_{n}$ such that $(Z_{n},W_{n})\stackrel{d}{\longrightsquigarrow}(Z,W)$ and $Z$ and $W$ are independent. If $R_{n,1}=o_{P}(1)$ and $R_{n,2}=o_{P}(1)$, by similar arguments to the above, it can be shown that
\begin{equation}\label{cond_conv}
Z_{n}\mid \left\{ Z_{n}+R_{n,1}\in B,W_{n}+R_{n,2}\in D\right\} \stackrel{d}{\longrightsquigarrow}Z\mid \left\{ Z\in B\right\} .
\end{equation}
A formal statement and its proof are provided in the appendix (see Lemma \ref{lemma.idea}). Under \eqref{full_model_one_step_est_condition} and some regularity conditions, we can show
\begin{equation}\label{selection_event}
\left\{\hat{M}=M,\hat{s}_{\hat{M}}=s_M\right\}=\left\{ Z_{n}+R_{n,1}\in B,W_{n}+R_{n,2}\in D\right\} ,
\end{equation}
where  $R_{n,1}=o_{P}(1)$ and $R_{n,2}=o_{P}(1)$ under the local alternative $\beta_n^{\ast}=n^{-1/2}\theta^{\ast}$,  $Z_{n}=n^{1/2}\left(\overline{\beta}^{M}-\widetilde{\beta}_{n}^{M}\right)$, and $W_n=n^{1/2}\mathbb{P}_n\tilde{\ell}_{-M}(\beta_{n}^{\ast})-n^{1/2}\mathcal{I}_{-M,M}\left(\overline{\beta}^{M}-\widetilde{\beta}_{n}^{M}\right)$. Here $\widetilde{\beta}_{n}^{M}=n^{-1/2}\widetilde{\theta}^{M}$ and $\widetilde{\theta}^{M}$ is defined in (\ref{eq.new.target}). In fact, $Z_{n}+R_{n,1}\in B$ and $W_{n}+R_{n,2}\in D$ correspond respectively to the active and inactive constraints of the Karush-Kuhn-Tucker (KKT) conditions for the lasso problem,
\[\argmin_{\beta,\Lambda}-\mathbb{P}_n\ell(\beta,\Lambda)+\frac{\lambda}{n}\|\beta\|_{1}.\]
Thus, by \eqref{cond_conv}, we can establish the asymptotic distribution of $\overline{\theta}^M$ conditional on the selection event under the local alternative, if $(Z_{n},W_{n})\underset{\beta_{n}^{\ast}}{\overset{d}{\longrightsquigarrow}}(Z,W)$ for some $Z$ and $W$ that are independent of each other.

The asymptotic distribution of $(Z_{n},W_{n})$ under $\beta^{\ast}$ can be obtained from the asymptotic results \eqref{asymp_onestepest} and \eqref{asymp_score}.
If the model with $\Lambda$ fixed at $\Lambda^{\ast}$ is differentiable in quadratic mean at $\beta^{\ast}$, we can then apply Le Cam's third lemma to transform the asymptotics under $\beta^{\ast}$ to the asymptotics under the local alternative $P_{\beta_n^{\ast},\Lambda^{\ast}}$. In particular, the results \eqref{asymp_onestepest} and \eqref{asymp_score} will be transformed to
\begin{equation}\label{AN_onestepest_local}
n^{1/2}\left(\overline{\beta}^{M}-\widetilde{\beta}_{n}^{M}\right)  =n^{1/2}\mathcal{I}_{M,M}^{-1}\mathbb{P}_n\tilde{\ell}_{M}(\beta_{n}^{\ast})+o_{P}(1)
\end{equation}
and
\begin{equation}\label{AN_score_local}
n^{1/2}\mathbb{P}_n\tilde{\ell}(\beta_{n}^{\ast}) \underset{\beta_{n}^{\ast}}{\overset{d}{\longrightsquigarrow}}N\left(0,\mathcal{I}\right),
\end{equation} 
respectively. Also, for a random vector $Z_{n}$, if $Z_{n}\underset{\beta^{\ast}}{\overset{d}{\longrightsquigarrow}}0$, then $Z_{n}\underset{\beta_{n}^{\ast}}{\overset{d}{\longrightsquigarrow}}0$ by Le Cam's third lemma. A simple application of this is that, if $\hat{I}_{n}/n$ is a consistent estimator for the efficient information matrix $\mathcal{I}$ under the model $P_{\beta^{\ast},\Lambda^{\ast}}$, then $\hat{I}_{n}/n$ is also consistent for $\mathcal{I}$ under $P_{\beta_n^{\ast},\Lambda^{\ast}}$. This result will be also used in deriving  the asymptotic distribution of $\overline{\theta}^M$ conditional on $\left\{ \hat{M}=M,\hat{s}_{\hat{M}}=s_M\right\}$ under the local alternative.

\subsection{Results for the Cox model with interval-censored data}
\label{section.result.interval}
We first give the conditions to establish the asymptotic result \eqref{pivot} under a general semiparametric model $P_{\beta,\Lambda}$ for which the least favorable direction $h^{\ast}$ \citep{murphy2000profile} at $(\beta^{\ast},\Lambda^{\ast})$ exists.

\begin{condition}
\label{condition.ratio}
Under the local alternative $\beta_n^{\ast}=n^{-1/2}\theta^{\ast}$, \[
\liminf_{n\to\infty}P_{\beta_{n}^{\ast},\Lambda^{\ast}}\left(\hat{M}=M,\hat{s}_{\hat{M}}=s_M \right)>0.
\]
\end{condition}

\begin{condition}
\label{condition.lasso.semi}
The tuning parameter $\lambda$ satisfies $\lim_{n\to\infty}\lambda/n^{1/2}=C$ where $C$ is a non-negative constant.
\end{condition}

\begin{condition}
\label{condition.Hessian}
Under $P_{\beta^{\ast},\Lambda^{\ast}}$ where $\beta^{\ast}=0$, $\hat{I}_{n}/n$ is a consistent estimator of the efficient information matrix $\mathcal{I}$ introduced in Section \ref{section.lasso.one.step}.
\end{condition}

\begin{condition}
\label{condition.one.step.semi}
Under $P_{\beta^{\ast},\Lambda^{\ast}}$, the estimator $\hat{\beta}^{M}_{\textrm{full}}=(\hat{\beta}^M,0)$ defined in Section \ref{section.lasso.one.step}
satisfies
\begin{equation} \label{eq.lasso.taylor}
n^{1/2}\mathbb{P}_{n}\left\{ \ell_{\beta}(\hat{\beta}^{M}_{\textrm{full}},\hat{\Lambda}^{M})-\tilde{\ell}(\beta^{\ast})\right\} =-n^{1/2}\mathcal{I}\left(\hat{\beta}^{M}_{\textrm{full}}-\beta^{\ast}\right)+o_{P}(1),
\end{equation}
and the one-step estimator $\overline{\beta}^M$ satisfies
\begin{equation}\label{asymp_eff}
    n^{1/2}\left(\overline{\beta}^{M}-\beta_{M}^{\ast}\right)=n^{1/2}\mathcal{I}_{M,M}^{-1}\mathbb{P}_n\tilde{\ell}_{M}(\beta^{\ast})+o_{P}(1),
\end{equation}
where $\tilde{\ell}(\beta^{\ast})$ is the efficient score function introduced in Section \ref{section.idea}, and $\tilde{\ell}_{M}(\beta^{\ast})$ is the subvector of $\tilde{\ell}(\beta^{\ast})$ with indices in $M$.
In addition, under $P_{\beta^{\ast},\Lambda^{\ast}}$,
\[
n^{1/2}\mathbb{P}_n\tilde{\ell}(\beta^{\ast})\stackrel{d}{\longrightsquigarrow}N\left(0,\mathcal{I}\right).
\]
\end{condition}

\begin{condition}
\label{condition.Donsker.semi}
There exists a neighborhood $V_1$ of $(\beta^{\ast},\Lambda^{\ast})$ such that the class of functions $\mathcal{F}_{1}=\{\ell_{\beta}(\beta,\Lambda): (\beta,\Lambda)\in V_1\}$ is $P_{\beta^{\ast},\Lambda^{\ast}}$-Donsker with a square-integrable envelope function $F_{1}$. There also exists a neighborhood $V_2$ of $(\beta^{\ast},\Lambda^{\ast},h^{\ast})$ such that the class of functions $\mathcal{F}_{2}=\{\ell_{\Lambda}(\beta,\Lambda)(h): (\beta,\Lambda,h)\in V_2\}$ is $P_{\beta^{\ast},\Lambda^{\ast}}$-Donsker with a square-integrable envelope function $F_{2}$.

\end{condition}

Condition \ref{condition.ratio} assumes that we have a non-vanishing probability of selecting $\left\{ M,s_M\right\} $. In our setting, $\text{P}\left( \hat{M}=M,\hat{s}_{\hat{M}}=s_M \right)$ is the inverse of the selective likelihood ratio in \citet{tian2018selective}. So Condition \ref{condition.ratio} is equivalent to the condition that the selective likelihood ratio is bounded, which appears in Lemma 3 of \cite{tian2018selective}. If Condition \ref{condition.ratio} does not hold, it is meaningless to conduct selective inference on $\beta_{M,n}^{\ast}$.

Condition \ref{condition.lasso.semi} is needed to ensure that the lasso estimator $\hat{\beta}$ is $n^{1/2}$-consistent, which is in turn needed for \eqref{asymp_eff} to hold. The condition is also necessary to show \eqref{selection_event}. Condition \ref{condition.Hessian} requires a consistent estimator of the efficient information matrix to be used in computing the one-step estimator $\overline{\beta}^M$. If the model is differentiable in quadratic mean at $\beta^{\ast}$, by Le Cam's third lemma, Condition \ref{condition.Hessian} implies that $\hat{I}_{M,n}/n$ is consistent for $\mathcal{I}_{M,M}$ under the local alternative. The second part of Condition \ref{condition.one.step.semi} is a direct result of the central limit theorem. For the first part of Condition \ref{condition.one.step.semi}, equation (\ref{eq.lasso.taylor}) is a Taylor expansion condition for $\ell_{\beta}(\hat{\beta}^{M}_{\textrm{full}},\hat{\Lambda}^{M})-\tilde{\ell}(\beta^{\ast})$, while (\ref{asymp_eff}) is the asymptotic efficiency of the one-step estimator, as in Theorem 5.45 of \cite{van2000asymptotic} and Theorem 7.2 of \citet{van2002}. Noting that $\overline{\beta}^{M}=\hat{\beta}^{M}+\lambda\hat{I}_{M,n}^{-1}s_M=\hat{\beta}^{M}+(\hat{I}_{M,n}/n)^{-1}\mathbb{P}_{n}\ell_{\beta_{M}}(\hat{\beta}^{M}_{\textrm{full}},\hat{\Lambda}^{M})$ , (\ref{eq.lasso.taylor}) actually implies (\ref{asymp_eff}) as long as Condition \ref{condition.Hessian} holds and $\|\mathbb{P}_{n}\ell_{\beta_M}(\hat{\beta}^{M}_{\textrm{full}},\hat{\Lambda}^{M})\|_{2}=O_{P}(n^{-1/2})$, where the latter is implied by the central limit theorem for $n^{1/2}\mathbb{P}_{n}\tilde{\ell}_M(\beta^{\ast})$ and $\|\hat{\beta}^{M}-\beta^{\ast}_M\|=O_{P}(n^{-1/2})$ (shown in the proof of Lemma \ref{lemma.one.step.interval}) due to (\ref{eq.lasso.taylor}). As discussed in Section \ref{section.idea}, Condition \ref{condition.one.step.semi} actually implies \eqref{AN_onestepest_local} and \eqref{AN_score_local} under the local alternative. A formal proof is given in the appendix (Section \ref{section.one.step.local.semi}). Condition \ref{condition.Donsker.semi} is a technical condition for us to apply the asymptotic continuity of empirical process \citep[Lemma 19.24,][]{van2000asymptotic} in proving \eqref{pivot}. 

\begin{theorem}
\label{theorem.semi}
If Conditions \ref{condition.ratio}--\ref{condition.Donsker.semi} hold under a semiparametric model $P_{\beta,\Lambda}$ which is differentiable in quadratic mean at $\beta^{\ast}$ when fixing $\Lambda$ at $\Lambda^{\ast}$ and for which the least favorable direction $h^{\ast}$ \citep{murphy2000profile} at $(\beta^{\ast},\Lambda^{\ast})$ exists, then
\[
F_{\gamma^{\top}\widetilde{\theta}^{M},\gamma^{\top}\mathcal{I}_{M,M}^{-1}\gamma}^{\mathcal{V}^{-},\mathcal{V}^{+}}(\gamma^{\top}\overline{\theta}^M)\mid \left\{ \hat{M}=M,\hat{s}_{\hat{M}}=s_M\right\} \underset{\beta_{n}^{\ast}}{\overset{d}{\longrightsquigarrow}}Unif(0,1),
\]
where $\widetilde{\theta}^{M}$ is defined in \eqref{eq.new.target}, and the truncation limits $\mathcal{V}^{-}$ and $\mathcal{V}^{+}$ are computed as in \eqref{V-} and \eqref{V+} with
\[
A=-\text{diag}(s_{M}),\quad b=-(n^{-1/2}\lambda)\text{diag}(s_{M})\mathcal{I}_{M,M}^{-1}s_{M}, \quad c=\mathcal{I}_{M,M}^{-1}\gamma(\gamma^{\top}\mathcal{I}_{M,M}^{-1}\gamma)^{-1},\quad \mbox{and} \quad
z=(I_{|M|}-c\gamma^{\top})\overline{\theta}^M.
\]
\end{theorem}

Theorem \ref{theorem.semi} provides a general result for semiparametric models without assuming that $M$ does not miss any significant variable. Thisis a big improvement of the asymptotic result conjectured in \cite{taylor2018post}. It is non-trivial for the lasso estimator to retain all the true predictors under the local alternative where $\beta_{n}^{\ast}$ vanishes at the convergence rate of the lasso itself, namely, $n^{-1/2}$ \citep{fu2000asymptotics}. Substituting $\widetilde{\theta}^{M}$ with the scaled minimizer of population log likelihood, $\theta_{n}^{M}$, defined in Section \ref{section.lasso.one.step}, and $\mathcal{I}_{M,M}$ with a consistent estimator still produces asymptotically valid confidence intervals, as long as the truncation limits $\mathcal{V}^{-}$ and $\mathcal{V}^{+}$ are well-separated with high probability. See Section 3.2 of \cite{tian2017asymptotics} for a detailed discussion on this issue.

We then verify the conditions of Theorem \ref{theorem.semi} under the Cox model with interval-censored data. Lemma \ref{lemma.one.step.interval} below shows that Condition \ref{condition.one.step.semi} holds under Conditions C1--C4 of \citet{li2020adaptive}, which are standard regularity conditions concerning the true parameters $(\beta^{\ast},\Lambda^{\ast})$, covariate vector $X$, and inspection times $\overrightarrow{U}$. These conditions are also a special case of the conditions in \cite{zeng2017maximum} in the context of the Cox model with interval-censored data. Condition \ref{condition.Donsker.semi} can be verified by similar arguments to \cite{zeng2017maximum}. Therefore, the result \eqref{pivot} holds under the Cox model with interval-censored data, which is formally stated in Theorem \ref{theorem.interval}.

\begin{lemma}
\label{lemma.one.step.interval}
Suppose $\lambda=O(n^{1/2})$. For any fixed $M$, the estimator $\hat{\beta}^{M}_{\textrm{full}}$ satisfies
\[
n^{1/2}\mathbb{P}_{n}\left\{ \ell_{\beta}(\hat{\beta}^{M}_{\textrm{full}},\hat{\Lambda}^{M})-\tilde{\ell}(\beta^{\ast})\right\} =-n^{1/2}\mathcal{I}\left(\hat{\beta}^{M}_{\textrm{full}}-\beta^{\ast}\right)+o_{P}(1),
\]
and the one-step estimator $\overline{\beta}^{M}=\hat{\beta}^{M}+\lambda\hat{I}_{M,n}^{-1}s_M$ satisfies
\[
n^{1/2}\left(\overline{\beta}^{M}-\beta_{M}^{\ast}\right)=n^{1/2}\mathcal{I}_{M,M}^{-1}\mathbb{P}_n\tilde{\ell}_{M}(\beta^{\ast})+o_{P}(1)
\]
under the model $P_{\beta^{\ast},\Lambda^{\ast}}$ and Conditions C1--C4 of \citet{li2020adaptive}.
\end{lemma}

\begin{theorem}
\label{theorem.interval}
Suppose Conditions C1--C4 of \cite{li2020adaptive} as well
as Conditions \ref{condition.ratio}--\ref{condition.Hessian} hold under the Cox model with interval-censored data. Then
\[F_{\gamma^{\top}\widetilde{\theta}^{M},\gamma^{\top}\mathcal{I}_{M,M}^{-1}\gamma}^{\mathcal{V}^{-},\mathcal{V}^{+}}(\gamma^{\top}\overline{\theta}^M)\mid \left\{ \hat{M}=M,\hat{s}_{\hat{M}}=s_M\right\} \underset{\beta_{n}^{\ast}}{\overset{d}{\longrightsquigarrow}}Unif(0,1).
\]
\end{theorem}

Theorem \ref{theorem.interval} assumes that we have a consistent estimator for $\mathcal{I}_{M,M}$. We now show that the methods discussed in Sections \ref{section.LS} -- \ref{section.sPRES} yield such consistent estimators under $P_{\beta^{\ast},\Lambda^{\ast}}$. Their consistency under the local alternative $P_{\beta_n^{\ast},\Lambda^{\ast}}$ is a direct result of Le Cam's third lemma, as discussed in Section \ref{section.idea}. Since both PRES and sPRES are numerical methods for estimating  $-\mathbb{P}_nl_{\beta\beta}(\hat{\beta})$, the second derivative of the profile log likelihood evaluated at $\hat{\beta}$, we only prove the consistency of $-\mathbb{P}_nl_{\beta\beta}(\hat{\beta})$. In Theorem \ref{theorem.LS}, we assume that the total variation of the least favorable direction $g^{\ast}$ over the union of the supports of the inspection times is bounded by a known constant $M$, and consequently $\hat{g}_{n}$ is searched for among all the $g$'s in $\mathcal{G}_n$ that have total variation bounded by $M$. This simplifies our proof to a great extent and will not affect the usage in practice, since $g^{\ast}$ has a bounded total variation as shown in the proof of Theorem \ref{theorem.LS} and we can always choose a very large $M$, say, $10^{8}$.

\begin{theorem}
\label{theorem.PRES}
Under $P_{\beta^{\ast},\Lambda^{\ast}}$ and Conditions C1-C4 of \citet{li2020adaptive}, the estimator
$
-\mathbb{P}_nl_{\beta\beta}(\hat{\beta})
$
is consistent for $\mathcal{I}$.
\end{theorem}

\begin{theorem}
\label{theorem.LS}
Let $\mathcal{G}_{M,n}=\left\{ g\in\mathcal{G}_{n}:V(g)\le M\right\} $ for a large constant $M$ such that $V(g^{\ast})<M$, where $V(g)$ denotes the total variation of $g$ over the union of the supports of the inspection times and $g^{\ast}$ is the least favorable direction, and let \[
\hat{g}_{n}=\min_{g\in\mathcal{G}_{M,n}^{p}}\mathbb{P}_n\bigg\|\ell_{\beta}(\hat{\beta},\hat{\Lambda})-\dot{\ell}_{\Lambda}(\hat{\beta},\hat{\Lambda})(g)\bigg\|^{2}.
\] Suppose Conditions C1--C4 of \cite{li2020adaptive} hold. Then
\[
\mathbb{P}_n\left[\left\{ \ell_{\beta}(\hat{\beta},\hat{\Lambda})-\dot{\ell}_{\Lambda}(\hat{\beta},\hat{\Lambda})(\hat{g}_{n})\right\} ^{\otimes2}\right]
\]
is consistent for $\mathcal{I}$ under $P_{\beta^{\ast},\Lambda^{\ast}}$.
\end{theorem}

 \subsection{Results under GLM and the Cox model with right-censored data}
\label{section.glm}
Theorem \ref{theorem.semi} can be adapted to GLM. Again let $l(\beta)$ denote the log likelihood for one subject, $l_{\beta}=\frac{\partial l}{\partial\beta}$ and $l_{\beta_{M}}=\frac{\partial l}{\partial\beta_{M}}$. Let $\mathcal{I}$ be the Fisher information matrix, and $\hat{I}_{n}/n$ be a consistent estimator of $\mathcal{I}$. Similar to equation (\ref{submodel_lasso_eq}), we consider an estimator $\hat{\beta}^{M}$ that solves the following equation,
\begin{equation}\label{eq.one.step.def.3}
 \frac{\partial}{\partial \beta_M}\mathbb{P}_nl((\beta_M,0))=\frac{\lambda}{n}s_M,
\end{equation}
and a corresponding one-step estimator
$
\overline{\beta}^{M}=\hat{\beta}^{M}+\lambda\hat{I}_{M,n}^{-1}s_{M}.
$
Let $\overline{\theta}^{M}=n^{1/2}\overline{\beta}^{M}$ as before. We now present Conditions \ref{condition.one.step.glm} and \ref{condition.Donsker.glm}, which replace Conditions \ref{condition.one.step.semi} and \ref{condition.Donsker.semi} respectively, and then give the asymptotic result for GLM.

\begin{condition}
\label{condition.one.step.glm}
Under $P_{\beta^{\ast}}$ where $\beta^{\ast}=0$, the  estimator for $\beta$, $\hat{\beta}^{M}_{\textrm{full}}=(\hat{\beta}^{M},0)$, satisfies
\begin{equation}\label{condition6-1_glm}
n^{1/2}\mathbb{P}_{n}\left\{ l_{\beta}(\hat{\beta}^{M}_{\textrm{full}})-l_{\beta}(\beta^{\ast})\right\} =-n^{1/2}\mathcal{I}\left(\hat{\beta}^{M}_{\textrm{full}}-\beta^{\ast}\right)+o_{P}(1),
\end{equation}
and the one-step estimator $\overline{\beta}^M$ satisfies
\[
n^{1/2}\left(\overline{\beta}^{M}-\beta_{M}^{\ast}\right)=n^{1/2}\mathcal{I}_{M,M}^{-1}\mathbb{P}_nl_{\beta_{M}}(\beta^{\ast})+o_{P}(1).
\]
In addition, under $P_{\beta^{\ast}}$,
\[
n^{1/2}\mathbb{P}_nl_{\beta}(\beta^{\ast})\stackrel{d}{\longrightsquigarrow}N\left(0,\mathcal{I}\right).
\]
\end{condition}

\begin{condition}
\label{condition.Donsker.glm}
There exists a neighborhood $V$ of $\beta^{\ast}$ such that the class of functions $
\mathcal{F}=\left\{ l_{\beta}(\tilde{\beta}):\tilde{\beta}\in V\right\} ,
$
is $P_{\beta^{\ast}}$-Donsker with a square-integrable envelope function.
\end{condition}

\begin{remark} 
The first part of Condition \ref{condition.one.step.glm} can be proved by standard asymptotic arguments as in Theorem 5.45 of \cite{van2000asymptotic}, and the second part is simply an application of the central limit theorem. Condition \ref{condition.Donsker.glm} can be easily verified for a specific GLM with explicit log likelihood function. In addition, we can simply use the observed information matrix $I(\hat{\beta})$ in place of $\hat{I}_{n}$ in Condition \ref{condition.Hessian}.
\end{remark}

\begin{theorem}
\label{theorem.GLM}
If Conditions \ref{condition.ratio}--\ref{condition.Hessian}, \ref{condition.one.step.glm}, and \ref{condition.Donsker.glm} hold under a specific GLM that is differentiable in quadratic mean at $\beta^{\ast}$, then
\[
F_{\gamma^{\top}\widetilde{\theta}^{M},\gamma^{\top}\mathcal{I}_{M,M}^{-1}\gamma}^{\mathcal{V}^{-},\mathcal{V}^{+}}(\gamma^{\top} \overline{\theta}^{M})\mid \left\{ \hat{M}=M,\hat{s}_{\hat{M}}=s_M\right\} \underset{\beta_{n}^{\ast}}{\overset{d}{\longrightsquigarrow}}Unif(0,1).
\]
\end{theorem}

We can also prove a similar result for the Cox model with right censored data. Due to the use of partial likelihood for estimating the regression coefficients and the associated asymptotic theory \citep{andersen1982cox}, the Cox model with right-censored data is similar to GLM in terms of the estimation and inference for the regression coefficients, which enables us to adapt the proof of Theorem \ref{theorem.GLM} to Theorem \ref{theorem.right} below.

\begin{theorem}
\label{theorem.right}
Suppose Conditions A--D in \citet{andersen1982cox} as well as Condition \ref{condition.ratio}-- \ref{condition.lasso.semi}  hold under the Cox model with right-censored data. Then
\[
F_{\gamma^{\top}\widetilde{\theta}^{M},\gamma^{\top}\mathcal{I}_{M,M}^{-1}\gamma}^{\mathcal{V}^{-},\mathcal{V}^{+}}(\gamma^{\top}\overline{\theta}^{M})\mid \left\{ \hat{M}=M,\hat{s}_{\hat{M}}=s_M\right\} \underset{\beta_{n}^{\ast}}{\overset{d}{\longrightsquigarrow}}Unif(0,1).
\]
\end{theorem}

\section{Simulations}
\label{section.simulation}
Numerical experiments are conducted to evaluate the finite-sample performance of our method. Our simulation scenario is similar to that of \cite{li2020adaptive}. Covariates are generated from $N_{10}(0,\Sigma)$ where $\Sigma$ is the covariance matrix with pairwise covariance $\Sigma_{ij}=0.2^{|i-j|}$. Failure times are generated from the Cox proportional hazards model with the cumulative hazard function $\Lambda(t|X_{i})=\Lambda^{\ast}(t)\exp(\beta^{\top}X_{i})$, where $\Lambda^{\ast}(t)=(\eta t)^{\kappa}$ is the Weibull cumulative hazard with $\kappa=1.5$ and $\eta=0.5$. We set $\beta_{j}=0$ for $j=3,\dots,8$ and the rest to be nonzero. For the magnitude of $\beta$, we consider two settings: $\beta_{j}=1$ for $j=1,2,9,10$, representing a strong signal setting, and $\beta_{j}=0.5$ for $j=1,2,9,10$, representing a weak signal setting. There are three inspection times generated from $U_{1}\sim Unif(3.2,4.8)$, $U_{2}=U_{1}+Unif(1.5,2.5)$, $U_{3}=U_{2}+Unif(1.5,2.5)$ for $\beta=1$, and $U_{1}\sim Unif(2.4,3.2)$, $U_{2}=U_{1}+Unif(1.5,2.5)$, $U_{3}=U_{2}+Unif(1.5,2.5)$ for $\beta=0.5$. The proportions of right-censored subjects are $31\%$ and $27.6\%$ for $\beta=1$ and $\beta=0.5$ respectively. We use two sample sizes, $n=200$ and $n=400$. Tuning parameters are set as $\lambda=Cn^{1/2}$ or selected via AIC. The former choice corresponds to the theoretical results, while the latter may be preferred in practice. We set $C=8.5$ for $\beta=1$ and $C=14.5$ for $\beta=0.5$. Two hundred Monte Carlo runs are carried out under each setting. In all the simulations, the model selected by lasso contains all the signal covariates.

Recall that our method consists of two steps. The first step is to obtain the one-step estimator, and the second step is to compute the pivot function. The matrix $\hat{I}_{M,n}$ has been used once in each of the two steps, and thus we may use two different $\hat{I}_{M,n}$'s in these two steps respectively. This maintains the asymptotic results as long as both $\hat{I}_{M,n}/n$'s are consistent for the efficient information matrix $\mathcal{I}_{M,M}$. By simulations,  we found that the method using the least squares (LS) approach in the first step did not perform well in finite samples. Thus, we only report the results of simulations using PRES (or sPRES) in
the first step and the LS approach in the second step, abbreviated as ``PRES+LS'' (or ``sPRES+LS''), as well as simulations using PRES (or sPRES) in both steps, simply referred to as ``PRES'' (or ``sPRES'').

Table \ref{tab:1.fixed} and \ref{tab:05.fixed} respectively show the coverage of the proposed confidence interval for $\beta=1$ and $\beta=0.5$ under $\lambda=Cn^{1/2}$ conditional on the event that the model selected by lasso contains all the signal covariates. The nominal coverage is $0.95$. Both PRES+LS and sPRES+LS appear to be conservative but improve with sample size. This is because the LS approach gives a smaller estimate of the efficient information matrix than PRES and sPRES in the simulations. PRES and sPRES are more anti-conservative than their counterparts and already perform well when $n=200$. Following \cite{xu2014standard}, we consider the increments $\epsilon$ to be $10^{-2}$ as well as $10^{-5}$. To explore the robustness of PRES and sPRES to small increments, we also set $\epsilon=10^{-7}$. As can be seen from the tables, none of the methods is sensitive to the choice of $\epsilon$. The simulation results with $\lambda$ selected via AIC are given in the appendix. Though without theory support, the results are satisfactory in terms of empirical coverage and its improvement with the increase of sample size.

In addition, we draw the uniform Q-Q plots of the null $p$-values, whose distribution should be $Unif(0,1)$ asymptotically according to the theory. Figure \ref{fig:1.5.fixed} shows that the distribution of the null $p$-value converges to $Unif(0,1)$ when the sample size increases, for $\lambda=Cn^{1/2}$ and $\epsilon=10^{-5}$ under $\beta=1$ and $\beta=0.5$ respectively. The Q-Q plots in the other settings are given in the appendix.

\begin{table}
\centering
  \caption{Simulation results on the coverage of confidence intervals for $\beta=1, \lambda=Cn^{1/2}$}
  \begin{threeparttable}
    \begin{tabular}{cccccccccc}
    \hline
    \multirow{2}[4]{*}{Coverage} & \multicolumn{4}{c}{$n = 200$} &       & \multicolumn{4}{c}{$n = 400$} \\
\cmidrule{2-5} \cmidrule{7-10}
 & $\beta_1$ & $\beta_2$ & $\beta_9$ & $\beta_{10}$ &       & $\beta_1$ & $\beta_2$ & $\beta_9$ & $\beta_{10}$ \\ \hline
    \multicolumn{4}{l}{\textit{PRES+LS}}      &       &       &       &       &  \\
    $\epsilon=10^{-2}$ & 0.99 & 0.995 & 0.99 & 0.99 &       & 0.965 & 0.975 & 0.95 & 0.975 \\
    $\epsilon=10^{-5}$ & 0.985 & 0.995 & 0.99 & 0.985 &       & 0.965 & 0.975 & 0.95 & 0.975 \\
    $\epsilon=10^{-7}$ & 0.985 & 0.995 & 0.99 & 0.985 &       & 0.965 & 0.975 & 0.95 & 0.975 \\[5pt]
    \multicolumn{4}{l}{\textit{sPRES+LS}} &       &       &       &       &  \\
    $\epsilon=10^{-2}$ & 0.99 & 0.995 & 0.99 & 0.99 &       & 0.965 & 0.975 & 0.95 & 0.975 \\
    $\epsilon=10^{-5}$ & 0.99 & 0.995 & 0.99 & 0.99 &      & 0.965 & 0.975 & 0.95 & 0.975 \\
    $\epsilon=10^{-7}$ & 0.99 & 0.995 & 0.99 & 0.99 &       & 0.965 & 0.975 & 0.95 & 0.975 \\[5pt]
    \multicolumn{4}{l}{\textit{PRES}} &       &       &       &       &  \\
    $\epsilon=10^{-2}$ & 0.96 & 0.955 & 0.935 & 0.945 &       & 0.95 & 0.94 & 0.92 & 0.935 \\
    $\epsilon=10^{-5}$ & 0.96 & 0.955 & 0.925 & 0.945 &       & 0.955 & 0.94 & 0.925 & 0.935 \\
    $\epsilon=10^{-7}$ & 0.965 & 0.955 & 0.925 & 0.945 &       & 0.955 & 0.94 & 0.925 & 0.935 \\[5pt]
    \multicolumn{4}{l}{\textit{sPRES}} &       &       &       &       &  \\
    $\epsilon=10^{-2}$ & 0.96 & 0.955 & 0.935 & 0.945 &        & 0.95 & 0.94 & 0.92 & 0.935 \\
    $\epsilon=10^{-5}$ & 0.96 & 0.955 & 0.935 & 0.945 &        & 0.955 & 0.94 & 0.925 & 0.935 \\
    $\epsilon=10^{-7}$ & 0.96 & 0.955 & 0.93 & 0.945 &        & 0.955 & 0.94 & 0.925 & 0.935 \\
    \hline
    \end{tabular}%
  \end{threeparttable}
  \label{tab:1.fixed}%
\end{table}

\begin{table}
\centering
  \caption{Simulation results on the coverage of confidence intervals for $\beta=0.5, \lambda=Cn^{1/2}$}
  \begin{threeparttable}
    \begin{tabular}{cccccccccc}
    \hline
    \multirow{2}[4]{*}{Coverage} & \multicolumn{4}{c}{$n = 200$} &       & \multicolumn{4}{c}{$n = 400$} \\
\cmidrule{2-5} \cmidrule{7-10}
 & $\beta_1$ & $\beta_2$ & $\beta_9$ & $\beta_{10}$ &       & $\beta_1$ & $\beta_2$ & $\beta_9$ & $\beta_{10}$ \\ \hline
    \multicolumn{4}{l}{\textit{PRES+LS}}      &       &       &       &       &  \\
    $\epsilon=10^{-2}$ & 0.955 & 0.965 & 0.975 & 0.98 &       & 0.975 & 0.955 & 0.94 & 0.97 \\
    $\epsilon=10^{-5}$ & 0.955 & 0.965 & 0.975 & 0.98 &       & 0.975 & 0.955 & 0.93 & 0.97 \\
    $\epsilon=10^{-7}$ & 0.955 & 0.965 & 0.975 & 0.98 &       & 0.97 & 0.955 & 0.93 & 0.97 \\[5pt]
    \multicolumn{4}{l}{\textit{sPRES+LS}} &       &       &       &       &  \\
    $\epsilon=10^{-2}$ & 0.955 & 0.965 & 0.975 & 0.98 &       & 0.975 & 0.955 & 0.94 & 0.97 \\
    $\epsilon=10^{-5}$ & 0.96 & 0.965 & 0.975 & 0.98 &      & 0.98 & 0.955 & 0.945 & 0.97 \\
    $\epsilon=10^{-7}$ & 0.96 & 0.97 & 0.975 & 0.98 &       & 0.98 & 0.955 & 0.945 & 0.97 \\[5pt]
    \multicolumn{4}{l}{\textit{PRES}} &       &       &       &       &  \\
    $\epsilon=10^{-2}$ & 0.945 & 0.94 & 0.96 & 0.965 &       & 0.94 & 0.96 & 0.93 & 0.965 \\
    $\epsilon=10^{-5}$ & 0.935 & 0.925 & 0.965 & 0.96 &      & 0.94 & 0.95 & 0.93 & 0.965 \\
    $\epsilon=10^{-7}$ & 0.935 & 0.925 & 0.96 & 0.965 &       & 0.93 & 0.95 & 0.93 & 0.965 \\[5pt]
    \multicolumn{4}{l}{\textit{sPRES}} &       &       &       &       &  \\
    $\epsilon=10^{-2}$ & 0.945 & 0.94 & 0.96 & 0.965  &        & 0.94 & 0.96 & 0.93 & 0.965 \\
    $\epsilon=10^{-5}$ & 0.945 & 0.945 & 0.96 & 0.97  &       & 0.945 & 0.96 & 0.935 & 0.97 \\
    $\epsilon=10^{-7}$ & 0.945 & 0.945 & 0.965 & 0.96  &        & 0.945 & 0.96 & 0.935 & 0.965 \\
    \hline
    \end{tabular}%
  \end{threeparttable}
  \label{tab:05.fixed}%
\end{table}


\begin{figure}[htb]
\centering
  \subfloat[$n=200$]{%
    \includegraphics[scale = 0.8]{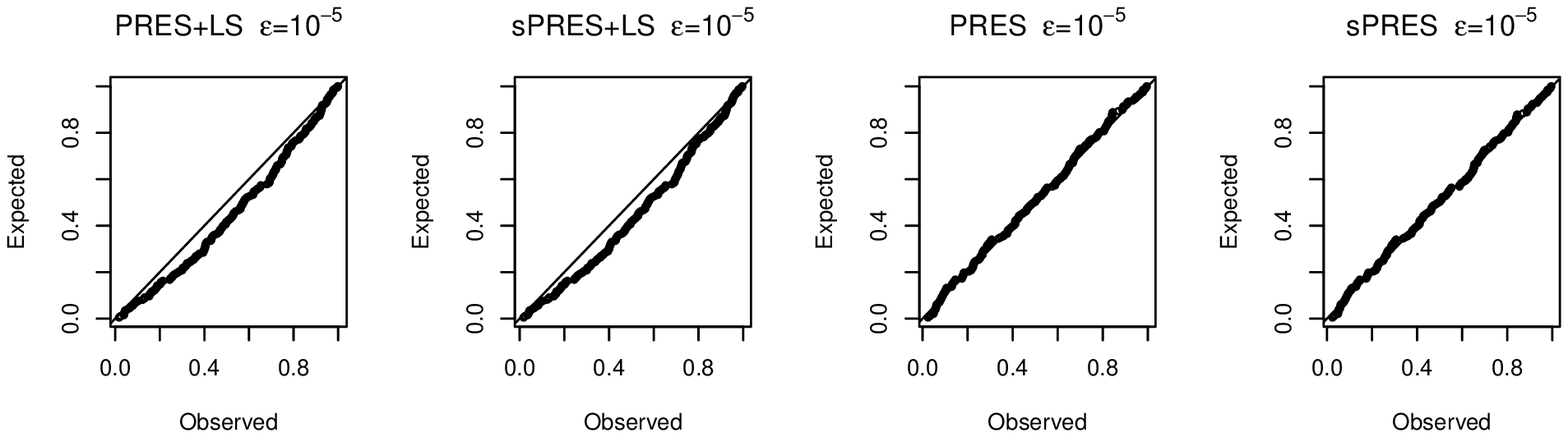}}\\
  \subfloat[$n=400$]{%
    \includegraphics[scale = 0.8]{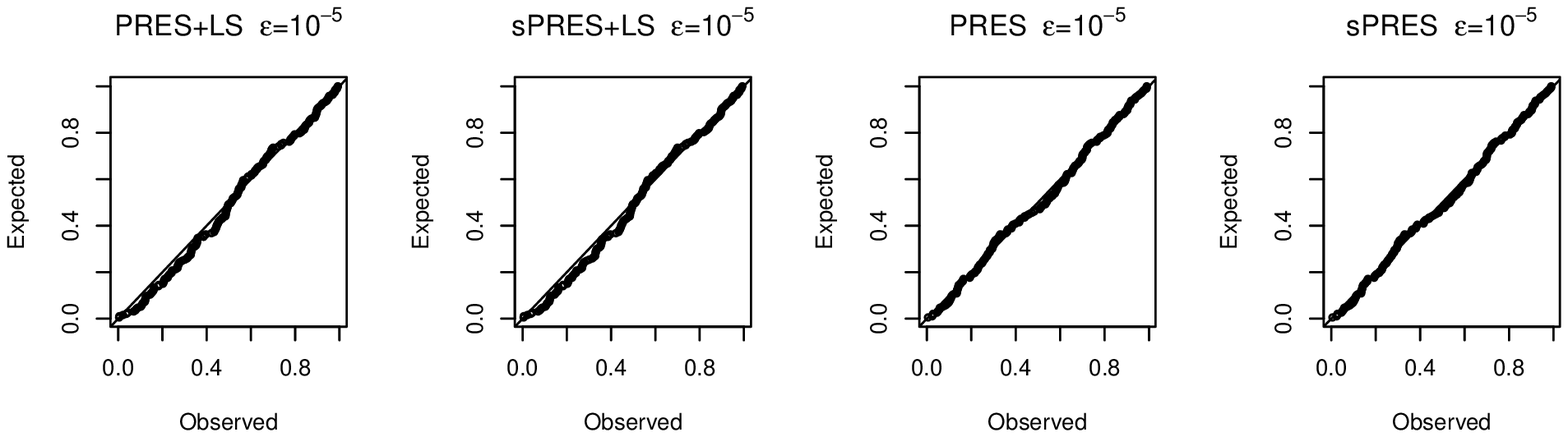}}
  \caption{Uniform Q-Q plots of the null $p$-values for $\lambda=Cn^{1/2}$ and $\epsilon=10^{-5}$ under $\beta=1$.}\label{fig:1.5.fixed}
\end{figure}

\begin{figure}[htb]
\centering
  \subfloat[$n=200$]{%
    \includegraphics[scale = 0.8]{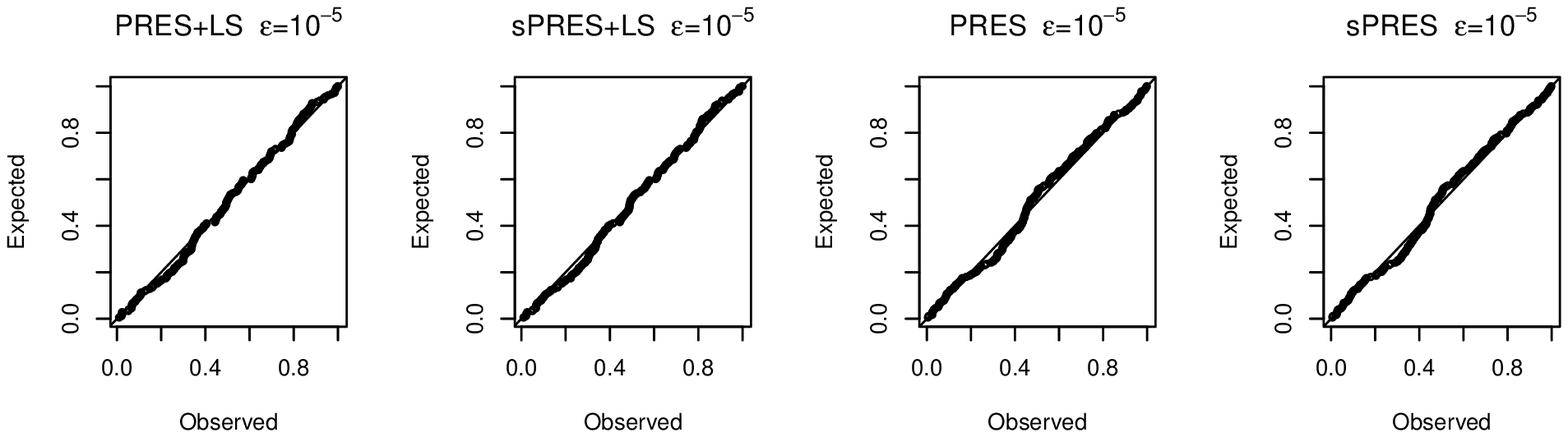}}\\
  \subfloat[$n=400$]{%
    \includegraphics[scale = 0.8]{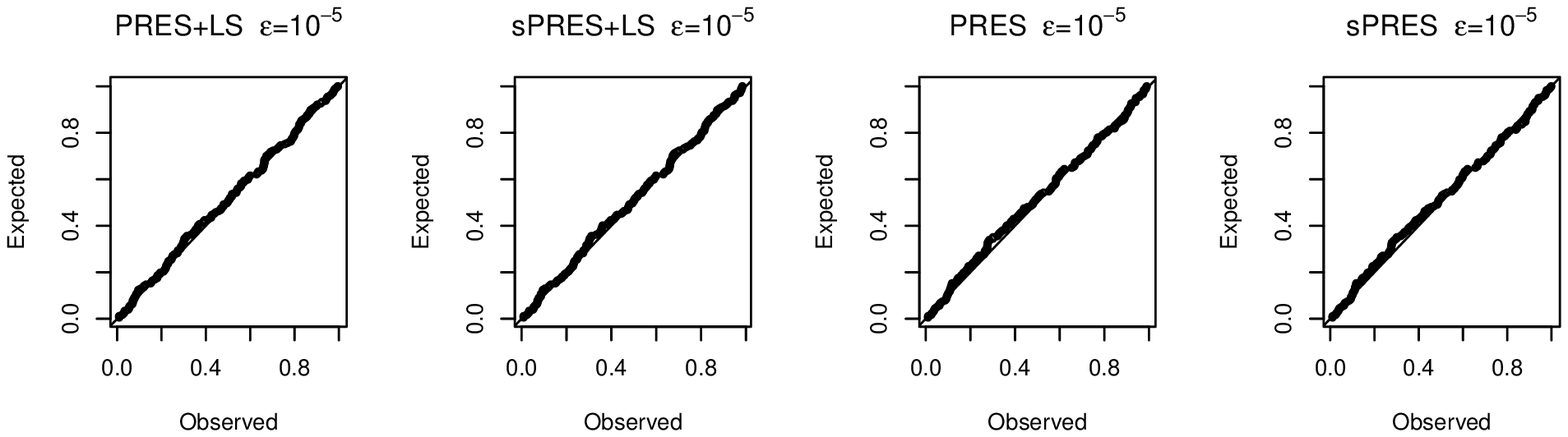}}
  \caption{Uniform Q-Q plots of the null $p$-values for $\lambda=Cn^{1/2}$ and $\epsilon=10^{-5}$ under $\beta=0.5$.}\label{fig:05.5.fixed}
\end{figure}

\section{Application}
\label{section.real.data}
To illustrate the utility of our method, we apply our method to the ADNI data mentioned in Section \ref{section.intro.intro}. ADNI is a multi-phase (ADNI1, ADNI GO, ADNI2, and ADNI3) longitudinal study with the primary goal of testing whether neuroimaging, biochemical, and genetic biomarkers, as well as clinical and
neuropsychological assessment, can be combined to measure the progression of mild
cognitive impairment (MCI) and early Alzheimer’s disease.
See \url{www.adni-info.org} for up-to-date information. Similar to \citet{li2017prediction,li2020penalized}, we aim to identify significant predictors of time from baseline visit to AD onset in ADNI1 participants who had MCI at study entry. Due to the periodic assessment of the participants' cognition status in the ADNI study, the response is interval-censored. Some ADNI subjects were diagnosed as AD at a study exam but classified as MCI at a later one. Considering that the AD development is an irreversible process, we regard all the diagnosis results of AD before a MCI as misclassifications, and correct them to MCI. Therefore, our analytic cohort consists of the 381 ADNI1 participants who were diagnosed as MCI in at least one study exam and had a baseline stage of MCI or AD.

We choose the covariates for the survival regression from the 24 covariates considered in \citet{li2020penalized}.
Different from \citet{li2020penalized}, we do not consider Alzheimer\textquoteright s Disease Assessment Scale scores of 11 items (ADAS11), delayed word recall score in ADAS (ADASQ4), the total number of words that were forgotten in the RAVLT delayed memory test (RAVLT.f) as well as whole brain volume (WholeBrain) in our analysis, because all of them have strong correlations (larger than 0.7) with at least one of the other covariates and none of these four covariates were selected by any variable selection method in \citet{li2020penalized}. Thus, we consider 20 covariates in the survival regression. They are gender (1 for male and 0 for female), marital status (MaritalStatus, 1 for married and 0 otherwise), baseline age, years of receiving education (PTEDUCAT), apolipoprotein E $\epsilon$4 (APOE$\epsilon$4) genotype, Alzheimer\textquoteright s Disease Assessment Scale scores of 13 items (ADAS13), Clinical Dementia Rating Scale Sum of Boxes score (CDRSB), mini-mental state examination score (MMSE), Rey auditory verbal learning test score of immediate recall (RAVLT.i), learning ability (RAVLT.l), the percentage of words that were forgotten in the RAVLT delayed memory test (RAVLT.perc.f), digit symbol substitution test score (DIGITSCOR), trails B score (TRABSCOR), functional assessment questionnaire score (FAQ), and volumes of ventricles, hippocampus, entorhinal, fusiform gyrus (Fusiform), and middle temporal gyrus (MidTemp), as well as intracerebral volume (ICV). We then remove the subjects with missing values in the 20 covariates, resulting in a sample size of 300. To apply our method, we use the values of the covariates at the baseline examination.

The results of our method using sPRES and $\epsilon=10^{-5}$ are summarized in Table \ref{tab:adni.sPRES+LS.1e-5.aic}. We do not choose PRES+LS or sPRES+LS here since their null $p$
-values in simulations are excessively conservative when the tuning parameter $\lambda$ is chosen via AIC. By contrast, PRES and sPRES display a minor tendency toward conservatism and more power under these circumstances (See Figure \ref{fig:1.2.aic}--\ref{fig:05.7.aic} in Section \ref{section.additional.sim}). Replacing sPRES with PRES or using the other increments investigated in this paper has limited impact on the results, and does not alter the significance of covariates. MaritalStatus, PTEDUCAT and TRABSCOR were not selected by lasso with the penalty parameter chosen based on AIC. So these covariates are not included in the further inference. However, only APOE$\epsilon$4, RAVLT.i, FAQ and MidTemp are significantly associated with the time to AD at the 0.05 level according to the post-lasso inference, which have also been consistently identified as significant predictors for conversion time from MCI to AD in previous studies utilizing the Cox model \citep{li2020penalized, du2022variable}. Figures \ref{fig:adni.p.value} and \ref{fig:adni.CI} compare respectively the $p$-values and confidence intervals between our method and the naive method, i.e., fitting a Cox model with the selected covariates. Same as the data analysis result in \citet{taylor2018post}, the naive method leads to a much smaller $p$-value most of the times, and the post-selection inference method results in wider confidence intervals.

\begin{table}

  \centering
  \caption{Results of the analysis of the ADNI dataset, where sPRES and $\epsilon=10^{-5}$ are used.}
  \begin{threeparttable}
    {\begin{tabular}{cccccc}
   \hline

    Covariate & Lasso & One-step & Low & Upp & $p$-value\\ \hline
    %
    Gender & -0.182733 & -0.183227 & -0.677409 & 1.229962 & 0.964657 \\
    MaritalStatus & - & - & - & - & -\\
    Age & -0.036592 & -0.036614 & -0.070081 & 0.041117 & 0.329114 \\
    PTEDUCAT & - & - & - & - & -\\
    APOE$\epsilon$4 & 0.341211 & 0.341332 & 0.053582 & 0.605883 & 0.022514 \\
    ADAS13 & 0.030229 & 0.030222 & -0.033369 & 0.079605 & 0.315359 \\
    CDRSB & 0.057736 & 0.057773 & -1.139812 & 0.409555 & 0.692216 \\
    MMSE & -0.065720 & -0.065763 & -0.189209 & 0.145961 & 0.528018 \\
    RAVLT.i & -0.061683 & -0.061722 & -0.094256 & -0.026042 & 0.001870 \\
    RAVLT.l & 0.066113 & 0.066246 & -0.107643 & 0.170389 & 0.423770 \\
    RAVLT.perc.f & 0.004058 & 0.004061 & -0.009655 & 0.011334 & 0.555652 \\
    DIGITSCOR & -0.007800 & -0.007803 & -0.055722 & 0.044346 & 0.759399 \\
    TRABSCOR & - & - & - & - & -\\
    FAQ & 0.056492 & 0.056502 & 0.005882 & 0.173402 & 0.031311 \\
    Ventricles &0.000001 & 0.000001 & -0.000114 & 0.000013 & 0.329765 \\
    Hippocampus & -0.000132 & -0.000132 & -0.000852 & 0.000425 & 0.542652 \\
    Entorhinal  & -0.000207 & -0.000207 & -0.000573 & 0.000403 & 0.482845 \\
    Fusiform & -0.000027 & -0.000027 & -0.000352 & 0.000474 & 0.929476 \\
    MidTemp   & -0.000157 & -0.000157 & -0.000300 & -0.000063 & 0.001516 \\
    ICV & 0.000002 & 0.000002 & -0.000001 & 0.000016 & 0.157096 \\ [5pt]
   \hline
    \end{tabular}}%
    \begin{tablenotes}%
    \item Lasso, the lasso estimate of the regression coefficient; One-step, the coefficient estimate from the proposed one-step estimator; Low, the lower endpoint of the post-lasso confidence interval; Upp, the upper endpoint of the post-lasso confidence interval; $p$-value, the two-sided $p$-value for testing $H_0: \beta_j^{\ast}=0$ $(j\in\hat{M})$ based on \eqref{pivot}.
  \end{tablenotes}%
  \end{threeparttable}
\label{tab:adni.sPRES+LS.1e-5.aic}
\end{table}

\begin{figure}
\centering
\includegraphics[scale=0.9]{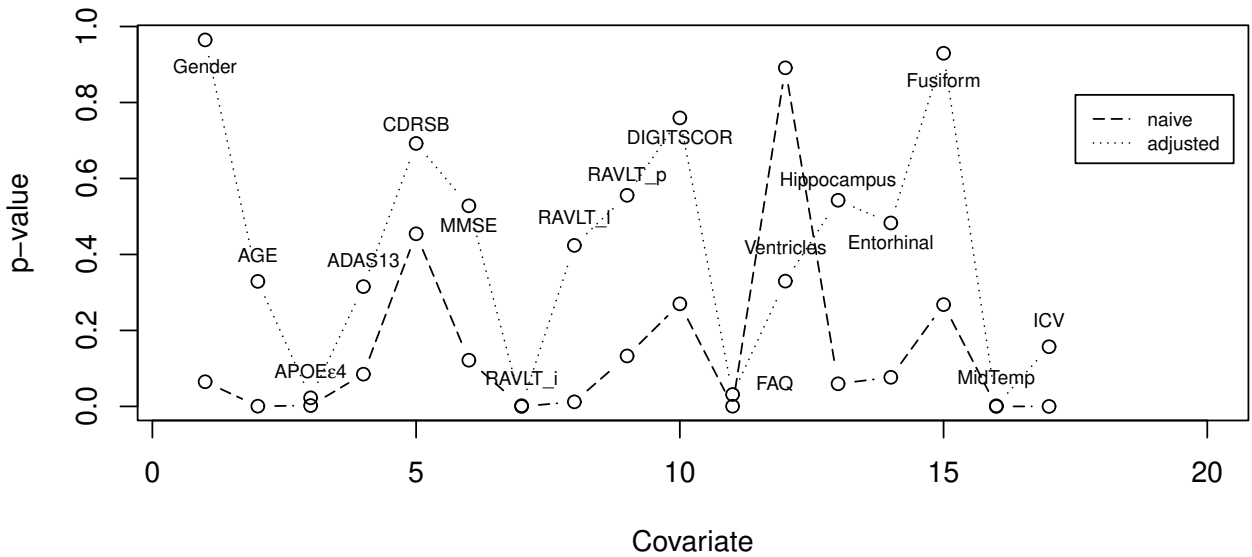}
\caption{Comparison of the $p$-values between the naive method and the method adjusting for the variable selection process in the analysis of the ADNI dataset.} \label{fig:adni.p.value}
\end{figure}

\begin{figure}
\centering
\includegraphics[scale=0.9]{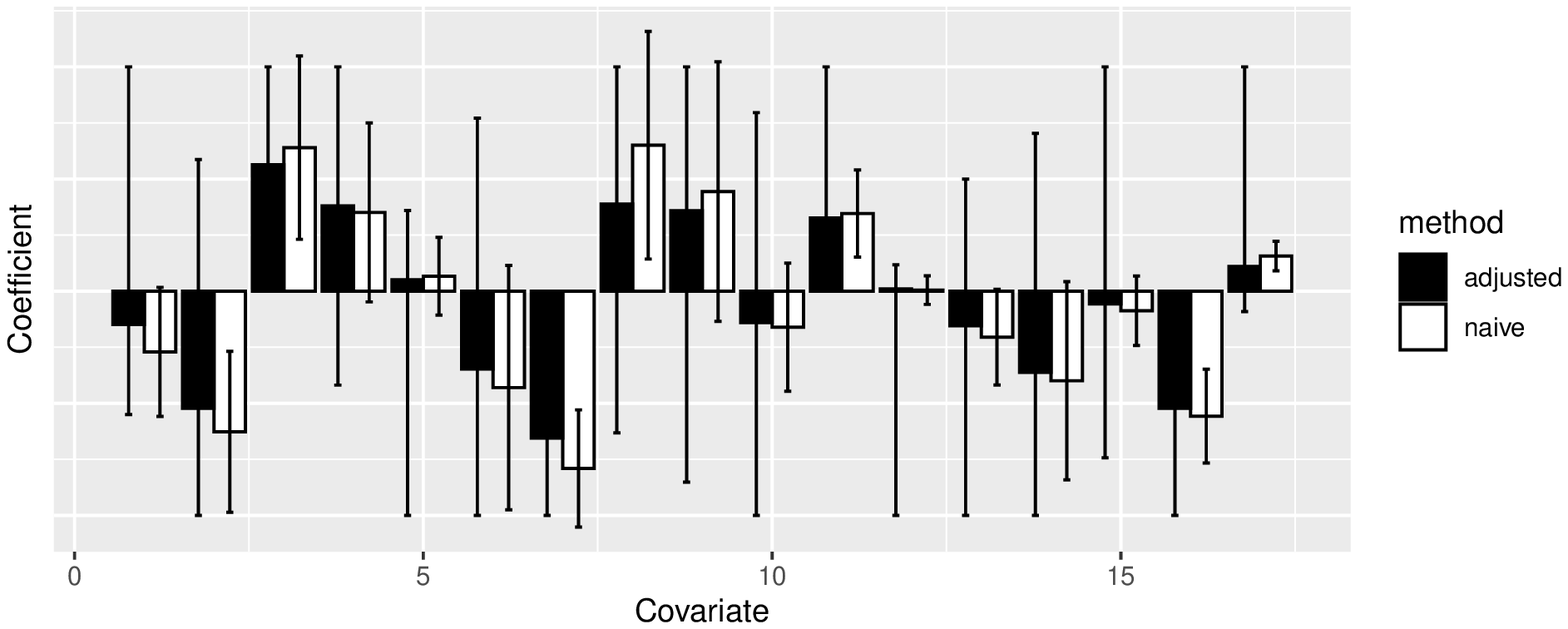}
\caption{Comparison of the confidence intervals between the naive method and the method adjusting for the variable selection process in the analysis of the ADNI dataset. The vertical lines represent the confidence intervals, and the bars represent the coefficient estimates, both of which are scaled appropriately for each covariate so that all the within-covariate differences are apparent in the same figure. } \label{fig:adni.CI}
\end{figure}

\section{Discussion}

We extend \citeauthor{taylor2018post}'s \citeyearpar{taylor2018post} post-selection conditional inference method to the Cox model with interval-censored data. A major contribution of this article is that it provides a rigorous and general framework to establish asymptotic theory of post-selection conditional inference, without assuming that the selected model includes all the signal covariates. Besides the Cox model with interval-censored data, we establish the asymptotic results for generalized linear models and the Cox model with right-censored data using the framework. We believe that our framework can be applied to many other parametric and semiparametric models, e.g., \cite{zeng2016maximum,zeng2017maximum,mao2017semiparametric,gao2019semiparametric,zeng2021maximum,gu2022maximum}, as long as the semiparametric efficiency of the one-step estimator, namely, a similar condition to Condition \ref{condition.one.step.semi}, can be verified. The verification would be similar to establishing the semiparametric efficiency of the nonparametric maximum likelihood estimator, as revealed in the proof of Lemma \ref{lemma.one.step.interval}.

This article considers post-selection conditional inference for the Cox model with interval-censored data under the traditional ``fixed $p$ large $n$'' setup. High-dimensional selective inference is a topic requiring more efforts. \citet{tian2017asymptotics} proposed a post-selection inference procedure similar to \citeauthor{lee2016exact}'s \citeyearpar{lee2016exact} for high-dimensional linear models and established the associated asymptotic results. They also developed a test for the global null hypothesis in high-dimensional generalized linear models. Nevertheless, selective inference in high-dimensional survival models has not been explored. It is noteworthy that our work is based on Le Cam's third lemma, which is restricted to the low-dimensional setting. Even
though \citet{jankova2018semiparametric} considered extending Le Cam's lemmas to high-dimensional models, it is still highly nontrivial to adapt their work to the Cox model and to construct a semiparametric efficient one-step estimator.

Unconditional inference on high-dimensional survival models is also a more challenging problem. Much work has been done on the models with right-censored data. For instance, \citet{zhong2015tests} investigated inference on the high-dimensional additive hazards model.
\citet{fang2017testing} studied inference on the high-dimensional Cox model. \citet{ning2017general} provided a general framework for inference on high dimensional models, which includes the additive hazards model as an illustrative example. \citet{chai2019inference} filled the gap for the accelerated failure time model. However, inference on high-dimensional survival models with interval-censored data has not been studied. It is worthwhile to work in this direction as the research results will facilitate genetic association studies of chronic diseases such as AD, diabetes and dental caries.

\section*{Declaration of the use of generative AI and AI-assisted technologies}

During the preparation of this work the authors used ChatGPT 3.5 in order to polish the manuscript for linguistic expression. After using this service the authors reviewed and edited the content as necessary and take full responsibility for the content of the publication.

\section*{Acknowledgements}

Jianrui Zhang and Dr. Li were supported in part by National Institutes of Health (R56AG075803). Dr. Weng was partially supported by the National Science Foundation (DMS-1915099). The data used in preparation of this article were obtained from the Alzheimer’s Disease Neuroimaging Initiative (ADNI) database (\url{adni.loni.ucla.edu}). As such, the investigators within the ADNI contributed to the design and implementation of ADNI and/or provided data but did not participate in the analysis or writing of this report. A complete listing of ADNI investigators can be found at \url{https://adni.loni.usc.edu/wp-content/uploads/how_to_apply/ADNI_Acknowledgement_List.pdf}.


\FloatBarrier

\section{Appendix}
\subsection{The order of the proofs}
It is easier to prove the asymptotic result for GLMs than for the general semiparametric model considered in Theorem \ref{theorem.semi}, and the proof for the latter can be built on that for the former. Therefore, we first prove the result for GLMs (Theorem \ref{theorem.GLM}), followed by the proof for the general semiparametric model (Theorem \ref{theorem.semi}). The proof of the asymptotic result for the Cox model with interval-censored data (Theorem \ref{theorem.interval}) builds upon Theorem \ref{theorem.semi} and thus is presented after it, along with the proof for the consistency of the proposed information estimators under interval censoring (Theorems \ref{theorem.PRES} and \ref{theorem.LS}). Then we give the proof for the Cox model with right censored data (Theorem \ref{theorem.right}), by adapting the proof of Theorem \ref{theorem.GLM}. Finally, We present two theorems regarding $\widetilde{\theta}^{M}=\theta_{M,n}^{M}+o(1)$.

\subsection{Proof of Theorem \ref{theorem.GLM}}
We first provide and prove the formal statement of \eqref{cond_conv} mentioned in Section \ref{section.idea}. Then we divide the proof of Theorem \ref{theorem.GLM} into several parts: deriving the asymptotic distributions of the one-step estimator and the score function under the local alternative, tackling the KKT conditions for the lasso, writing the active and inactive constraints in the KKT conditions into desired forms respectively, deducing the asymptotic independence between the active and the inactive constraints, and concluding the proof.

\begin{lemma}
\label{lemma.idea}
Suppose we have two sequences of random vectors $Z_{n}$ and $W_{n}$ such that $(Z_{n},W_{n})\stackrel{d}{\longrightsquigarrow}(Z,W)$. If $(Z,W)$ has a continuous joint distribution, $Z$ and $W$ are independent, $R_{n,1}=o_{P}(1)$, $R_{n,2}=o_{P}(1)$, and
\[
\liminf_{n\to\infty}\text{P}\left(Z_{n}+R_{n,1}\in B,W_{n}+R_{n,2}\in D\right)>0,
\]
then
\[
Z_{n}\mid \left\{ Z_{n}+R_{n,1}\in B,W_{n}+R_{n,2}\in D\right\} \stackrel{d}{\longrightsquigarrow}Z\mid \left\{ Z\in B\right\} .
\]
\end{lemma}

\begin{proof-of-lemma}[\ref{lemma.idea}]
By Slutsky's lemma, $(Z_{n}+R_{n,1},W_{n}+R_{n,2})\stackrel{d}{\longrightsquigarrow}(Z,W)$.
Since $(Z,W)$ has a continuous joint distribution,
\[
\text{P}\left(Z\in B,W\in D\right)=\lim_{n\to\infty}\text{P}\left(Z_{n}+R_{n,1}\in B,W_{n}+R_{n,2}\in D\right)>0,
\]
and $Z$ has a continuous marginal distribution. Due to the independence of $Z$ and $W$, it suffices to show that for all $t$,
\[
\text{P}\left(Z_{n}\le t,Z_{n}+R_{n,1}\in B,W_{n}+R_{n,2}\in D\right)\longrightarrow \text{P}\left(Z\le t,Z\in B,W\in D\right)
\]
and
\[
\text{P}\left(Z_{n}+R_{n,1}\in B,W_{n}+R_{n,2}\in D\right)\longrightarrow \text{P}\left(Z\in B,W\in D\right).
\]
The latter one is obvious. For the former, note that
\[
\text{P}\left(Z_{n}+R_{n,1}\le t,Z_{n}+R_{n,1}\in B,W_{n}+R_{n,2}\in D\right)\longrightarrow \text{P}\left(Z\le t,Z\in B,W\in D\right),
\]
and thus it suffices to show that
\[
\text{P}\left(Z_{n}\le t,Z_{n}+R_{n,1}>t\right)\longrightarrow0\quad \mbox{and}\quad \text{P}\left(Z_{n}>t,Z_{n}+R_{n,1}\le t\right)\longrightarrow0.
\]
Because for all $\delta>0$,
\begin{align*}
\text{P}\left(Z_{n}\le t,Z_{n}+R_{n,1}>t\right) & =\text{P}\left(t-R_{n,1}<Z_{n}\le t\right)=\text{P}\left(t-R_{n,1}<Z_{n}\le t,R_{n,1}>\delta\right)+\text{P}\left(t-R_{n,1}<Z_{n}\le t,R_{n,1}\le\delta\right)\\
 & \le\text{P}\left(R_{n,1}>\delta\right)+\text{P}\left(t-\delta\le Z_{n}\le t\right)\longrightarrow\text{P}\left(t-\delta\le Z\le t\right),
\end{align*}
letting $\delta\to0$, we obtain
\[
\text{P}\left(Z_{n}\le t,Z_{n}+R_{n,1}>t\right)\longrightarrow0.
\]
By similar arguments, $\text{P}\left(Z_{n}>t,Z_{n}+R_{n,1}\le t\right)\longrightarrow0$.
\end{proof-of-lemma}

\subsubsection{The one-step estimator and the score function under the local alternative}
\label{section.one.step.local}
We are going to show asymptotic results analogous to \eqref{AN_onestepest_local} and \eqref{AN_score_local} in the case of GLM. By Condition \ref{condition.one.step.glm} and the differentiability of the GLM in quadratic mean at $\beta^{\ast}$, we have
\begin{equation}\label{condition.LeCam3}
\begin{pmatrix}n^{1/2}\left(\overline{\beta}^{M}-\beta_M^{\ast}\right)\\
n^{1/2}\mathcal{I}_{M,M}^{-1}\mathbb{P}_nl_{\beta_{M}}(\beta^{\ast})\\
\log\frac{dP_{\beta_{n}^{\ast}}^{n}}{dP_{\beta^{\ast}}^{n}}
\end{pmatrix}\underset{\beta^{\ast}}{\overset{d}{\longrightsquigarrow}}N\left(\begin{pmatrix}0\\
0\\
-\frac{1}{2}\theta^{\ast}{}^{\top}\mathcal{I}\theta^{\ast}
\end{pmatrix},\begin{pmatrix}\mathcal{I}_{M,M}^{-1} & \mathcal{I}_{M,M}^{-1} & \widetilde{\theta}^{M}\\
\mathcal{I}_{M,M}^{-1} & \mathcal{I}_{M,M}^{-1} & \widetilde{\theta}^{M}\\
\widetilde{\theta}^{M\top} & \widetilde{\theta}^{M\top} & \theta^{\ast}{}^{\top}\mathcal{I}\theta^{\ast}
\end{pmatrix}\right).
\end{equation}
By Le Cam's third lemma,
\begin{equation}\label{distr_of_1step_and_score}
\begin{pmatrix}n^{1/2}\left(\overline{\beta}^{M}-\beta_M^{\ast}\right)\\
n^{1/2}\mathcal{I}_{M,M}^{-1}\mathbb{P}_nl_{\beta_{M}}(\beta^{\ast})
\end{pmatrix}\underset{\beta_{n}^{\ast}}{\overset{d}{\longrightsquigarrow}}N\left(\begin{pmatrix}\widetilde{\theta}^{M}\\
\widetilde{\theta}^{M}
\end{pmatrix},\begin{pmatrix}\mathcal{I}_{M,M}^{-1} & \mathcal{I}_{M,M}^{-1}\\
\mathcal{I}_{M,M}^{-1} & \mathcal{I}_{M,M}^{-1}
\end{pmatrix}\right),
\end{equation}
which implies that under the local alternative,
\[
n^{1/2}\left(\overline{\beta}^{M}-0\right)=n^{1/2}\mathcal{I}_{M,M}^{-1}\mathbb{P}_nl_{\beta_{M}}(\beta^{\ast})+o_{P}(1).
\]
Therefore,
\[
n^{1/2}\left(\overline{\beta}^{M}-\widetilde{\beta}_{n}^{M}\right)=n^{1/2}\mathcal{I}_{M,M}^{-1}\mathbb{P}_nl_{\beta_{M}}(\beta^{\ast})-\widetilde{\theta}^{M}+o_{P}(1),
\]
under the local alternative.

We now prove that under the local alternative,
\begin{equation}\label{onestep.local.asymp}
n^{1/2}\left(\overline{\beta}^{M}-\widetilde{\beta}_{n}^{M}\right)=n^{1/2}\mathcal{I}_{M,M}^{-1}\mathbb{P}_nl_{\beta_{M}}(\beta_{n}^{\ast})+o_{P}(1).
\end{equation}
It suffices to show that
\[
n^{1/2}\mathbb{P}_n\left\{ l_{\beta_{M}}(\beta^{\ast})-l_{\beta_{M}}(\beta_{n}^{\ast})\right\} =\mathcal{I}_{M,M}\widetilde{\theta}^{M}+o_{P}(1)
\]
under the local alternative. By Le Cam's third lemma, it is sufficient to prove the above under $P_{\beta^{\ast}}$. Because of Condition \ref{condition.Donsker.glm}, we have by Lemma 19.24 of \citet{van2000asymptotic} and the Taylor expansion that,
\begin{align*}
n^{1/2}\mathbb{P}_{n}\left\{ l_{\beta_{M}}(\beta^{\ast})-l_{\beta_{M}}(\beta_{n}^{\ast})\right\}  & =n^{1/2}P\left\{ l_{\beta_{M}}(\beta^{\ast})-l_{\beta_{M}}(\beta_{n}^{\ast})\right\} +o_{P}(1)\\
 & =n^{1/2}Pl_{\beta_{M}\beta_{M}}(\beta^{\ast})\cdot(\beta_{M}^{\ast}-\beta_{M,n}^{\ast})+n^{1/2}Pl_{\beta_{M}\beta_{-M}}(\beta^{\ast})\cdot(\beta_{-M}^{\ast}-\beta_{-M,n}^{\ast})\\
 & \quad\quad+o_{P}\left(\|n^{1/2}(\beta_{n}^{\ast}-\beta^{\ast})\|_{2}+1\right)\\
 & =\mathcal{I}_{M,M}\theta_{M}^{\ast}+\mathcal{I}_{M,-M}\theta_{-M}^{\ast}+o_{P}(1)\\
 & =\mathcal{I}_{M,M}\widetilde{\theta}^{M}+o_{P}(1).
\end{align*}

Similar arguments can show that
\begin{equation}\label{score.local.asymp}
n^{1/2}\mathbb{P}_nl_{\beta}(\beta_{n}^{\ast})\underset{\beta_{n}^{\ast}}{\overset{d}{\longrightsquigarrow}}N\left(0,\mathcal{I}\right).
\end{equation}

\subsubsection{Tackle the KKT conditions}
\label{section.stationarity}
Note that the estimator $\hat{\beta}^{M}_{\textrm{full}}$ is equal to the lasso estimator $\hat{\beta}$ conditional on the selection event $\left\{ \hat{M}=M,\hat{s}_{\hat{M}}=s_{M}\right\}$. Thus, by similar arguments as in the proof of Lemma 4.1 in \cite{lee2016exact} and as in \cite{taylor2018post}, the event $\left\{ \hat{M}=M,\hat{s}_{\hat{M}}=s_{M}\right\}$ occurs
if and only if $\text{sign}(\hat{\beta}^{M}) =s_{M},$ and
\begin{equation}\label{inactive.constraint}
\bigg\|\frac{\partial l_n}{\partial\beta_{-M}}\bigg|_{\beta=\hat{\beta}^{M}_{\textrm{full}}}\bigg\|_{\infty}  \le\lambda.
\end{equation}
Recall that the one-step estimator is
\begin{align*}
\overline{\beta}^{M} & =\hat{\beta}^{M}+\lambda\hat{I}_{M,n}^{-1}s_{M},
\end{align*}
Writing the active constraint $\text{sign}(\hat{\beta}^{M}) =s_{M}$ as $\text{diag}(s_{M})\hat{\beta}^{M}>0$ and plugging in the one-step estimator, the constraint becomes
\[
\text{diag}(s_{M})\left(\overline{\beta}^{M}-\lambda\hat{I}_{M,n}^{-1}s_{M}\right)>0,
\]
namely,
\begin{equation}\label{active.constraint}
-\text{diag}(s_{M})\overline{\beta}^{M}<-\lambda\text{diag}(s_{M})\hat{I}_{M,n}^{-1}s_{M}.
\end{equation}

\subsubsection{Active constraint}
\label{section.active}
Multiplying both sides of the active constraint \eqref{active.constraint} by $n^{1/2}$, we obtain
\[
-n^{1/2}\text{diag}(s_{M})\overline{\beta}^{M}\le-\lambda n^{1/2}\text{diag}(s_{M})\hat{I}_{M,n}^{-1}s_{M}.
\]
This can be equivalently written as
\[
-n^{1/2}\text{diag}(s_{M})\overline{\beta}^{M}\le-\lambda n^{-1/2}\text{diag}(s_{M})\mathcal{I}_{M,M}^{-1}s_{M}+R_{M,1},
\]
where
\[
R_{M,1}=\lambda n^{-1/2}\text{diag}(s_{M})\left(\mathcal{I}_{M,M}^{-1}-n\hat{I}_{M,n}^{-1}\right)s_{M}.
\]
Due to Conditions \ref{condition.lasso.semi} and \ref{condition.Hessian}, it is plain that $R_{M,1}=o_{P}(1)$. Let $b=-C\text{diag}(s_{M})\mathcal{I}_{M,M}^{-1}s_{M}$. Since $\lambda n^{-1/2} = C+o(1)$, we may write the active constraint as
\[
-n^{1/2}\text{diag}(s_{M})\left(\overline{\beta}^{M}-\widetilde{\beta}_{n}^{M}\right)\le b+\text{diag}(s_{M})\widetilde{\theta}^{M}+o_{P}(1),
\]
where we have used $\widetilde{\theta}^{M}=n^{1/2}\widetilde{\beta}_{n}^{M}$. Note that $n^{1/2}\left(\overline{\beta}^{M}-\widetilde{\beta}_{n}^{M}\right)$ is asymptotically normal, while $s_{M}$, $\widetilde{\theta}^{M}$ and $b$ are fixed. So the active constraint can be written into the form of $Z_{n}+R_{n,1}\in B$ as in Lemma \ref{lemma.idea}, where $Z_{n}=n^{1/2}\left(\overline{\beta}^{M}-\widetilde{\beta}_{n}^{M}\right)$.

\subsubsection{Inactive constraint}
\label{section.inactive}
The inactive constraint \eqref{inactive.constraint} can be written as
\[
\bigg\|n^{1/2}\mathbb{P}_nl_{\beta_{-M}}(\hat{\beta}^{M}_{\textrm{full}})\bigg\|_{\infty}\le n^{-1/2}\lambda.
\]

Because of Condition \ref{condition.Donsker.glm}, we have by Lemma 19.24 of \citet{van2000asymptotic} and the Taylor expansion that,  under $P_{\beta^{\ast}}$,
\begin{align}
n^{1/2}\mathbb{P}_n\left\{ l_{\beta_{-M}}(\beta_{n}^{\ast})-l_{\beta_{-M}}(\beta^{\ast})\right\}  & =n^{1/2}P \left\{ l_{\beta_{-M}}(\beta_{n}^{\ast})-l_{\beta_{-M}}(\beta^{\ast})\right\} +o_{P}(1) \nonumber \\
 & =-n^{1/2}\mathcal{I}_{-M,M}(\beta_{M,n}^{\ast}-\beta_{M}^{\ast})-n^{1/2}\mathcal{I}_{-M,-M}(\beta_{-M,n}^{\ast}-\beta_{-M}^{\ast})+o_{P}\left(\|n^{1/2}(\beta_{n}^{\ast}-\beta^{\ast})\|_{2}+1\right) \nonumber \\
 & =-n^{1/2}\mathcal{I}_{-M,M}(\beta_{M,n}^{\ast}-\beta_{M}^{\ast})-n^{1/2}\mathcal{I}_{-M,-M}(\beta_{-M,n}^{\ast}-\beta_{-M}^{\ast})+o_{P}(1). \label{intermediate_result1}
\end{align}

Combining \eqref{intermediate_result1} and \eqref{condition6-1_glm} in Condition \ref{condition.one.step.glm}, we have
\[
n^{1/2}\mathbb{P}_n\left\{ l_{\beta_{-M}}(\hat{\beta}^{M})-l_{\beta_{-M}}(\beta_{n}^{\ast})\right\} =-n^{1/2}\mathcal{I}_{-M,M}(\hat{\beta}^{M}-\beta_{M,n}^{\ast})-n^{1/2}\mathcal{I}_{-M,-M}(0-\beta_{-M,n}^{\ast})+o_{P}(1),
\]
which also holds under the local alternative.

By Conditions \ref{condition.lasso.semi} and \ref{condition.Hessian},
\begin{align*}
-n^{1/2}\mathcal{I}_{-M,M}(\hat{\beta}^{M}-\beta_{M,n}^{\ast}) & -n^{1/2}\mathcal{I}_{-M,-M}(0-\beta_{-M,n}^{\ast})\\
 & =-n^{1/2}\mathcal{I}_{-M,M}(\hat{\beta}^{M}-\overline{\beta}^{M})-n^{1/2}\mathcal{I}_{-M,M}(\overline{\beta}^{M}-\widetilde{\beta}_{n}^{M})\\
 & \quad-n^{1/2}(\mathcal{I}_{-M,M}\widetilde{\beta}_{n}^{M}-\mathcal{I}_{-M,M}\beta_{M,n}^{\ast}-\mathcal{I}_{-M,-M}\beta_{-M,n}^{\ast})\\
 & =-n^{1/2}\mathcal{I}_{-M,M}\left(-\lambda\hat{I}_{M,n}^{-1}s_{M}\right)-n^{1/2}\mathcal{I}_{-M,M}(\overline{\beta}^{M}-\widetilde{\beta}_{n}^{M})\\
 & \quad-n^{1/2}(\mathcal{I}_{-M,M}\beta_{M,n}^{\ast}+\mathcal{I}_{-M,M}\mathcal{I}_{M,M}^{-1}\mathcal{I}_{M,-M}\beta_{-M,n}^{\ast}-\mathcal{I}_{-M,M}\beta_{M,n}^{\ast}-\mathcal{I}_{-M,-M}\beta_{-M,n}^{\ast})\\
 & =(\lambda n^{-1/2})\cdot\mathcal{I}_{-M,M}\left(\hat{I}_{M,n}/n\right)^{-1}s_{M}-n^{1/2}\mathcal{I}_{-M,M}(\overline{\beta}^{M}-\widetilde{\beta}_{n}^{M})\\
 & \quad-(\mathcal{I}_{-M,M}\mathcal{I}_{M,M}^{-1}\mathcal{I}_{M,-M}-\mathcal{I}_{-M,-M})\theta_{-M,n}^{\ast}\\
 & =C\mathcal{I}_{-M,M}\mathcal{I}_{M,M}^{-1}s_{M}-n^{1/2}\mathcal{I}_{-M,M}(\overline{\beta}^{M}-\widetilde{\beta}_{n}^{M})\\
 & \quad-(\mathcal{I}_{-M,M}\mathcal{I}_{M,M}^{-1}\mathcal{I}_{M,-M}-\mathcal{I}_{-M,-M})\theta_{-M}^{\ast}+o_{P}(1).
\end{align*}
Thus,
\begin{align*}
n^{1/2}\mathbb{P}_{n}l_{\beta_{-M}}(\hat{\beta}^{M}_{\textrm{full}}) & =n^{1/2}\mathbb{P}_{n}l_{\beta_{-M}}(\beta_{n}^{\ast})-n^{1/2}\mathcal{I}_{-M,M}(\overline{\beta}^{M}-\widetilde{\beta}_{n}^{M})+C\mathcal{I}_{-M,M}\mathcal{I}_{M,M}^{-1}s_{M}\\
 & \quad-(\mathcal{I}_{-M,M}\mathcal{I}_{M,M}^{-1}\mathcal{I}_{M,-M}-\mathcal{I}_{-M,-M})\theta_{-M}^{\ast}+o_{P}(1).
\end{align*}
Then the inactive constraint can be written as
\begin{align*}
\bigg\| n^{1/2}\mathbb{P}_{n}l_{\beta_{-M}}(\beta_{n}^{\ast}) & -n^{1/2}\mathcal{I}_{-M,M}(\overline{\beta}^{M}-\widetilde{\beta}_{n}^{M})+C\mathcal{I}_{-M,M}\mathcal{I}_{M,M}^{-1}s_{M}\\
 & \quad-(\mathcal{I}_{-M,M}\mathcal{I}_{M,M}^{-1}\mathcal{I}_{M,-M}-\mathcal{I}_{-M,-M})\theta_{-M}^{\ast}+o_{P}(1)\bigg\|_{\infty}\le n^{-1/2}\lambda,
\end{align*}
Note that $C\mathcal{I}_{-M,M}\mathcal{I}_{M,M}^{-1}s_{M}$ and $(\mathcal{I}_{-M,M}\mathcal{I}_{M,M}^{-1}\mathcal{I}_{M,-M}-\mathcal{I}_{-M,-M})\theta_{-M}^{\ast}$ are constant vectors, so the inactive constraint takes the form of $W_{n}+R_{n,2}\in D$ as in Lemma  \ref{lemma.idea}, where $W_{n}=n^{1/2}\mathbb{P}_nl_{\beta_{-M}}(\beta_{n}^{\ast})-n^{1/2}\mathcal{I}_{-M,M}(\overline{\beta}^{M}-\widetilde{\beta}_{n}^{M})$.

\subsubsection{Asymptotic independence}\label{section.asympindep}
Now we prove the asymptotic independence between the active constraint and the inactive constraint under the local alternative. By \eqref{onestep.local.asymp} and \eqref{score.local.asymp},
\begin{align*}
 & \begin{pmatrix}n^{1/2}(\overline{\beta}^{M}-\widetilde{\beta}_{n}^{M})\\
n^{1/2}\mathbb{P}_{n}l_{\beta_{-M}}(\beta_{n}^{\ast})-n^{1/2}\mathcal{I}_{-M,M}(\overline{\beta}^{M}-\widetilde{\beta}_{n}^{M})
\end{pmatrix}\\
 & \quad\quad\quad\quad\quad\quad\quad=\begin{pmatrix}\mathcal{I}_{M,M}^{-1} & 0\\
-\mathcal{I}_{-M,M}\mathcal{I}_{M,M}^{-1} & I
\end{pmatrix}\begin{pmatrix}\mathbb{P}_{n}l_{\beta_{M}}(\beta_{n}^{\ast})\\
\mathbb{P}_{n}l_{\beta_{-M}}(\beta_{n}^{\ast})
\end{pmatrix}+o_{P}(1)\\
 & \quad\quad\quad\quad\quad\quad\quad\underset{\beta_{n}^{\ast}}{\overset{d}{\longrightsquigarrow}}N\left(0,\begin{pmatrix}\mathcal{I}_{M,M}^{-1} & 0\\
-\mathcal{I}_{-M,M}\mathcal{I}_{M,M}^{-1} & I
\end{pmatrix}\begin{pmatrix}\mathcal{I}_{M,M} & \mathcal{I}_{M,-M}\\
\mathcal{I}_{M,-M}^{\top} & \mathcal{I}_{-M,-M}
\end{pmatrix}\begin{pmatrix}\mathcal{I}_{M,M}^{-1} & -\mathcal{I}_{M,M}^{-1}\mathcal{I}_{M,-M}\\
0 & I
\end{pmatrix}\right)\\
 & \quad\quad\quad\quad\quad\quad\quad\stackrel{d}{=}N\left(0,\begin{pmatrix}\mathcal{I}_{M,M}^{-1} & 0\\
0 & -\mathcal{I}_{-M,M}\mathcal{I}_{M,M}^{-1}\mathcal{I}_{M,-M}+\mathcal{I}_{-M,-M}
\end{pmatrix}\right).
\end{align*}
\subsubsection{Conclude the proof}
Combining the results of Sections \ref{section.one.step.local}-- \ref{section.asympindep},  we have by Lemma \ref{lemma.idea} that
\[
n^{1/2}(\overline{\beta}^{M}-\widetilde{\beta}_{n}^{M})\mid\left\{ \hat{M}=M,\hat{s}_{\hat{M}}=s_M\right\} \underset{\beta_{n}^{\ast}}{\overset{d}{\longrightsquigarrow}} Z\mid\left\{ -\text{diag}(s_{M})Z\le-C\text{diag}(s_{M})\mathcal{I}_{M,M}^{-1}s_{M}+\text{diag}(s_{M})\widetilde{\theta}^{M}\right\} ,
\]
where $Z\sim N\left(0,\mathcal{I}_{M,M}^{-1}\right)$. Let $\tilde{Z}=Z+n^{1/2}\widetilde{\beta}_{n}^{M}=Z+\widetilde{\theta}^{M}$. Then $\tilde{Z}\sim N\left(\widetilde{\theta}^{M},\mathcal{I}_{M,M}^{-1}\right)$
and the above result implies
\[
n^{1/2}\overline{\beta}^{M}\mid\left\{ \hat{M}=M,\hat{s}_{\hat{M}}=s_{M}\right\} \underset{\beta_{n}^{\ast}}{\overset{d}{\longrightsquigarrow}} \tilde{Z}\mid\left\{ -\text{diag}(s_{M})\tilde{Z}\le-C\text{diag}(s_{M})\mathcal{I}_{M,M}^{-1}s_{M}\right\} .
\]
By Theorem 5.2 of \citet{lee2016exact}, we have
\[
F_{\gamma^{\top}\widetilde{\theta}^{M},\gamma^{\top}\mathcal{I}_{M,M}^{-1}\gamma}^{\mathcal{V}^{-},\mathcal{V}^{+}}\left(\gamma^{\top}\tilde{Z}\right)\mid\left\{ -\text{diag}(s_{M})\tilde{Z}\le-C\text{diag}(s_{M})\mathcal{I}_{M,M}^{-1}s_{M}\right\} \sim Unif(0,1),
\]
where $\mathcal{V}^{-}$ and $\mathcal{V}^{+}$ are computed as in \eqref{V-} and \eqref{V+} with
\[
A=-\text{diag}(s_{M}),\quad b=-(n^{-1/2}\lambda)\text{diag}(s_{M})\mathcal{I}_{M,M}^{-1}s_{M},\quad c=\mathcal{I}_{M,M}^{-1}\gamma(\gamma^{\top}\mathcal{I}_{M,M}^{-1}\gamma)^{-1},\quad \mbox{and} \quad
z=(I_{|M|}-c\gamma^{\top})\overline{\theta}^M.
\]
So
\[
F_{\gamma^{\top}\widetilde{\theta}^{M},\gamma^{\top}\mathcal{I}_{M,M}^{-1}\gamma}^{\mathcal{V}^{-},\mathcal{V}^{+}}\left(\gamma^{\top}\overline{\theta}^{M}\right)\mid\left\{ \hat{M}=M,\hat{s}_{\hat{M}}=s_{M}\right\} \underset{\beta_{n}^{\ast}}{\overset{d}{\longrightsquigarrow}} Unif(0,1).
\]

\subsection{Proof of Theorem \ref{theorem.semi}}
Our proof of Theorem \ref{theorem.GLM} can be adapted to prove Theorem \ref{theorem.semi}. One big change is that $l_{\beta}$ should be replaced with the efficient score function $\tilde{\ell}$. As a result, the condition for Le Cam's third lemma, applied in Section \ref{section.one.step.local}, needs to be re-verified, and the asymptotic continuity \citep[Lemma 19.24,][]{van2000asymptotic} and Taylor expansion arguments, used in Sections \ref{section.one.step.local} and \ref{section.inactive}, must be established afresh under the setting of Theorem \ref{theorem.semi}. It is plain to see that if the result \eqref{distr_of_1step_and_score} of Le Cam's third lemma and the asymptotic continuity and Taylor expansion arguments in Sections \ref{section.one.step.local} and \ref{section.inactive} can be re-established under the considered semiparametric model, which will be done below, our proof of Theorem \ref{theorem.GLM} still works for Theorem \ref{theorem.semi} after $l_{\beta}$ is replaced by $\tilde{l}$.

\subsubsection{Modify the arguments in Section \ref{section.one.step.local}}
\label{section.one.step.local.semi}
To re-verify the condition of Le Cam's third lemma, consider a one-dimensional submodel $t\mapsto P_{\beta^{\ast}+t\theta^{\ast},\Lambda^\ast}$, which becomes  $P_{\beta_{n}^{\ast},\Lambda^\ast}$ when $t=n^{-1/2}$. The score function of this submodel is $\theta^{\ast\top}\ell_{\beta}(\beta^{\ast},\Lambda^{\ast})$. Because the model $P_{\beta,\Lambda^\ast}$ is differentiable in quadratic mean at $\beta^{\ast}$, we have by Lemma 25.14 of \cite{van2000asymptotic} that
\[
\log\frac{dP_{\beta_{n}^{\ast},\Lambda^\ast}^{n}}{dP_{\beta^{\ast},\Lambda^\ast}^{n}}=n^{1/2}\theta^{\ast\top}\mathbb{P}_{n}\ell_{\beta}(\beta^{\ast},\Lambda^{\ast})-\frac{1}{2}\theta^{\ast\top}P\ell_{\beta}(\beta^{\ast},\Lambda^{\ast})^{\otimes2}\theta^{\ast}+o_{P}(1).
\]
Together with Condition \ref{condition.one.step.semi} , it implies that
\[
\begin{pmatrix}n^{1/2}\left(\overline{\beta}^{M}-\beta_{M}^{\ast}\right)\\
n^{1/2}\mathcal{I}_{M,M}^{-1}\mathbb{P}_{n}\tilde{\ell}_{M}(\beta^{\ast})\\
\log\frac{dP_{\beta_{n}^{\ast},\Lambda^\ast}^{n}}{dP_{\beta^{\ast},\Lambda^\ast}^{n}}
\end{pmatrix}\underset{\beta^{\ast}}{\overset{d}{\longrightsquigarrow}}N\left(\begin{pmatrix}0\\
0\\
-\frac{1}{2}\theta^{\ast\top}P\ell_{\beta}(\beta^{\ast},\Lambda^{\ast})^{\otimes2}\theta^{\ast}
\end{pmatrix},\begin{pmatrix}\mathcal{I}_{M,M}^{-1} & \mathcal{I}_{M,M}^{-1} & \widetilde{\theta}^{M}\\
\mathcal{I}_{M,M}^{-1} & \mathcal{I}_{M,M}^{-1} & \widetilde{\theta}^{M}\\
\widetilde{\theta}^{M\top} & \widetilde{\theta}^{M\top} & \theta^{\ast\top}P\ell_{\beta}(\beta^{\ast},\Lambda^{\ast})^{\otimes2}\theta^{\ast}
\end{pmatrix}\right),
\]
 where we used a property of the efficient score function $\tilde{\ell}(\beta^{\ast})$ that
$P\left\{ \ell_{\beta}(\beta^{\ast},\Lambda^{\ast})\tilde{\ell}(\beta^{\ast})^{\top}\right\} =P\left\{ \tilde{\ell}(\beta^{\ast})^{\otimes2}\right\} =\mathcal{I}.$

By Le Cam's third lemma,
\[
\begin{pmatrix}n^{1/2}\left(\overline{\beta}^{M}-\beta_{M}^{\ast}\right)\\
n^{1/2}\mathcal{I}_{M,M}^{-1}\mathbb{P}_{n}\tilde{\ell}_{M}(\beta^{\ast})
\end{pmatrix}\underset{\beta_n^{\ast}}{\overset{d}{\longrightsquigarrow}}N\left(\begin{pmatrix}\widetilde{\theta}^{M}\\
\widetilde{\theta}^{M}
\end{pmatrix},\begin{pmatrix}\mathcal{I}_{M,M}^{-1} & \mathcal{I}_{M,M}^{-1}\\
\mathcal{I}_{M,M}^{-1} & \mathcal{I}_{M,M}^{-1}
\end{pmatrix}\right),
\]
which is the same as the result \eqref{distr_of_1step_and_score} for GLM.

It remains to show that under $P_{\beta^{\ast},\Lambda^\ast}$,
\begin{equation}\label{asympcont_effscore}
    n^{1/2}\mathbb{P}_n\left\{ \tilde{\ell}(\beta^{\ast})-\tilde{\ell}(\beta_{n}^{\ast})\right\} =\mathcal{I}\theta^{\ast}+o_{P}(1),
\end{equation}
as (\ref{asympcont_effscore}) implies
\[
n^{1/2}\mathbb{P}_{n}\left\{ \tilde{\ell}_{M}(\beta^{\ast})-\tilde{\ell}_{M}(\beta_{n}^{\ast})\right\} =\mathcal{I}_{M,M}\theta_{M}^{\ast}+\mathcal{I}_{M,-M}\theta_{-M}^{\ast}+o_{P}(1)=\mathcal{I}_{M,M}\widetilde{\theta}^{M}+o_{P}(1).
\]

Note that the derivatives of  $\ell_{\beta}(\beta^{\ast}+t\theta^{\ast},\Lambda^{\ast})$ and $\ell_{\Lambda}(\beta^{\ast}+t\theta^{\ast},\Lambda^{\ast})(h^{\ast})$ with respect to $t$ are $\ell_{\beta\beta}(\beta^{\ast}+t\theta^{\ast},\Lambda^{\ast})\theta^{\ast}$ and $\ell_{\Lambda\beta}(\beta^{\ast}+t\theta^{\ast},\Lambda^{\ast})(h^{\ast})\theta^{\ast}$ respectively, where $\ell_{\beta\beta}$ and $\ell_{\Lambda\beta}$ are defined as in \citet{zeng2016maximum}, that is, $\ell_{\beta\beta}(\beta,\Lambda)$ is the second derivative of $\ell(\beta,\Lambda)$ with respect to $\beta$, and $\ell_{\Lambda\beta}(\beta,\Lambda)(h)$ is the derivative of $\ell_{\Lambda}(h)$ with respect to $\beta$. Because of Condition \ref{condition.Donsker.semi}, we have by Lemma 19.24 of \citet{van2000asymptotic}  and the Taylor expansion that
\begin{align*}
n^{1/2}\mathbb{P}_n\left\{ \tilde{\ell}(\beta^{\ast})-\tilde{\ell}(\beta_{n}^{\ast})\right\}  & =n^{1/2}\mathbb{P}_n\left\{ \ell_{\beta}(\beta^{\ast},\Lambda^{\ast})-\ell_{\beta}(\beta_{n}^{\ast},\Lambda^{\ast})\right\}  -n^{1/2}\mathbb{P}_n\left\{ \ell_{\Lambda}(\beta^{\ast},\Lambda^{\ast})(h^{\ast})-\ell_{\Lambda}(\beta_{n}^{\ast},\Lambda^{\ast})(h^{\ast})\right\}  \\
 & =n^{1/2}P\left\{ \ell_{\beta}(\beta^{\ast},\Lambda^{\ast})-\ell_{\beta}(\beta_{n}^{\ast},\Lambda^{\ast})\right\}  -n^{1/2}P\left\{ \ell_{\Lambda}(\beta^{\ast},\Lambda^{\ast})(h^{\ast})-\ell_{\Lambda}(\beta_{n}^{\ast},\Lambda^{\ast})(h^{\ast})\right\}  +o_{P}(1)\\
 & =-P\left\{ \ell_{\beta\beta}-\ell_{\Lambda\beta}(h^{\ast})\right\} \theta^{\ast}+o_{P}(1)\\
 & =P\left[\left\{ \ell_{\beta}(\beta^{\ast},\Lambda^{\ast})-\ell_{\Lambda}(\beta^{\ast},\Lambda^{\ast})(h^{\ast})\right\} ^{\otimes2}\right]\theta^{\ast}+o_{P}(1)\\
 & =\mathcal{I}\theta^{\ast}+o_{P}(1).
\end{align*}
where $\ell_{\beta\beta}$ and $\ell_{\Lambda\beta}$ are evaluated at $(\beta^{\ast},\Lambda^{\ast})$ and we have used an identity
\[
P\left\{ \ell_{\Lambda}(\beta^{\ast},\Lambda^{\ast})(h^{\ast})\ell_{\beta}(\beta^{\ast},\Lambda^{\ast})^{\top}\right\} =P\left[\left\{ \ell_{\Lambda}(\beta^{\ast},\Lambda^{\ast})(h^{\ast})\right\} ^{\otimes2}\right]
\]
from the proof of Theorem 2 in \cite{zeng2016maximum}. The result \eqref{asympcont_effscore} then follows.

\subsubsection{Modify the arguments in Section \ref{section.inactive}}
In order to prove that
\begin{align*}
\bigg\| n^{1/2}\mathbb{P}_{n}\tilde{\ell}_{-M}(\beta_{n}^{\ast}) & -n^{1/2}\mathcal{I}_{-M,M}(\overline{\beta}^{M}-\widetilde{\beta}_{n}^{M})+C\mathcal{I}_{-M,M}\mathcal{I}_{M,M}^{-1}s_{M}\\
 & \quad-(\mathcal{I}_{-M,M}\mathcal{I}_{M,M}^{-1}\mathcal{I}_{M,-M}-\mathcal{I}_{-M,-M})\theta_{-M}^{\ast}+o_{P}(1)\bigg\|_{\infty}\le n^{-1/2}\lambda,
\end{align*}
it suffices to show that
\[
n^{1/2}\mathbb{P}_n\left\{ \ell_{\beta}(\hat{\beta}^{M}_{\textrm{full}},\hat{\Lambda}^{M})-\tilde{\ell}(\beta_{n}^{\ast})\right\} =-n^{1/2}\mathcal{I}(\hat{\beta}^{M}_{\textrm{full}}-\beta_{n}^{\ast})+o_{P}(1).
\]
 By Condition \ref{condition.one.step.semi},
\[
n^{1/2}\mathbb{P}_n\left\{ \ell_{\beta}(\hat{\beta}^{M}_{\textrm{full}},\hat{\Lambda}^{M})-\tilde{\ell}(\beta^{\ast})\right\} =-n^{1/2}\mathcal{I}(\hat{\beta}^{M}_{\textrm{full}}-\beta^{\ast})+o_{P}(1).
\]
As shown in Section \ref{section.one.step.local.semi},
\[
n^{1/2}\mathbb{P}_n\left\{ \tilde{\ell}(\beta_{n}^{\ast})-\tilde{\ell}(\beta^{\ast})\right\} =-n^{1/2}\mathcal{I}(\beta_{n}^{\ast}-\beta^{\ast})+o_{P}(1).
\]
The above two equations imply
\[
n^{1/2}\mathbb{P}_n\left\{ \ell_{\beta}(\hat{\beta}^{M}_{\textrm{full}},\hat{\Lambda}^{M})-\tilde{\ell}(\beta_{n}^{\ast})\right\} =-n^{1/2}\mathcal{I}(\hat{\beta}^{M}_{\textrm{full}}-\beta_{n}^{\ast})+o_{P}(1).
\]

\subsection{Proof of Theorem \ref{theorem.interval}}
As discussed in Section \ref{section.result.interval}, it suffices to prove Lemma \ref{lemma.one.step.interval} in order to prove Theorem \ref{theorem.interval}. The proof is given below.

\begin{proof-of-lemma}[\ref{lemma.one.step.interval}]
 By the definition of $\hat{\beta}^{M}$ and the concavity of $\ell(\beta,\Lambda)$ \citep[shown in the proof of Theorem 1 in][]{li2020adaptive}, we have
\[(\hat{\beta}^{M}_{\textrm{full}},\hat{\Lambda}^{M})=\argmax_{(\beta,\Lambda):\beta_{-M}=0}\mathbb{P}_{n}\ell(\beta,\Lambda)-\frac{\lambda}{n}s_M^{\top}\beta_M.
\]
Then following the proof of Theorem 1 in \cite{li2020adaptive}, it is easy to show that $\|\hat{\beta}^{M}-\beta^{\ast}\|_{2}=O_{P}(n^{-1/2})$.
Since
\[
\mathbb{P}_{n}\ell(\hat{\beta}^{M}_{\textrm{full}},\hat{\Lambda}^{M})-\mathbb{P}_{n}\ell(\beta^{\ast},\Lambda^{\ast})\ge\frac{\lambda}{n}s_{M}^{\top}\left\{ \hat{\beta}^{M}-\beta^{\ast}_M\right\} \ge-\frac{\lambda}{n}\|\beta^{\ast}_M-\hat{\beta}^{M}\|_{1}\ge-\frac{\lambda}{n}\sqrt{p}\|\beta^{\ast}_M-\hat{\beta}^{M}\|_{2}=-O_{P}(n^{-1}),
\]
it follows that $\sup_{t\in[\zeta,\tau]}\bigg|\hat{\Lambda}^{M}(t)-\Lambda^{\ast}(t)\bigg|\longrightarrow0$ a.s. by adapting the proof of Theorem 1 in \cite{zeng2017maximum}. Furthermore, $\bigg\|\dot{\Lambda}_{\hat{\beta}^{M}_{\textrm{full}}}+h^{\ast}\bigg\|_{L_2(\Lambda^{\ast})}=o_{P}(1)$ following similar arguments to the counterpart in the proof of Theorem 3 in \cite{zeng2017maximum}, where $\dot{\Lambda}_{\tilde{\beta}}=\partial(d\hat{\Lambda}_{\beta})/\partial\beta\big|_{\beta=\tilde{\beta}} (d\hat{\Lambda}_{\tilde{\beta}})^{-1} $ and $\|h\|_{L_2(\Lambda^{\ast})}^2=\int_{\zeta}^{\tau}h^{2}(t)d\Lambda^{\ast}(t)$ as in \cite{zeng2017maximum}.

We now prove
\begin{equation} \label{eq.lasso.taylor.2}
n^{1/2}\mathbb{P}_{n}\left\{ \ell_{\beta}(\hat{\beta}^{M}_{\textrm{full}},\hat{\Lambda}^{M})-\tilde{\ell}(\beta^{\ast})\right\} =-n^{1/2}\mathcal{I}\left(\hat{\beta}^{M}_{\textrm{full}}-\beta^{\ast}\right)+o_{P}(1).
\end{equation}
Recall that
\[
\mathbb{G}_nl_{\beta}(\hat{\beta}^{M}_{\textrm{full}})=\mathbb{G}_n\left\{ \ell_{\beta}(\hat{\beta}^{M}_{\textrm{full}},\hat{\Lambda}^{M})+\ell_{\Lambda}(\hat{\beta}^{M}_{\textrm{full}},\hat{\Lambda}^{M})(\dot{\Lambda}_{\hat{\beta}^{M}_{\textrm{full}}})\right\} .
\]
By the arguments in the proof of Lemma 2 of \cite{zeng2017maximum},
$\ell_{\beta}(\hat{\beta}^{M}_{\textrm{full}},\hat{\Lambda}^{M})$ and $\ell_{\Lambda}(\hat{\beta}^{M}_{\textrm{full}},\hat{\Lambda}^{M})(\dot{\Lambda}_{\hat{\beta}^{M}_{\textrm{full}}})$ belong to two Donsker classes respectively.
So \[
\mathbb{G}_n\ell_{\Lambda}(\hat{\beta}^{M}_{\textrm{full}},\hat{\Lambda}^{M})(\dot{\Lambda}_{\hat{\beta}^{M}_{\textrm{full}}})=\mathbb{G}_n\ell_{\Lambda}(\beta^{\ast},\Lambda^{\ast})(-h^{\ast})+o_{P}(1)=\mathbb{G}_n\ell_{\Lambda}(\hat{\beta}^{M}_{\textrm{full}},\hat{\Lambda}^{M})(-h^{\ast})+o_{P}(1)
\]
by Lemma 19.24 of \citet{van2000asymptotic}. Note that
\begin{align*}
\mathbb{G}_n\ell_{\beta}(\hat{\beta}^{M}_{\textrm{full}},\hat{\Lambda}^{M}) & =n^{1/2}\mathbb{P}_n\ell_{\beta}(\hat{\beta}^{M}_{\textrm{full}},\hat{\Lambda}^{M})-n^{1/2}\left\{ P\ell_{\beta}(\hat{\beta}^{M}_{\textrm{full}},\hat{\Lambda}^{M})-P\ell_{\beta}(\beta^{\ast},\Lambda^{\ast})\right\}
\end{align*}
and
\begin{align*}
\mathbb{G}_n\ell_{\Lambda}(\hat{\beta}^{M}_{\textrm{full}},\hat{\Lambda}^{M})(-h^{\ast}) & =-n^{1/2}\left\{ P\ell_{\Lambda}(\hat{\beta}^{M}_{\textrm{full}},\hat{\Lambda}^{M})(-h^{\ast})-P\ell_{\Lambda}(\beta^{\ast},\Lambda^{\ast})(-h^{\ast})\right\} .
\end{align*}
We have shown before that
\[
\mathbb{P}_n\ell(\hat{\beta}^{M}_{\textrm{full}},\hat{\Lambda}^{M})-\mathbb{P}_n\ell(\beta^{\ast},\Lambda^{\ast})\ge -O_{P}(n^{-1}).
\]
This result enables us to adapt the proof of Lemma A1 in \cite{zeng2016maximum} to show that
\[
E\left(\sum_{k=1}^K\left[\hat{\Lambda}^{M}(U_k)\exp(\hat{\beta}^{M\top}_{\textrm{full}}X)-\Lambda^{\ast}(U_k)\exp({\beta^{\ast}}^{\top}X)\right]^2\right)=O_p(n^{-2/3}).
\]
Then, as argued in the proof of Theorem 2 of \cite{zeng2016maximum}, the second-order terms of the Taylor expansions of $-n^{1/2}\left\{ P\ell_{\beta}(\hat{\beta}^{M}_{\textrm{full}},\hat{\Lambda}^{M})-P\ell_{\beta}(\beta^{\ast},\Lambda^{\ast})\right\}$ and $-n^{1/2}\left\{ P\ell_{\Lambda}(\hat{\beta}^{M}_{\textrm{full}},\hat{\Lambda}^{M})(-h^{\ast})-P\ell_{\Lambda}(\beta^{\ast},\Lambda^{\ast})(-h^{\ast})\right\}$  around $(\beta^{\ast},\Lambda^{\ast})$ are bounded by
\[
    n^{1/2}O\left\{\|\hat{\beta}^{M}_{\textrm{full}}-\beta^{\ast}\|_2^2+E\left(\sum_{k=1}^K\left[\hat{\Lambda}^{M}(U_k)\exp(\hat{\beta}^{M\top}_{\textrm{full}}X)-\Lambda^{\ast}(U_k)\exp({\beta^{\ast}}^{\top}X)\right]^2\right)\right\}
    =O(n^{1/2}\|\hat{\beta}^{M}_{\textrm{full}}-\beta^{\ast}\|_2^2+n^{-1/6})
    =o_{P}(1).
\]
So
\begin{align*}
\mathbb{G}_n\ell_{\beta}(\hat{\beta}^{M}_{\textrm{full}},\hat{\Lambda}^{M}) & = n^{1/2}\mathbb{P}_n\ell_{\beta}(\hat{\beta}^{M}_{\textrm{full}},\hat{\Lambda}^{M})-n^{1/2}\left[P\ell_{\beta\beta}(\hat{\beta}^{M}_{\textrm{full}}-\beta^{\ast})+P\ell_{\beta\Lambda}\left\{ d(\hat{\Lambda}^{M}-\Lambda^{\ast})/d\Lambda^{\ast}\right\} \right] +o_{P}(1)
\end{align*}
and
\begin{align*}
\mathbb{G}_n\ell_{\Lambda}(\hat{\beta}^{M}_{\textrm{full}},\hat{\Lambda}^{M})(\dot{\Lambda}_{\hat{\beta}^{M}_{\textrm{full}}}) & =-n^{1/2}\left[P\ell_{\Lambda\beta}(-h^{\ast})(\hat{\beta}^{M}_{\textrm{full}}-\beta^{\ast})+P\ell_{\Lambda\Lambda}\left\{ -h^{\ast},d(\hat{\Lambda}^{M}-\Lambda^{\ast})/d\Lambda^{\ast}\right\} \right] +o_{P}(1).
\end{align*}
Here $\ell_{\beta\Lambda}(d(\hat{\Lambda}^{M}-\Lambda^{\ast})/d\Lambda^{\ast})$ is the derivative of $\ell_{\beta}(\beta,\Lambda)$ along the submodel $d\Lambda_{\epsilon,h}=(1+\epsilon h)d\Lambda$ defined in Section \ref{section.idea} with $h=d(\hat{\Lambda}^{M}-\Lambda^{\ast})/d\Lambda^{\ast}$, which is equivalent to the submodel $\Lambda_{\epsilon}=\Lambda^{\ast}+\epsilon(\hat{\Lambda}^{M}-\Lambda^{\ast})$, $\ell_{\Lambda\Lambda}(h,d(\hat{\Lambda}^{M}-\Lambda^{\ast})/d\Lambda^{\ast})$ is the derivative of $\ell_{\Lambda}(h)$ along the submodel $\Lambda_{\epsilon}=\Lambda^{\ast}+\epsilon(\hat{\Lambda}^{M}-\Lambda^{\ast})$, and $\ell_{\beta\beta},\ell_{\beta\Lambda},\ell_{\Lambda\beta}$ and $\ell_{\Lambda\Lambda}$ are evaluated at $(\beta^{\ast},\Lambda^{\ast})$. Thus we have
\begin{align*}
\mathbb{G}_nl_{\beta}(\hat{\beta}^{M}_{\textrm{full}}) & =n^{1/2}\mathbb{P}_n\ell_{\beta}(\hat{\beta}^{M}_{\textrm{full}},\hat{\Lambda}^{M})+n^{1/2}\mathcal{I}(\hat{\beta}^{M}_{\textrm{full}}-\beta^{\ast}) +o_{P}(1).
\end{align*}
By Lemma 19.24 of \citet{van2000asymptotic},
\begin{align*}
\mathbb{G}_{n}l_{\beta}(\hat{\beta}^{M}_{\textrm{full}}) & =\mathbb{G}_{n}\left\{ \ell_{\beta}(\hat{\beta}^{M}_{\textrm{full}},\hat{\Lambda}^{M})+\ell_{\Lambda}(\hat{\beta}^{M}_{\textrm{full}},\hat{\Lambda}^{M})(\dot{\Lambda}_{\hat{\beta}^{M}_{\textrm{full}}})\right\} \\
 & =\mathbb{G}_{n}\left\{ \ell_{\beta}(\beta^{\ast},\Lambda^{\ast})-\ell_{\Lambda}(\beta^{\ast},\Lambda^{\ast})(h^{\ast})\right\} +o_{P}(1)\\
 & =\mathbb{G}_{n}\tilde{\ell}(\beta^{\ast})+o_{P}(1),
\end{align*}
Then (\ref{eq.lasso.taylor.2}) follows from $\mathbb{G}_{n}=n^{1/2}(\mathbb{P}_{n}-P)$ and $P\tilde{\ell}(\beta^{\ast})=0$.

In addition, (\ref{eq.lasso.taylor.2}) implies that
\[
n^{1/2}\mathbb{P}_{n}\ell_{\beta}(\hat{\beta}^{M}_{\textrm{full}},\hat{\Lambda}^{M})=\mathbb{G}_{n}\tilde{\ell}(\beta^{\ast})-n^{1/2}\mathcal{I}(\hat{\beta}^{M}_{\textrm{full}}-\beta^{\ast})+o_{P}(1),
\]
Since $\mathbb{G}_{n}\tilde{\ell}(\beta^{\ast})\stackrel{d}{\longrightsquigarrow}N\left(0,\mathcal{I}\right)$ due to the central limit theorem and $\|\hat{\beta}^{M}_{\textrm{full}}-\beta^{\ast}\|_{2}=O_{P}(n^{-1/2})$, we have $\|n^{1/2}\mathbb{P}_{n}\ell_{\beta}(\hat{\beta}^{M}_{\textrm{full}},\hat{\Lambda}^{M})\|_{2}=O_P(1)$. Furthermore, by (\ref{eq.lasso.taylor.2}) again,
\[
n^{1/2}\mathcal{I}(\hat{\beta}^{M}_{\textrm{full}}-\beta^{\ast})=-n^{1/2}\mathbb{P}_{n}\ell_{\beta}(\hat{\beta}^{M}_{\textrm{full}},\hat{\Lambda}^{M})+\mathbb{G}_{n}\tilde{\ell}(\beta^{\ast})+o_{P}(1).
\]
Note that $\hat{\beta}^{M}_{\textrm{full},-M}=\beta_{-M}^{\ast}=0$ and
\[
\overline{\beta}^{M}=\hat{\beta}^{M}+(\hat{I}_{M,n}/n)^{-1}\mathbb{P}_{n}l_{\beta_{M}}(\hat{\beta}^{M}_{\textrm{full}})=\hat{\beta}^{M}+(\hat{I}_{M,n}/n)^{-1}\mathbb{P}_{n}\ell_{\beta_{M}}(\hat{\beta}^{M}_{\textrm{full}},\hat{\Lambda}^{M})
\]
So,
\begin{align*}
n^{1/2}\mathcal{I}_{M,M}(\overline{\beta}^{M}-\beta_{M}^{\ast}) & =n^{1/2}\mathcal{I}_{M,M}(\overline{\beta}^{M}-\hat{\beta}^{M})+n^{1/2}\mathcal{I}_{M,M}(\hat{\beta}^{M}-\beta_{M}^{\ast})\\
 & =n^{1/2}\mathcal{I}_{M,M}(\hat{I}_{M,n}/n)^{-1}\mathbb{P}_{n}\ell_{\beta_{M}}(\hat{\beta}^{M}_{\textrm{full}},\hat{\Lambda}^{M})-n^{1/2}\mathbb{P}_{n}\ell_{\beta_{M}}(\hat{\beta}^{M}_{\textrm{full}},\hat{\Lambda}^{M})+\mathbb{G}_{n}\tilde{\ell}_{M}(\beta^{\ast})+o_{P}(1)\\
 & =\mathcal{I}_{M,M}\{(\hat{I}_{M,n}/n)^{-1}-\mathcal{I}_{M,M}^{-1}\}n^{1/2}\mathbb{P}_{n}\ell_{\beta_{M}}(\hat{\beta}^{M}_{\textrm{full}},\hat{\Lambda}^{M})+\mathbb{G}_{n}\tilde{\ell}_{M}(\beta^{\ast})+o_{P}(1)\\
 & =\mathbb{G}_{n}\tilde{\ell}_{M}(\beta^{\ast})+o_{P}(1).
\end{align*}
\end{proof-of-lemma}

\subsection{Proof of Theorem \ref{theorem.PRES}}
Similar to the proof of Lemma \ref{lemma.one.step.interval}, it is easy to show that $\|\hat{\beta}-\beta^{\ast}\|_{2}=O_{P}(n^{-1/2})$. Further more, \[
\mathbb{P}_{n}\ell(\hat{\beta},\hat{\Lambda})-\mathbb{P}_{n}\ell(\beta^{\ast},\Lambda^{\ast})\ge\frac{\lambda}{n}\|\hat{\beta}\|_{1}-\frac{\lambda}{n}\|\beta^{\ast}\|_{1}\ge-\frac{\lambda}{n}\|\beta^{\ast}-\hat{\beta}\|_{1}\ge-\frac{\lambda}{n}\sqrt{p}\|\beta^{\ast}-\hat{\beta}\|_{2}=-O_{P}(n^{-1}).
\]
So similar arguments to the proof of Lemma \ref{lemma.one.step.interval} lead to $\sup_{t\in[\zeta,\tau]}\bigg|\hat{\Lambda}(t)-\Lambda^{\ast}(t)\bigg|\longrightarrow0$ a.s. and $\bigg\|\dot{\Lambda}_{\hat{\beta}}+h^{\ast}\bigg\|_{L_2(\Lambda^{\ast})}=o_{P}(1)$.

Note that
\[
l_{\beta}(\hat{\beta})=\ell_{\beta}(\hat{\beta},\hat{\Lambda})+\ell_{\Lambda}(\hat{\beta},\hat{\Lambda})(\dot{\Lambda}_{\hat{\beta}})
\]
and
\[
\mathbb{P}_n\ell_{\Lambda}(\hat{\beta},\hat{\Lambda})(\dot{\Lambda}_{\hat{\beta}})=0.
\]
So
\[
\mathbb{P}_nl_{\beta\beta}(\hat{\beta})=\mathbb{P}_n\left\{ \ell_{\beta\beta}(\hat{\beta},\hat{\Lambda})+\ell_{\beta\Lambda}(\hat{\beta},\hat{\Lambda})(\dot{\Lambda}_{\hat{\beta}})\right\} .
\]
By similar arguments to the proof of Lemma 1 in \citet{zeng2017maximum}, it can be shown that $\ell_{\beta\beta}(\hat{\beta},\hat{\Lambda})$
and $\ell_{\beta\Lambda}(\hat{\beta},\hat{\Lambda})(\dot{\Lambda}_{\hat{\beta}})$ belong to two Glivenko-Cantelli classes respectively. Therefore,
\begin{align*}
-\mathbb{P}_nl_{\beta\beta}(\hat{\beta}) & =-P\left\{ \ell_{\beta\beta}(\hat{\beta},\hat{\Lambda})+\ell_{\beta\Lambda}(\hat{\beta},\hat{\Lambda})(\dot{\Lambda}_{\hat{\beta}})\right\} +o_{P}(1)\\
 & =-P\left\{ \ell_{\beta\beta}(\beta^{\ast},\Lambda^{\ast})-\ell_{\beta\Lambda}(\beta^{\ast},\Lambda^{\ast})(h^{\ast})\right\} +o_{P}(1)\\
 & =P\left[\left\{ \ell_{\beta}(\beta^{\ast},\Lambda^{\ast})-\ell_{\Lambda}(\beta^{\ast},\Lambda^{\ast})(h^{\ast})\right\} ^{\otimes2}\right]+o_{P}(1)\\
 & =\mathcal{I}+o_{P}(1).
\end{align*}
Here we have used an identity
\[
P\left\{ \ell_{\beta}(\beta^{\ast},\Lambda^{\ast})\ell_{\Lambda}(\beta^{\ast},\Lambda^{\ast})(h^{\ast})^{\top}\right\} =P\left[\left\{ \ell_{\Lambda}(\beta^{\ast},\Lambda^{\ast})(h^{\ast})\right\} ^{\otimes2}\right]
\]
from the proof of Theorem 2 in \cite{zeng2016maximum}.

\subsection{Proof of Theorem \ref{theorem.LS}}
Following similar arguments to the proof of Theorem 4.1 in \citet{HuWe1997}, we can show that the least favorable direction $g^{\ast}$ exists and has a bounded total variation over $[\zeta,\tau]$, the union of the support of the inspection times. Then,
by Theorem 9.1 in \citet{huang2012least}, we only need to verify their conditions (A1) and (A2) as well as the existence of a sequence of $g_n\in\mathcal{G}_{M,n}^p$ converging to $g^{\ast}$ in probability
to prove the consistency of the least squares information estimator.

Since $\ell_{\beta}(\hat{\beta},\hat{\Lambda})$ and $\dot{\ell}_{\Lambda}(\hat{\beta},\hat{\Lambda})(g)$ belong
to two Glivenko--Cantelli classes respectively, where $g\in\mathcal{G}_{M,n}^{p}$, $\left\{ \ell_{\beta}(\hat{\beta},\hat{\Lambda})-\dot{\ell}_{\Lambda}(\hat{\beta},\hat{\Lambda})(g)\right\} ^{\otimes2}$ also belongs to a Glivenko--Cantelli class due to the uniform boundedness of the classes $\{\ell_{\beta}(\beta,\Lambda):(\beta,\Lambda)\in V_1\}$ and $\{\dot{\ell}_{\Lambda}(\beta,\Lambda):(\beta,\Lambda)\in V_1\ \mbox{and}\ g\in\mathcal{G}_M\}$, where $\mathcal{G}_M=\{g:g(0)=0\ \mbox{and}\ V(g)\le M\}$. So Condition (A1) in \citet{huang2012least} is satisfied.

For Condition (A2), we only need to show that $\sup_{g\in\mathcal{G}_{M,n}}P\left\{ \ell_{\Lambda}(\hat{\beta},\hat{\Lambda})(g)-\dot{\ell}_{\Lambda}(\beta^{\ast},\Lambda^{\ast})(g)\right\} ^{2}=o_{P}(1)$, because the rest are direct results of the dominated convergence theorem. Let
$$
q(\beta,\Lambda;t,X)=\Lambda(t)\exp\left(\beta^{\top}X\right)$$ and $$Q(\beta,\Lambda;t,X)=\exp\left\{ -q(t,X;\beta,\Lambda)\right\}.
$$
Then for the $i$th subject, its log likelihood is
\[
\ell(\beta,\Lambda;X,L,R)=\log\left\{ Q(L,X;\beta,\Lambda)-Q(R,X;\beta,\Lambda)\right\}
\]
and the score function with respect to $\Lambda$ along the direction $g$ is
\begin{align*}
\dot{\ell}_{\Lambda}(\beta,\Lambda;L,R,X)(g) & =-\frac{Q(\beta,\Lambda;L,X)\cdot q(\beta,g;L,X)-Q(\beta,\Lambda;R,X)\cdot q(\beta,g;R,X)}{Q(\beta,\Lambda;L,X)-Q(\beta,\Lambda;R,X)}.
\end{align*}
Since $g$ has bounded variation, the covariate $X$ is bounded, $\|\hat{\beta}-\beta^{\ast}\|_2=o_{P}(1)$, and $\sup_{t\in[\zeta,\tau]}|\hat{\Lambda}(t)-\Lambda^{\ast}(t)|\to0$ almost surely, it suffices to show that
\begin{align*}
P|Q(\hat{\beta},\hat{\Lambda};L,X)-Q(\beta^{\ast},\Lambda^{\ast};L,X)|^{2} & =o_{P}(1),\quad P|Q(\hat{\beta},\hat{\Lambda};R,X)-Q(\beta^{\ast},\Lambda^{\ast};R,X)|^{2}=o_{P}(1),\\
\sup_{g\in\mathcal{G}_{M,n}}P|q(\hat{\beta},g;L,X)-q(\beta^{\ast},g;L,X)|^{2} & =o_{P}(1),\quad\sup_{g\in\mathcal{G}_{M,n}}P|q(\hat{\beta},g;R,X)-q(\beta^{\ast},g;R,X)|^{2}=o_{P}(1).
\end{align*}
In the proof of Theorem \ref{theorem.PRES}, we showed that
\[
\mathbb{P}_{n}\ell(\hat{\beta},\hat{\Lambda})-\mathbb{P}_{n}\ell(\beta^{\ast},\Lambda^{\ast})\ge O_{P}(n^{-1}).
\]
So, by similar arguments to the proof of Lemma A1 in \citet{zeng2016maximum}, the Hellinger distance between $(\hat{\beta},\hat{\Lambda})$ and $(\beta^{\ast},\Lambda^{\ast})$ can be shown to be $O_{P}(n^{-1/3})$. Then, following again the proof of Lemma A1 in \cite{zeng2016maximum}, we can show that
\[
\mathbb{E}\left\{ \sum_{k=1}^{K}|Q(\hat{\beta},\hat{\Lambda};U_{k},X)-Q(\beta^{\ast},\Lambda^{\ast};U_{k},X)|^{2}\right\} =O_{P}(n^{-2/3})=o_{P}(1),
\]
which implies
\begin{align*}
P|Q(\hat{\beta},\hat{\Lambda};L,X)-Q(\beta^{\ast},\Lambda^{\ast};L,X)|^{2} &=o_{P}(1)\quad \mbox{and}\quad P|Q(\hat{\beta},\hat{\Lambda};L,X)-Q(\beta^{\ast},\Lambda^{\ast};R,X)|^{2}=o_{P}(1),
\end{align*}
since both $L$ and $R$ belong to the set of inspections times $\{U_{1},U_{2},\dots,U_{K}\}.$ Furthermore, by the mean value theorem, the result $\hat{\beta}-\beta^{\ast}=o_{P}(1)$ implies
\begin{align*}
\sup_{g\in\mathcal{G}_{M,n}}P|q(\hat{\beta},g;L,X)-q(\beta^{\ast},g;L,X)|^{2} & =o_{P}(1)\quad \mbox{and}\quad \sup_{g\in\mathcal{G}_{M,n}}P|q(\hat{\beta},g;R,X)-q(\beta^{\ast},g;R,X)|^{2}=o_{P}(1),
\end{align*}
since $g$ has bounded variation and the covariate $X$ is bounded. So Condition (A2) is satisfied.

To apply Theorem 9.1 in \citet{huang2012least}, it remains to show that there exists $g_{n}\in\mathcal{G}_{M,n}^{p}$ converging to $g^{\ast}$ in probability. This is equivalent to showing that each component converges. Thus we assume $g^{\ast}$ is of one-dimension without loss of generality. That is, we need to show that there exists a sequence of $g_{n}$'s  such that
\[
P\left[\left\{ g_{n}(L)-g^{\ast}(L)\right\} ^{2}+\left\{ g_{n}(R)-g^{\ast}(R)\right\} ^{2}\right]=o_{P}(1),
\]
where the norm used is inherited from Section 9.4.1 of \citet{huang2012least}. Let $u_{0}=0$ and
\[
g_{n}(t)=\sum_{k=1}^{m}\left\{ g^{\ast}(u_{k})-g^{\ast}(u_{k-1})\right\} I(t\ge u_{k}),
\]
which is a step function that only changes its value at $u_{1},u_{2},\dots,u_{m}$ and satisfies $g_{n}(u_{k})=g^{\ast}(u_{k})$, $(k=1,\dots,m)$. We prove a stronger result,
\[
\sup_{t\in[\zeta,\tau]}|g_{n}(t)-g^{\ast}(t)|=o_{P}(1).
\]
Due to our definition of $g_{n}$, this is equivalent to showing that
\[
\sup_{t\in[\zeta,u_{1}]}|g^{\ast}(t)-g^{\ast}(\zeta)|=o_{P}(1)\quad\mbox{and}\quad \max_{k\in\left\{ 1,2,\dots,m\right\} }\sup_{t\in[u_{k},u_{k+1}]}|g^{\ast}(t)-g^{\ast}(u_{k})|=o_{P}(1),
\]
where $u_{m+1}=\tau$. As $g^{\ast}$ is continuously differentiable on $[\zeta,\tau]$ \citep{HuWe1997}, $g^{\ast\prime}$, the derivative of $g^{\ast}$,  is bounded on $[\zeta,\tau]$. By the mean value theorem,
\begin{align*}
\sup_{t\in[\zeta,u_{1}]}|g^{\ast}(t)-g^{\ast}(\zeta)| & \le\|g^{\ast\prime}\|_{L^{\infty}[\zeta,\tau]}^{2}\cdot(u_{1}-\zeta),\\
\max_{k\in\left\{ 1,2,\dots,m\right\} }\sup_{t\in[u_{k},u_{k+1}]}|g^{\ast}(t)-g^{\ast}(u_{k})| & \le\|g^{\ast\prime}\|_{L^{\infty}[\zeta,\tau]}^{2}\cdot\max_{k\in\left\{ 1,2,\dots,m\right\} }(u_{k+1}-u_{k}).
\end{align*}
We consider a partition $t_{j}=\zeta+j(\tau-\zeta)/N$, $(j=0,1,\dots,N)$. If $\left\{ u_{1},\dots,u_{m}\right\} \bigcap[t_{j-1},t_{j}]\neq\emptyset$ for all $j=1,\dots,N$, $u_{1}-\zeta\le(\tau-\zeta)/N$ and $u_{k+1}-u_{k}\le2(\tau-\zeta)/N$ for all $k=1,\dots,m$, and so
\begin{align*}
\sup_{t\in[\zeta,u_{1}]}|g^{\ast}(t)-g^{\ast}(\zeta)| & \le\|g^{\ast\prime}\|_{L^{\infty}[\zeta,\tau]}^{2}\cdot\frac{\tau-\zeta}{N},\\
\max_{k\in\left\{ 1,2,\dots,m\right\} }\sup_{t\in[u_{k},u_{k+1}]}|g^{\ast}(t)-g^{\ast}(u_{k})| & \le\|g^{\ast\prime}\|_{L^{\infty}[\zeta,\tau]}^{2}\cdot\frac{2(\tau-\zeta)}{N}.
\end{align*}
By Lemma \ref{lemma.inspection.2} below, which builds on Lemma \ref{lemma.inspection.1}, as $n\to\infty$,
\[
\text{P}\left(\left\{ u_{1},\dots,u_{m}\right\} \bigcap[t_{j-1},t_{j}]\neq\emptyset\quad\text{for all }j=1,\dots,N\right)\longrightarrow1.
\]
Therefore
\begin{align*}
\text{P}\left\{ \sup_{t\in[\zeta,u_{1}]}|g^{\ast}(t)-g^{\ast}(\zeta)|>\|g^{\ast\prime}\|_{L^{\infty}[\zeta,\tau]}^{2}\cdot\frac{\tau-\zeta}{N}\right\}  & \longrightarrow0,\\
\text{P}\left\{ \max_{k\in\left\{ 1,2,\dots,m\right\} }\sup_{t\in[u_{k},u_{k+1}]}|g^{\ast}(t)-g^{\ast}(u_{k})|>\|g^{\ast\prime}\|_{L^{\infty}[\zeta,\tau]}^{2}\cdot\frac{2(\tau-\zeta)}{N}\right\}  & \longrightarrow0.
\end{align*}
Letting $N\to\infty$, we get the desired result.

\begin{lemma}
\label{lemma.inspection.1}
For any interval $[\zeta_{0},\tau_{0}]\subset[\zeta,\tau]$ where $\zeta_{0}<\tau_{0}$,
\[
\text{P}\left(L\in[\zeta_{0},\tau_{0}]\right)>0,\quad\text{P}\left(R\in[\zeta_{0},\tau_{0}]\right)>0,
\]
where $L$ and $R$ are the last inspection time before failure time $T$ and first inspection time after $T$ respectively.
\end{lemma}

\begin{proof}
We first show that $\text{P}\left(L\in[\zeta_{0},\tau_{0}]\right)>0$. Consider the inspection times $U_{1},\dots,U_{K}$.  If $T\in\left(\frac{\zeta_{0}+\tau_{0}}{2},\tau_{0}\right]$ and there exists $U_{k}$ such that $U_{k}\in\left[\zeta_{0},\frac{\zeta_{0}+\tau_{0}}{2}\right]$,
then we must have $L\in[\zeta_{0},\tau_{0}]$. Since the failure time and the inspection times are independent given the covariates, it suffices to show that
\[
\int\text{P}\left(\left\{ U_{1},\dots,U_{K}\right\} \bigcap\left[\zeta_{0},\frac{\zeta_{0}+\tau_{0}}{2}\right]\neq\emptyset\mid X\right)\cdot\text{P}\left(T\in\left(\frac{\zeta_{0}+\tau_{0}}{2},\tau_{0}\right]\mid X\right)\cdot p_{X}(x)dx>0.
\]
As $\Lambda^{\ast}$ is strictly increasing on $[\zeta,\tau]$ and the components of $X$ are bounded, there exists $\epsilon>0$ such that
\[
\text{P}\left(T\in\left(\frac{\zeta_{0}+\tau_{0}}{2},\tau_{0}\right]\mid X\right)=\exp\left\{ -\Lambda^{\ast}\left(\frac{\zeta_{0}+\tau_{0}}{2}\right)\exp(X^{\top}\beta^{\ast})\right\} -\exp\left\{ -\Lambda^{\ast}(\tau_{0})\exp(X^{\top}\beta^{\ast})\right\} \ge\epsilon.
\]
In addition,
\[
\int\text{P}\left(\left\{ U_{1},\dots,U_{K}\right\} \bigcap\left[\zeta_{0},\frac{\zeta_{0}+\tau_{0}}{2}\right]\neq\emptyset\mid X\right)p_{X}(x)dx=\text{P}\left(\left\{ U_{1},\dots,U_{K}\right\} \bigcap\left[\zeta_{0},\frac{\zeta_{0}+\tau_{0}}{2}\right]\neq\emptyset\right)>0,
\]
where the last inequality comes from the condition that the union of the support of $(U_{1},\dots,U_{K})$ is $[\zeta,\tau]$. Thus,
\[
\text{P}\left(L\in[\zeta_{0},\tau_{0}]\right)\ge\epsilon\cdot\text{P}\left(\left\{ U_{1},\dots,U_{K}\right\} \bigcap\left[\zeta_{0},\frac{\zeta_{0}+\tau_{0}}{2}\right]\neq\emptyset\right)>0.
\]
By similar arguments, we can show $\text{P}\left(R\in[\zeta_{0},\tau_{0}]\right)>0$.
\end{proof}

\begin{lemma}
\label{lemma.inspection.2}
For any interval $[\zeta_{0},\tau_{0}]\subset[\zeta,\tau]$ where $\zeta_{0}<\tau_{0}$, as $n\to\infty$,
\[
\text{P}\left(\left\{ u_{1},\dots,u_{m}\right\} \bigcap[\zeta_{0},\tau_{0}]\neq\emptyset\right)\longrightarrow1.
\]
\end{lemma}

\begin{proof}
If there exist $i$ and $j$ $(1\le i,j\le n)$ such that
\[
L_{i}\in\left[\zeta_{0},\frac{\zeta_{0}+\tau_{0}}{2}\right]\quad\mbox{and}\quad R_{j}\in\left[\frac{\zeta_{0}+\tau_{0}}{2},\tau_{0}\right],
\]
then there must be at least one maximal intersection in $[\zeta_0,\tau_0]$. So it suffices to show that
\[
\text{P}\left(\left\{ L_{1},\dots,L_{n}\right\} \bigcap\left[\zeta_{0},\frac{\zeta_{0}+\tau_{0}}{2}\right]\neq\emptyset\right)\longrightarrow1\quad\mbox{and}\quad\text{P}\left(\left\{ R_{1},\dots,R_{n}\right\} \bigcap\left[\frac{\zeta_{0}+\tau_{0}}{2},\tau_{0}\right]\neq\emptyset\right)\longrightarrow1.
\]
This is a direct result of Lemma \ref{lemma.inspection.1}.
\end{proof}

\subsection{Proof of Theorem \ref{theorem.right}}
Following \cite{andersen1982cox}, we assume without loss of generality that the failure times lie in the time interval $[0,1]$ and use $N_{i}(s)$ and $Y_{i}(s)$ to denote the counting process and at-risk process of the $i$th subject respectively. Let
\begin{align*}
\mathcal{C}(\beta) & =\frac{1}{n}\sum_{i=1}^{n}\int_{0}^{1}X_{i}^{\top}\beta-\log\left\{ \sum_{j=1}^{n}Y_{j}(s)e^{X_{j}^{\top}\beta}\right\} dN_i(s),\\
\mathcal{U}(\beta) & =\frac{\partial\mathcal{C}}{\partial\beta}=\frac{1}{n}\sum_{i=1}^{n}\int_{0}^{1}X_{i}-\frac{S^{(1)}(s,\beta)}{S^{(0)}(s,\beta)}dN_{i}(s),\\
I(\beta) & =\frac{\partial\mathcal{U}}{\partial\beta}=\frac{1}{n}\sum_{i=1}^{n}\int_{0}^{1}\frac{S^{(2)}(s,\beta)}{S^{(0)}(s,\beta)}-\left\{ \frac{S^{(1)}(s,\beta)}{S^{(0)}(s,\beta)}\right\} ^{\otimes2}dN_{i}(s),
\end{align*}
where
\[
S^{(k)}(t,\beta)=\frac{1}{n}\sum_{i=1}^{n}X_{i}^{\otimes k}Y_{i}(t)e^{X_{i}^{\top}\beta}\quad (k=0,1,2).
\]
Here $\mathcal{C}(\beta)$ is the logarithm of the Cox partial likelihood divided by $n$.

From the proof of Theorem 3.2 of \cite{andersen1982cox}, we have
\[
n^{1/2}\mathcal{U}(\beta^{\ast})\stackrel{d}{\longrightsquigarrow} N(0,\mathcal{I}),
\]
which implies that $\mathcal{U}(\beta^{\ast})$ plays the same role as the  score function $\mathbb{P}_{n}l_{\beta}(\beta^{\ast})$ does in GLM. So we may adapt the proof of Theorem \ref{theorem.GLM} and re-check its conditions under the Cox model with right censored data to prove Theorem \ref{theorem.right}. For any (random) sequence $\beta_{n}$ such that $\|\beta_n-\beta^{\ast}\|_2=O_P(n^{-1/2})$, by equation (2.5) in \cite{andersen1982cox},
\[
n^{1/2}\mathcal{U}(\beta_{n})-n^{1/2}\mathcal{U}(\beta^{\ast})=n^{1/2}I(\tilde{\beta}_{n})\cdot(\beta_{n}-\beta^{\ast}),
\]
where $\tilde{\beta}_{n}$ is on the line segment between $\beta_{n}$ and $\beta^{\ast}$. Since $\beta_{n}$ converges to $\beta^{\ast}$, $\tilde{\beta}_{n}$  converges to $\beta^{\ast}$. Thus, by the proof of Theorem 3.2 in \cite{andersen1982cox}, $I(\tilde{\beta}_{n})\stackrel{p}{\longrightarrow}-\mathcal{I}$. As a result,
\begin{equation} \label{eq.lasso.taylor.3}
n^{1/2}\mathcal{U}(\beta_{n})-n^{1/2}\mathcal{U}(\beta^{\ast})=-n^{1/2}\mathcal{I}(\beta_{n}-\beta^{\ast})+o_{P}(1).
\end{equation}
This enables us to circumvent the asymptotic continuity and Taylor expansion arguments in Section \ref{section.one.step.local} and Section \ref{section.inactive}. What remains to be re-verified is the condition \eqref{condition.LeCam3} for Le Cam's third lemma and Condition \ref{condition.one.step.glm}. To verify them, it suffices to prove
\begin{equation}\label{1st.thing.to.show}
n^{1/2}\mathcal{U}(\beta^{\ast})=n^{1/2}\mathbb{P}_{n}\tilde{\ell}(\beta^{\ast})+o_{P}(1)
\end{equation}
and Lemma \ref{lemma.one.step.right} below.

To show \eqref{1st.thing.to.show}, we leverage the approximately least favorable submodel of the Cox model with right censored data in Section 3 of \cite{murphy2000profile}. Suppose that for each $(\beta,\Lambda)$, there exists a map $t\mapsto\Lambda_{t}(\beta,\Lambda)$ from a fixed neighborhood of $\beta$ into the parameter set for $\Lambda$ such that the map $t\mapsto l(t,\beta,\Lambda)(Z)$ defined by $l(t,\beta,\Lambda)(Z)=\ell\left(t,\Lambda_{t}(\beta,\Lambda)\right)(Z)$ is twice continuously differentiable for all $Z$, where $Z$ denotes the observed data. Denote its derivative by $\dot{l}(t,\beta,\Lambda)(Z)$. The submodel with parameters $\left(t,\Lambda_{t}(\beta,\Lambda)\right)$ should pass through $(\beta,\Lambda)$ at $t=\beta$, i.e., $\Lambda_{\beta}(\beta,\Lambda)=\Lambda$ for every $(\beta,\Lambda)$. Another requirement for the submodel to be approximately least favorable is that $\dot{l}(\beta^{\ast},\beta^{\ast},\Lambda^{\ast})$ is identical to the efficient score function $\tilde{\ell}(\beta^{\ast})$. Note that
\[
\mathcal{U}(\beta^{\ast})=\mathbb{P}_n\dot{l}(\beta^{\ast},\beta^{\ast},\hat{\Lambda}_{\beta^{\ast}})=\mathbb{P}_n\ell_{\beta}(\beta^{\ast},\hat{\Lambda}_{\beta^{\ast}}),
\]
where $\hat{\Lambda}_{\beta^{\ast}}$ is the maximum likelihood estimator for $\Lambda$ fixing $\beta$ at $\beta^{\ast}$.
So equation \eqref{1st.thing.to.show} is equivalent to
\[
n^{1/2}\mathbb{P}_{n}\dot{l}(\beta^{\ast},\beta^{\ast},\hat{\Lambda}_{\beta^{\ast}})=n^{1/2}\mathbb{P}_{n}\tilde{\ell}(\beta^{\ast})+o_{P}(1).
\]
On page 458 of \citet{van2000asymptotic}, it was shown that $\{\dot{l}(t,\beta,\Lambda)\}$ is a Donsker class. If $\hat{\Lambda}_{\beta^{\ast}}=\Lambda^{\ast}+o_{P}(1)$, by Lemma 19.24 of \citet{van2000asymptotic} \citep[see also equation (13) in][] {murphy2000profile}, we have
\[
\mathbb{G}_{n}\dot{l}(\beta^{\ast},\beta^{\ast},\hat{\Lambda}_{\beta^{\ast}})=\mathbb{G}_{n}\tilde{\ell}(\beta^{\ast})+o_{P}(1).
\]
Furthermore,
$
P\dot{l}(\beta^{\ast},\beta^{\ast},\hat{\Lambda}_{\beta^{\ast}})=0
$
for the Cox model with right-censored data \citep[see page 459 of][]{murphy2000profile}.  So, to prove \eqref{1st.thing.to.show},  it remains to show the consistency of $\hat{\Lambda}_{\beta^{\ast}}$. Similar to equation (2.8) in \cite{andersen1982cox},
\begin{align*}
n^{1/2}\left\{ \hat{\Lambda}_{\beta^{\ast}}(t)-\Lambda^{\ast}(t)\right\}  & =n^{1/2}\left[\int_{0}^{t}\frac{d\overline{N}(s)}{nS^{(0)}(s,\beta^{\ast})}-\int_{0}^{t}\lambda^{\ast}(t)I\left\{ \sum_{i=1}^{n}Y_{i}(s)>0\right\} ds\right]\\
 & \quad+n^{1/2}\left[\int_{0}^{t}\lambda^{\ast}(t)I\left\{ \sum_{i=1}^{n}Y_{i}(s)>0\right\} ds-\Lambda^{\ast}(t)\right],
\end{align*}
where $\overline{N}(s)=\sum_{i=1}^{n}N_{i}(s)$. As discussed in Theorem 3.4 of \cite{andersen1982cox}, the second term is asymptotically negligible, while the first term converges weakly to a Gaussian process. So the consistency follows.

To finish the proof, it remains to verify that Condition \ref{condition.one.step.glm} with $\mathbb{P}_nl_{\beta}(\cdot)$ and
 $\mathbb{P}_nl_{\beta_{M}}(\cdot)$ replaced by $\mathcal{U}(\cdot)$ and $\mathcal{U}_M(\cdot)$ respectively holds for the Cox model with right-censored data, which can be implied by \eqref{1st.thing.to.show} and the lemma below.

\begin{lemma}
\label{lemma.one.step.right}
For any fixed $M$, consider the estimator for $\beta_M$,
\[
\hat{\beta}^{M}=\argmin_{\beta_{M}}-\mathcal{C}(\beta_{M},0)+\frac{\lambda}{n}s_{M}^{\top}\beta_{M},
\]
where $\lambda=O(n^{1/2})$, and the one-step estimator
\[
\overline{\beta}^{M}=\hat{\beta}^{M}+(\hat{I}_{M,n}/n)^{-1}\mathcal{U}_{M}(\hat{\beta}^M_{\textrm{full}}),
\]
where $\hat{\beta}^{M}_{\textrm{full}}=(\hat{\beta}^M,0)$, $\hat{I}_{M,n}=(\hat{I_{n}})_{M,M}$ and $\hat{I}_{n}/n$ is a consistent estimator of $\mathcal{I}$. Then under the model $P_{\beta^{\ast},\Lambda^{\ast}}$ and Conditions A - D in \cite{andersen1982cox},
\begin{align}
n^{1/2}\mathcal{U}(\hat{\beta}^{M}_{\textrm{full}})-n^{1/2}\mathcal{U}(\beta^{\ast}) & =-n^{1/2}\mathcal{I}(\hat{\beta}^{M}_{\textrm{full}}-\beta^{\ast})+o_{P}(1),\label{eq.one.step.right.1}\\
n^{1/2}\left(\overline{\beta}^{M}-\beta_{M}^{\ast}\right) & =n^{1/2}\mathcal{I}_{M,M}^{-1}\mathcal{U}_{M}(\beta^{\ast})+o_{P}(1). \label{eq.one.step.right.2}
\end{align}
\end{lemma}

\begin{proof}
First of all, it can be shown that $\|\hat{\beta}^{M}_{\textrm{full}}-\beta^{\ast}\|_{2}=O_{P}(n^{-1/2})$ by a simple modification of the proof of Theorem 1 in \cite{zhang2007adaptive}. Then (\ref{eq.one.step.right.1}) is a direct result of (\ref{eq.lasso.taylor.3}) by replacing $\beta_n$ with $\hat{\beta}^{M}_{\textrm{full}}$. So by equation (\ref{eq.lasso.taylor.3}),
\[
n^{1/2}\mathcal{U}(\hat{\beta}^{M}_{\textrm{full}})=n^{1/2}\mathcal{U}(\beta^{\ast})+n^{1/2}\mathcal{I}(\hat{\beta}^{M}_{\textrm{full}}-\beta^{\ast})+o_{P}(1),
\]
Since $n^{1/2}\mathcal{U}(\beta^{\ast})\stackrel{d}{\longrightsquigarrow} N(0,\mathcal{I})$ and $\|\hat{\beta}^{M}_{\textrm{full}}-\beta^{\ast}\|_{2}=O_{P}(n^{-1/2})$, we have $\|n^{1/2}\mathcal{U}(\hat{\beta}^{M}_{\textrm{full}})\|_{2}=O_{P}(1)$. Furthermore,
\begin{align*}
n^{1/2}\mathcal{I}(\hat{\beta}^{M}_{\textrm{full}}-\beta^{\ast})&=n^{1/2}\mathcal{U}(\beta^{\ast})-n^{1/2}\mathcal{U}(\hat{\beta}^{M})+o_{P}(1).
\end{align*}
Recall that $\hat{\beta}_{\textrm{full},-M}^{M}=\beta_{-M}^{\ast}=0$, $\hat{I}_{n}/n$ is a consistent estimator for the efficient information matrix $\mathcal{I}$ and
$
\overline{\beta}^{M}=\hat{\beta}_{M}+(\hat{I}_{M,n}/n)^{-1}\mathcal{U}_{M}(\hat{\beta}^{M}_{\textrm{full}}).
$
So
\begin{align*}
n^{1/2}\mathcal{I}_{M,M}(\overline{\beta}^{M}-\beta_{M}^{\ast}) & =n^{1/2}\mathcal{I}_{M,M}(\overline{\beta}^{M}-\hat{\beta}^{M})+n^{1/2}\mathcal{I}_{M,M}(\hat{\beta}^{M}-\beta_{M}^{\ast})\\
 & =n^{1/2}\mathcal{I}_{M,M}\left(\hat{I}_{M,n}/n\right)^{-1}\mathcal{U}_{M}(\hat{\beta}^{M}_{\textrm{full}})-n^{1/2}\mathcal{U}_{M}(\hat{\beta}^{M}_{\textrm{full}})+n^{1/2}\mathcal{U}_{M}(\beta^{\ast})+o_{P}(1)\\
 & =\mathcal{I}_{M,M}\{(\hat{I}_{M,n}/n)^{-1}-\mathcal{I}_{M,M}^{-1}\}n^{1/2}\mathcal{U}_{M}(\hat{\beta}^{M}_{\textrm{full}})+n^{1/2}\mathcal{U}_{M}(\beta^{\ast})+o_{P}(1)\\
 & =n^{1/2}\mathcal{U}_{M}(\beta^{\ast})+o_{P}(1).
\end{align*}
\end{proof}

\subsection{Interpretation of $\widetilde{\theta}^{M}$}
\label{section.interpretation}
We now provide the two theorems on $\widetilde{\theta}^{M}=\theta_{n}^{M}+o(1)$, corresponding to GLM and the Cox model with interval-censored data respectively.

\begin{theorem} \label{theorem.glm.interpret}
Consider a specific GLM that is differentiable in quadratic mean at $\beta^{\ast}$. Let $\beta_{n}^{\ast}=n^{-1/2}\theta^{\ast}$ be the local alternative, $\widetilde{\beta}_{n}^{M}=\beta_{M,n}^{\ast}+\mathcal{I}_{M,M}^{-1}\mathcal{I}_{M,-M}\beta_{-M,n}^{\ast}$ be the new inference target, and
\[
\beta_{n}^{M}=\argmax_{\beta_{M}}P_{\beta_{n}^{\ast}}l(\beta_{M},0)
\]
be the minimizer of the population log likelihood under
the local alternative $P_{\beta_{n}^{\ast}}$ restricted to the submodel $M$. Suppose the covariate vector $X$ is almost surely bounded, the covariance matrix of $X$ is positive definite, and $\beta^{\ast}$ belongs to a known compact set. Then
\begin{equation} \label{eq.glm.interpret}
\|\beta_{n}^{M}-\widetilde{\beta}_{n}^{M}\|_{2}=o(n^{-1/2}).
\end{equation}
\end{theorem}

\begin{theorem} \label{theorem.semi.interpret}
Consider the Cox model with interval-censored data. Let $\beta_{n}^{\ast}=n^{-1/2}\theta^{\ast}$ be the local alternative, $\widetilde{\beta}_{n}^{M}=\beta_{M,n}^{\ast}+\mathcal{I}_{M,M}^{-1}\mathcal{I}_{M,-M}\beta_{-M,n}^{\ast}$ be our new inference target, and
\[
(\beta_{n}^{M},\Lambda_{n}^{M})=\argmax_{(\beta_{M},\Lambda)}P_{\beta_{n}^{\ast},\Lambda^{\ast}}\ell((\beta_{M},0),\Lambda)
\]
be the minimizer of the population log likelihood under the local alternative $P_{\beta_{n}^{\ast},\Lambda^{\ast}}$ restricted to the submodel $M$. Suppose Conditions C1--C4 of \cite{li2020adaptive} hold. Then
\begin{equation} \label{eq.semi.interpret}
\|\beta_{n}^{M}-\widetilde{\beta}_{n}^{M}\|_{2}=o(n^{-1/2}).
\end{equation}
\end{theorem}

Since $\tilde{\theta}^{M}=n^{1/2}\widetilde{\beta}_{n}^{M}$ and $\theta_{n}^{M}=n^{1/2}\beta_{n}^{M}$, Theorem \ref{theorem.glm.interpret} and \ref{theorem.semi.interpret} imply that $\theta_{n}^{M}=\tilde{\theta}^{M}+o(1)$. Proof of Theorem \ref{theorem.glm.interpret} and \ref{theorem.semi.interpret} is given below, which is similar to regular asymptotic arguments, namely, verifying consistency, establishing rate of convergence and finally deriving limiting distribution \citep[page 278 of][]{wellner2013weak}. We discuss the possible generalizations of Theorem \ref{theorem.semi.interpret} to the Cox model with right-censored data and other semiparametric models in Section \ref{section.semi.interpret}.

\subsubsection{Proof of Theorem \ref{theorem.glm.interpret}}
Since we only consider the submodel $M$, to simplify our notations, we use $l(\beta_{M})$ to represent $l(\beta)$ for all those $\beta$ such that $\beta_{-M}=0$ (e.g., $l(\widetilde{\beta}_{n}^{M})$ is used to denote $l(\widetilde{\beta}_{n}^{M},0)$) and conduct similar abbreviations for $l_{\beta}(\beta)$ and $l_{\beta_{M}}(\beta)$.

We first show that $\|\beta_{n}^{M}-\widetilde{\beta}_{n}^{M}\|_{2}=o(1)$,
i.e., the consistency of $\beta_{n}^{M}$. Since $\widetilde{\beta}_{n}^{M}=O(n^{-1/2})$, it is suffice to show $\beta_{n}^{M}=o(1)$. Let $\mu$ denote the dominating measure on the sample space of the data. Note that
\begin{align*}
dP_{\beta_{n}^{\ast}}-dP_{\beta^{\ast}} & =\left\{e^{l(\beta_{n}^{\ast})}-e^{l(\beta^{\ast})}\right\}d\mu=e^{l(\beta^{\ast})}\left\{ e^{l(\beta_{n}^{\ast})-l(\beta^{\ast})}-1\right\}d\mu \\
 & =\left[e^{l_{\beta}(\beta^{\ast})^{\top}(\beta_{n}^{\ast}-\beta^{\ast})+o\left\{ l_{\beta}(\beta^{\ast})^{\top}(\beta_{n}^{\ast}-\beta^{\ast})\right\} }-1\right]dP_{\beta^{\ast}}\\
 & =\left[l_{\beta}(\beta^{\ast})^{\top}(\beta_{n}^{\ast}-\beta^{\ast})+o\left\{ l_{\beta}(\beta^{\ast})^{\top}(\beta_{n}^{\ast}-\beta^{\ast})\right\} \right]dP_{\beta^{\ast}},
\end{align*}
which implies that
\begin{equation} \label{eq.glm.interpret.expand}
 n^{1/2}\left\{ P_{\beta_{n}^{\ast}}l(\beta_{M})-P_{\beta^{\ast}}l(\beta_{M})\right\} =n^{1/2}\int l(\beta_{M})\left(dP_{\beta_{n}^{\ast}}-dP_{\beta^{\ast}}\right)=P_{\beta^{\ast}}\left\{ l(\beta_{M})l_{\beta}(\beta^{\ast})^{\top}\right\} \theta^{\ast}+\int o\left\{ l(\beta_{M})l_{\beta}(\beta^{\ast})^{\top}\right\} \theta^{\ast}dP_{\beta^{\ast}}.
\end{equation}
Due to the boundedness of covariate vector $X$ and response $Y$, for all bounded $\beta_{M}$, $l_{\beta}(\beta^{\ast})$ and $l(\beta_{M})$ are bounded almost surely. So equation (\ref{eq.glm.interpret.expand}) implies that the remainder term
\[
\int o\left\{ l(\beta_{M})l_{\beta}(\beta^{\ast})^{\top}\right\} \theta^{\ast}dP_{\beta^{\ast}}
\]
converges to $0$ uniformly for all bounded $\beta_{M}$. Thus,
\[
n^{1/2}\left\{ P_{\beta_{n}^{\ast}}l(\beta_{M})-P_{\beta^{\ast}}l(\beta_{M})\right\} =n^{1/2}\int l(\beta_{M})\left(dP_{\beta_{n}^{\ast}}-dP_{\beta^{\ast}}\right)=P_{\beta^{\ast}}\left\{ l(\beta_{M})l_{\beta}(\beta^{\ast})^{\top}\right\} \theta^{\ast}+o(1),
\]
which implies that $P_{\beta_{n}^{\ast}}l(\beta_{M})$ converges to $P_{\beta^{\ast}}l(\beta_{M})$ uniformly for all bounded $\beta_{M}$. We let $\mathbb{M}_{n}(\beta_{M})=P_{\beta_{n}^{\ast}}l(\beta_{M})$ and $\mathbb{M}(\beta_{M})=P_{\beta^{\ast}}l(\beta_{M})$. Since the covariance matrix of $X$ is positive definite, $\beta_{M}^{\ast}$ is the unique maximizer of $\mathbb{M}(\beta_{M})$. Furthermore, $P_{\beta^{\ast}}l(\beta_{M})$ is continuous in $\beta_{M}$, and so $\beta_{M}^{\ast}$ is a well-separated maximum of $\mathbb{M}(\beta_{M})$ \citep[page 46 of][]{van2000asymptotic}. Applying Corollary 3.2.3(i) on \cite{wellner2013weak}, we see that $\|\beta_{n}^{M}-\beta_{M}^{\ast}\|_{2}=o(1)$. This is our desired result since $\beta_{M}^{\ast}=0$.

To derive (\ref{eq.glm.interpret}), we note that
\begin{align} \label{eq.glm.interpret.1}
0=n^{1/2}P_{\beta_{n}^{\ast}}l_{\beta_{M}}(\beta_{n}^{M}) & =n^{1/2}\left\{ P_{\beta_{n}^{\ast}}l_{\beta_{M}}(\beta_{M}^{\ast})-P_{\beta^{\ast}}l_{\beta_{M}}(\beta_{M}^{\ast})\right\} +n^{1/2}P_{\beta_{n}^{\ast}}\left\{ l_{\beta_{M}}(\beta_{n}^{M})-l_{\beta_{M}}(\beta_{M}^{\ast})\right\} .
\end{align}
By similar arguments to those below equation (\ref{eq.glm.interpret.expand}), the first term on the right-hand side is
\begin{align} \label{eq.glm.interpret.2}
n^{1/2}\left\{ P_{\beta_{n}^{\ast}}l_{\beta_{M}}(\beta_{M}^{\ast})-P_{\beta^{\ast}}l_{\beta_{M}}(\beta_{M}^{\ast})\right\}  & =P_{\beta^{\ast}}\left\{ l_{\beta_{M}}(\beta_{M}^{\ast})l_{\beta}(\beta^{\ast})^{\top}\right\} \theta^{\ast}+o(1)\nonumber \\
 & =n^{1/2}\left(\mathcal{I}_{M,M}\beta_{M,n}^{\ast}+\mathcal{I}_{M,-M}\beta_{-M,n}^{\ast}\right)+o(1)=n^{1/2}\mathcal{I}_{M,M}\widetilde{\beta}_{n}^{M}+o(1).
\end{align}
For the second term on the right-hand side, by Taylor expansion, we know that under $P_{\beta^{\ast}}$,
\begin{equation} \label{eq.glm.interpret.3}
\|\beta_{n}^{M}-\beta_{M}^{\ast}\|_{2}^{-1}\left\{ l_{\beta_{M}}(\beta_{n}^{M})-l_{\beta_{M}}(\beta_{M}^{\ast})-l_{\beta_{M}\beta_{M}}(\beta^{\ast})(\beta_{n}^{M}-\beta_{M}^{\ast})\right\} =o_{P}(1).
\end{equation}
So by Le Cam's third lemma, the left-hand side in (\ref{eq.glm.interpret.3}) is also $o_{P}(1)$ under $P_{\beta_{n}^{\ast}}$. Thus,
\[
\|\beta_{n}^{M}-\beta_{M}^{\ast}\|_{2}^{-1}P_{\beta_{n}^{\ast}}\left\{ l_{\beta_{M}}(\beta_{n}^{M})-l_{\beta_{M}}(\beta_{M}^{\ast})-l_{\beta_{M}\beta_{M}}(\beta^{\ast})(\beta_{n}^{M}-\beta_{M}^{\ast})\right\} =o(1),
\]
which implies that
\begin{equation} \label{eq.glm.interpret.4}
\|\beta_{n}^{M}-\beta_{M}^{\ast}\|_{2}^{-1}P_{\beta_{n}^{\ast}}\left\{ l_{\beta_{M}}(\beta_{n}^{M})-l_{\beta_{M}}(\beta_{M}^{\ast})\right\} =\|\beta_{n}^{M}-\beta_{M}^{\ast}\|_{2}^{-1}P_{\beta_{n}^{\ast}}l_{\beta_{M}\beta_{M}}(\beta^{\ast})(\beta_{n}^{M}-\beta_{M}^{\ast})+o(1).
\end{equation}
In addition, for the $j$th column of $l_{\beta_{M}\beta_{M}}(\beta^{\ast})$,
denoted as $l_{\beta_{M}\beta_{M},j}(\beta^{\ast})$ ($j=1,\dots,p$),
we have by similar arguments to those below \eqref{eq.glm.interpret.expand} that
\begin{equation} \label{eq.glm.interpret.5}
P_{\beta_{n}^{\ast}}l_{\beta_{M}\beta_{M},j}(\beta^{\ast})-P_{\beta^{\ast}}l_{\beta_{M}\beta_{M},j}(\beta^{\ast})=P_{\beta^{\ast}}\left\{ l_{\beta_{M}\beta_{M},j}(\beta^{\ast})l_{\beta}(\beta^{\ast})^{\top}\right\} (\beta_{n}^{\ast}-\beta_{M}^{\ast})+o(1)=o(1).
\end{equation}
Combining (\ref{eq.glm.interpret.4}) and (\ref{eq.glm.interpret.5}), we obtain
\begin{equation} \label{eq.glm.interpret.6}
\|\beta_{n}^{M}-\beta_{M}^{\ast}\|_{2}^{-1}P_{\beta_{n}^{\ast}}\left\{ l_{\beta_{M}}(\beta_{n}^{M})-l_{\beta_{M}}(\beta_{M}^{\ast})\right\} =\|\beta_{n}^{M}-\beta_{M}^{\ast}\|_{2}^{-1}P_{\beta^{\ast}}l_{\beta_{M}\beta_{M}}(\beta^{\ast})(\beta_{n}^{M}-\beta_{M}^{\ast})+o(1).
\end{equation}
Putting (\ref{eq.glm.interpret.1}), (\ref{eq.glm.interpret.2}) and (\ref{eq.glm.interpret.6}) together and recalling that $\beta_{M}^{\ast}=0$, we have
\[
n^{1/2}\mathcal{I}_{M,M}\widetilde{\beta}_{n}^{M}+o(1)=n^{1/2}\mathcal{I}_{M,M}\beta_{n}^{M}+o(n^{1/2}\|\beta_{n}^{M}\|_{2})
\]
which implies that $\beta_{n}^{M}=O(n^{-1/2})$ and thus $n^{1/2}(\beta_{n}^{M}-\widetilde{\beta}_{n}^{M})=o(1)$ as $\mathcal{I}_{M,M}$ is positive definite.
\\ \\
\begin{remark}
To provide more insights into $\widetilde{\theta}^{M}$ and $\widetilde{\beta}_{n}^{M}$, we show that in the Gaussian linear model $Y=X\beta^{\ast}+\epsilon$,
\[
\widetilde{\beta}^{M}=\beta_{M}^{\ast}+\mathcal{I}_{M,M}^{-1}\mathcal{I}_{M,-M}\beta_{-M}^{\ast},\quad\text{and }\beta^{M}=\argmin_{b^{M}}E\|Y-X_{M}b^{M}\|_{2}^{2}
\]
are exactly the same for both the fixed design case and the random design case (assuming that the covariates have mean $0$), which indicates that our $\widetilde{\beta}_{n}^{M}$ is exactly the same as the inference target $\beta^{M}$ in equation (1.2) of \cite{lee2016exact}. Here $Y$ is the response, $\beta^{\ast}$ is the true regression coefficient,
and $\epsilon$ is the Gaussian random error. The information matrix $\mathcal{I}=X^{\top}X$ in fixed design case and $\mathcal{I}=\text{cov}(X)$ in random design case. For either case, applying the KKT condition to
\[
\beta^{M}=\argmin_{b^{M}}E\|Y-X_{M}b^{M}\|_{2}^{2}
\]
leads to
\[
\mathcal{I}_{M,M}\beta^{M}=\mathcal{I}_{M,M}\beta_{M}^{\ast}+\mathcal{I}_{M,-M}\beta_{-M}^{\ast},
\]
which implies that $\widetilde{\beta}^{M}=\beta^{M}$.
\end{remark}

\subsubsection{Proof of Theorem \ref{theorem.semi.interpret}} \label{section.semi.interpret}

Similar to the proof of Theorem \ref{theorem.glm.interpret}, we use $\ell(\beta_{M},\Lambda)$ to represent $\ell\left\{ (\beta_{M}^{\top},0^{\top})^{\top},\Lambda\right\} $
for all those $\beta$'s such that $\beta_{-M}=0$ and conduct similar abbreviations of $\ell_{\beta}(\beta,\Lambda)$, $\ell_{\Lambda}(\beta,\Lambda)$, etc.

Our first step is to show $\|\beta_{n}^{M}-\widetilde{\beta}_{n}^{M}\|_{2}=o(1)$. Since $\widetilde{\beta}_{n}^{M}=O(n^{-1/2})$ and $\beta_{M}^{\ast}=0$, it is suffice to show $\|\beta_{n}^{M}-\beta_{M}^{\ast}\|_{2}=o(1)$. Following \cite{zeng2017maximum}, let
\[
m(\beta_{M},\Lambda)=\log\frac{e^{\ell(\beta_{M},\Lambda)}+e^{\ell(\beta_{M}^{\ast},\Lambda^{\ast})}}{2}.
\]
We have
\[
P_{\beta_{n}^{\ast},\Lambda^{\ast}}m(\beta_{M},\Lambda)-P_{\beta^{\ast},\Lambda^{\ast}}m(\beta_{M},\Lambda)=\int m(\beta_{M},\Lambda)\left(dP_{\beta_{n}^{\ast},\Lambda^{\ast}}-dP_{\beta^{\ast},\Lambda^{\ast}}\right)=\int m(\beta_{M},\Lambda)\left(e^{\ell(\beta_{n}^{\ast},\Lambda^{\ast})}-e^{\ell(\beta^{\ast},\Lambda^{\ast})}\right)d\mu,
\]
where $\mu$ is the dominating measure on the sample space of the data. Note that the log-likelihood for one subject
\[
\ell(\beta_{M},\Lambda)=\log\left[\exp\left\{ -\Lambda(L)\exp(\beta_{M}^{\top}X_{M})\right\} -\exp\left\{ -\Lambda(R)\exp(\beta_{M}^{\top}X_{M})\right\} \right]
\]
is always nonpositive, where $L$ and $R$ are the inspection times bracketing the failture time and $X_{M}$ is the covariate vector in the submodel $M$. In addition, $m(\beta_{M},\Lambda)$ is also bounded below by $\log\left\{ e^{\ell(\beta_{M}^{\ast},\Lambda^{\ast})}/2\right\} $, which is in turn bounded below by a constant. Also, $\ell(\beta_{n}^{\ast},\Lambda^{\ast})$ and $\ell(\beta^{\ast},\Lambda^{\ast})$ are bounded, and $\ell(\beta_{n}^{\ast},\Lambda^{\ast})$ converges to $\ell(\beta^{\ast},\Lambda^{\ast})$. Thus, $P_{\beta_{n}^{\ast},\Lambda^{\ast}}m(\beta_{M},\Lambda)$ converges uniformly to $P_{\beta^{\ast},\Lambda^{\ast}}m(\beta_{M},\Lambda)$ for all $\Lambda$ and bounded $\beta_{M}$ as $n\to\infty$. Then $\|\beta_{n}^{M}-\beta_{M}^{\ast}\|_{2}=o(1)$ and $\sup_{t\in[\zeta,\tau]}|\Lambda_{n}^{M}(t)-\Lambda^{\ast}(t)|\to0$ come from adapting proof of Theorem 1 in \cite{zeng2017maximum} by replacing all $\mathbb{P}_{n}$ with $P_{\beta_{n}^{\ast},\Lambda^{\ast}}$.

Now our second step is to establish the rate of convergence of $(\beta_{n}^{M},\Lambda_{n}^{M})$ to $(\beta_{M}^{\ast},\Lambda^{\ast})$. To simplify our notations, let $p(\beta_{M},\Lambda)=e^{\ell(\beta_{M},\Lambda)}$. Since $(\beta_{n}^{M},\Lambda_{n}^{M})$ is consistent for $(\beta_{M}^{\ast},\Lambda^{\ast})$, we can assume that for sufficiently large $n$, we have
\[
\frac{1}{2}\le\frac{p(\beta_{n}^{M},\Lambda_{n}^{M})}{p(\beta_{M}^{\ast},\Lambda^{\ast})}\le2.
\]
Note that
\begin{align*}
dP_{\beta_{n}^{\ast},\Lambda^{\ast}}-dP_{\beta^{\ast},\Lambda^{\ast}} & =\left\{e^{\ell(\beta_{n}^{\ast},\Lambda^{\ast})}-e^{\ell(\beta^{\ast},\Lambda^{\ast})}\right\}d\mu=e^{\ell(\beta^{\ast},\Lambda^{\ast})}\left\{ e^{\ell(\beta_{n}^{\ast},\Lambda^{\ast})-\ell(\beta^{\ast},\Lambda^{\ast})}-1\right\}d\mu \\
 & =\left[e^{\ell_{\beta}(\beta^{\ast},\Lambda^{\ast})^{\top}(\beta_{n}^{\ast}-\beta^{\ast})+o\left\{ \ell_{\beta}(\beta^{\ast},\Lambda^{\ast})^{\top}(\beta_{n}^{\ast}-\beta^{\ast})\right\} }-1\right]dP_{\beta^{\ast},\Lambda^{\ast}}\\
 & =\left[\ell_{\beta}(\beta^{\ast},\Lambda^{\ast})^{\top}(\beta_{n}^{\ast}-\beta^{\ast})+o\left\{ \ell_{\beta}(\beta^{\ast},\Lambda^{\ast})^{\top}(\beta_{n}^{\ast}-\beta^{\ast})\right\} \right]dP_{\beta^{\ast},\Lambda^{\ast}}.
\end{align*}
For those $(\beta_{M},\Lambda)$ such that $1/2\le p(\beta_{M},\Lambda)/p(\beta_{M}^{\ast},\Lambda^{\ast})\le2$, by similar arguments to those below (\ref{eq.glm.interpret.expand}),
\[
n^{1/2}\left\{ P_{\beta_{n}^{\ast},\Lambda^{\ast}}\ell(\beta_{M},\Lambda)-P_{\beta^{\ast},\Lambda^{\ast}}\ell(\beta_{M},\Lambda)\right\} =n^{1/2}\int\ell(\beta_{M},\Lambda)\left(dP_{\beta_{n}^{\ast},\Lambda^{\ast}}-dP_{\beta^{\ast},\Lambda^{\ast}}\right)=P_{\beta^{\ast},\Lambda^{\ast}}\left\{ \ell(\beta_{M},\Lambda)\ell_{\beta}(\beta^{\ast},\Lambda^{\ast})^{\top}\right\} \theta^{\ast}+o(1).
\]
Let $\mathbb{M}_{n}(\beta_{M},\Lambda)=P_{\beta_{n}^{\ast},\Lambda^{\ast}}\ell(\beta_{M},\Lambda)$
and $\mathbb{M}(\beta_{M},\Lambda)=P_{\beta^{\ast},\Lambda^{\ast}}\ell(\beta_{M},\Lambda)$. For any sufficiently small $\delta>0$, we consider
\[
(\breve{\beta}_{n},\breve{\Lambda}_{n})=\argmax_{(\beta,\Lambda):H\left\{(\beta_M,\Lambda),(\beta_{M}^{\ast},\Lambda^{\ast})\right\} \le\delta,1/2\le p(\beta_{M},\Lambda)/p(\beta_{M}^{\ast},\Lambda^{\ast})\le2}n^{1/2}|(\mathbb{M}_{n}-\mathbb{M})(\beta_{M},\Lambda)-(\mathbb{M}_{n}-\mathbb{M})(\beta_{M}^{\ast},\Lambda^{\ast})|,
\]
where
\[
H\left\{ (\beta_M,\Lambda),(\beta_{M}^{\ast},\Lambda^{\ast})\right\} =\left[\int\left\{ \sqrt{p(\beta_M,\Lambda)/p(\beta_{M}^{\ast},\Lambda^{\ast})}-1\right\} ^{2}dP_{\beta^{\ast},\Lambda^{\ast}}\right]^{1/2}
\]
is the Hellinger distance. Then
\begin{align*}
n^{1/2}|(\mathbb{M}_{n}-\mathbb{M})(\breve{\beta}_{M,n},\breve{\Lambda}_{n})-(\mathbb{M}_{n}-\mathbb{M})(\beta_{M}^{\ast},\Lambda^{\ast})| & =n^{1/2}\bigg|P_{\beta_{n}^{\ast},\Lambda^{\ast}}\ell(\breve{\beta}_{M,n},\breve{\Lambda}_{n})-P_{\beta^{\ast},\Lambda^{\ast}}\ell(\breve{\beta}_{M,n},\breve{\Lambda}_{n})\\
 & \quad\quad\quad\quad-P_{\beta_{n}^{\ast},\Lambda^{\ast}}\ell(\beta_{M}^{\ast},\Lambda^{\ast})+P_{\beta^{\ast},\Lambda^{\ast}}\ell(\beta_{M}^{\ast},\Lambda^{\ast})\bigg|\\
 & =\bigg|P_{\beta^{\ast},\Lambda^{\ast}}\left[\left\{ \ell(\breve{\beta}_{M,n},\breve{\Lambda}_{n})-\ell(\beta_{M}^{\ast},\Lambda^{\ast})\right\} \ell_{\beta}(\beta^{\ast},\Lambda^{\ast})^{\top}\right]\theta^{\ast}\bigg|+o(1).
\end{align*}
Due to boundedness of $\ell_{\beta}(\beta^{\ast},\Lambda^{\ast})$,
\begin{align*}
\bigg|P_{\beta^{\ast},\Lambda^{\ast}}\left[\left\{ \ell(\breve{\beta}_{M,n},\breve{\Lambda}_{n})-\ell(\beta_{M}^{\ast},\Lambda^{\ast})\right\} \ell_{\beta}(\beta^{\ast},\Lambda^{\ast})^{\top}\right]\theta^{\ast}\bigg| & \lesssim P_{\beta^{\ast},\Lambda^{\ast}}|\ell(\breve{\beta}_{M,n},\breve{\Lambda}_{n})-\ell(\beta_{M}^{\ast},\Lambda^{\ast})|\\
 & =\int\bigg|\log\left\{ p(\breve{\beta}_{M,n},\breve{\Lambda}_{n})/p(\beta_{M}^{\ast},\Lambda^{\ast})\right\} \bigg|dP_{\beta^{\ast},\Lambda^{\ast}}\\
 & \lesssim\int\bigg|p(\breve{\beta}_{M,n},\breve{\Lambda}_{n})/p(\beta_{M}^{\ast},\Lambda^{\ast})-1\bigg|dP_{\beta^{\ast},\Lambda^{\ast}}\\
 & \lesssim H\left\{ (\breve{\beta}_{M,n},\breve{\Lambda}_{n}),(\beta_{M}^{\ast},\Lambda^{\ast})\right\} \\
 & \le \delta,
\end{align*}
where the second inequality comes from $|\log t|\le2|t-1|$ for all $t\in[1/2,2]$, and the third inequality is a result that the total variation distance is bounded above by the Hellinger distance \citep[Page 212 of][]{van2000asymptotic}. In addition, by Lemma 1.3 of \citet{geer2000empirical},
\[
\mathbb{M}(\beta_{M},\Lambda)-\mathbb{M}(\beta_{M}^{\ast},\Lambda^{\ast})\lesssim -H^2\left\{ (\beta_M,\Lambda),(\beta_{M}^{\ast},\Lambda^{\ast})\right\}
\]

The above results in conjunction with the ``consistency" of $(\beta^{M}_{n},\Lambda^{M}_{n})$ and the fact that $(\beta^{M}_{n},\Lambda^{M}_{n})$ maximizes $\mathbb{M}_n(\beta_M,\Lambda)$ imply that all the conditions of
 Theorem 3.4.1 in \cite{wellner2013weak} are satisfied when the $r_n$ and $\delta_n$ therein are chosen to be $n^{1/2}$ and $n^{-1/2}$ respectively. Therefore,
\[
H\left\{ (\beta_{n}^{M},\Lambda_{n}^{M}),(\beta_{M}^{\ast},\Lambda^{\ast})\right\} =O(n^{-1/2}).
\]
Then by similar arguments to the proof of Lemma A1 in \cite{zeng2016maximum}, we have
\[
E_{\beta^{\ast},\Lambda^{\ast}}\left(\sum_{k=1}^K\left[\Lambda^{M}_n(U_k)\exp(\beta_{n}^{M\top}X_M)-\Lambda^{\ast}(U_k)\exp(\beta_M^{\ast\top}X_M)\right]^2\right)=O(n^{-1}).
\]

The final step is to derive $\|\beta_{n}^{M}-\widetilde{\beta}_{n}^{M}\|_{2}=o(n^{-1/2}).$ By similar arguments to the proof of Theorem 2 in \cite{zeng2016maximum}, we have
\begin{align}
&n^{1/2}\left\{P_{\beta_{n}^{\ast},\Lambda^{\ast}}\ell_{\beta_{M}}  (\beta_{n}^{M},\Lambda_{n}^{M})-P_{\beta^{\ast},\Lambda^{\ast}}\ell_{\beta_{M}}  (\beta_{n}^{M},\Lambda_{n}^{M})\right\} \nonumber\\
=&-n^{1/2}P_{\beta^{\ast},\Lambda^{\ast}}\ell_{\beta_{M}}  (\beta_{n}^{M},\Lambda_{n}^{M}) \nonumber\\
=&-n^{1/2}P_{\beta^{\ast},\Lambda^{\ast}}\left[ \ell_{\beta_{M}\beta_{M}}(\beta_{n}^{M}-\beta_{M}^{\ast})+\ell_{\beta_{M}\Lambda}\left\{ d(\Lambda_{n}^{M}-\Lambda^{\ast})/d\Lambda^{\ast}\right\} \right] \nonumber\\
&+n^{1/2}O\left\{\|\beta_{n}^{M}-\beta_{M}^{\ast}\|^2+E_{\beta^{\ast},\Lambda^{\ast}}\left(\sum_{k=1}^K\left[\Lambda^{M}_n(U_k)\exp(\beta_{n}^{M\top}X_M)-\Lambda^{\ast}(U_k)\exp(\beta_M^{\ast\top}X_M)\right]^2\right)\right\} \nonumber\\
=&-n^{1/2}P_{\beta^{\ast},\Lambda^{\ast}}\left[ \ell_{\beta_{M}\beta_{M}}(\beta_{n}^{M}-\beta_{M}^{\ast})+\ell_{\beta_{M}\Lambda}\left\{ d(\Lambda_{n}^{M}-\Lambda^{\ast})/d\Lambda^{\ast}\right\} \right]+O\left(n^{1/2}\|\beta_{n}^{M}-\beta_{M}^{\ast}\|^2+n^{-1/2}\right) \label{eq.semi.interpret.a}
\end{align}
and
\begin{align}
&n^{1/2}\left\{P_{\beta_{n}^{\ast},\Lambda^{\ast}}\ell_{\Lambda}  (\beta_{n}^{M},\Lambda_{n}^{M})(h_{M}^{\ast})-P_{\beta^{\ast},\Lambda^{\ast}}\ell_{\Lambda}  (\beta_{n}^{M},\Lambda_{n}^{M})(h_{M}^{\ast})\right\} \nonumber\\
=&-n^{1/2}P_{\beta^{\ast},\Lambda^{\ast}}\ell_{\Lambda}  (\beta_{n}^{M},\Lambda_{n}^{M})(h_{M}^{\ast}) \nonumber\\
=&-n^{1/2}P_{\beta^{\ast},\Lambda^{\ast}}\left[ \ell_{\Lambda\beta_{M}}(h_{M}^{\ast})(\beta_{n}^{M}-\beta_{M}^{\ast})+\ell_{\Lambda\Lambda}\left(h_{M}^{\ast}, d(\Lambda_{n}^{M}-\Lambda^{\ast})/d\Lambda^{\ast}\right) \right] \nonumber\\
&+n^{1/2}O\left\{\|\beta_{n}^{M}-\beta_{M}^{\ast}\|^2+E_{\beta^{\ast},\Lambda^{\ast}}\left(\sum_{k=1}^K\left[\Lambda^{M}_n(U_k)\exp(\beta_{n}^{M\top}X_M)-\Lambda^{\ast}(U_k)\exp(\beta_M^{\ast\top}X_M)\right]^2\right)\right\} \nonumber\\
=&-n^{1/2}P_{\beta^{\ast},\Lambda^{\ast}}\left[ \ell_{\Lambda\beta_{M}}(h_{M}^{\ast})(\beta_{n}^{M}-\beta_{M}^{\ast})+\ell_{\Lambda\Lambda}\left(h_{M}^{\ast}, d(\Lambda_{n}^{M}-\Lambda^{\ast})/d\Lambda^{\ast}\right) \right]+O\left(n^{1/2}\|\beta_{n}^{M}-\beta_{M}^{\ast}\|^2+n^{-1/2}\right). \label{eq.semi.interpret.b}
\end{align}
Subtracting \eqref{eq.semi.interpret.b} from \eqref{eq.semi.interpret.a} yields
\begin{equation}\label{asymp_rep_beta_n_M}
 n^{1/2}(P_{\beta_{n}^{\ast},\Lambda^{\ast}}-P_{\beta^{\ast},\Lambda^{\ast}})\{\ell_{\beta_{M}}  (\beta_{n}^{M},\Lambda_{n}^{M})-\ell_{\Lambda}  (\beta_{n}^{M},\Lambda_{n}^{M})(h_{M}^{\ast})\}=n^{1/2}\mathcal{I}_{M,M}(\beta_{n}^{M}-\beta_{M}^{\ast})+O\left(n^{1/2}\|\beta_{n}^{M}-\beta_{M}^{\ast}\|^2+n^{-1/2}\right).
\end{equation}
Since $\mathcal{I}_{M,M}$ is positive definite, equation \eqref{asymp_rep_beta_n_M} entails $\beta_{n}^{M}-\beta_{M}^{\ast}=O(n^{-1/2})$ and thus further yields
\begin{equation}\label{asymp_rep_beta_n_M_o1}
 n^{1/2}(P_{\beta_{n}^{\ast},\Lambda^{\ast}}-P_{\beta^{\ast},\Lambda^{\ast}})\{\ell_{\beta_{M}}  (\beta_{n}^{M},\Lambda_{n}^{M})-\ell_{\Lambda}  (\beta_{n}^{M},\Lambda_{n}^{M})(h_{M}^{\ast})\}=n^{1/2}\mathcal{I}_{M,M}(\beta_{n}^{M}-\beta_{M}^{\ast})+o(1).
\end{equation}

Next, by similar arguments to the proof of Theorem \ref{theorem.glm.interpret}, we have
\begin{align}
n^{1/2}\left\{ P_{\beta_{n}^{\ast},\Lambda^{\ast}}\ell_{\beta_{M}}(\beta_{M}^{\ast},\Lambda^{\ast})-P_{\beta^{\ast},\Lambda^{\ast}}\ell_{\beta_{M}}(\beta_{M}^{\ast},\Lambda^{\ast})\right\}  & =n^{1/2}P_{\beta^{\ast},\Lambda^{\ast}}\left\{ \ell_{\beta_{M}}(\beta_{M}^{\ast},\Lambda^{\ast})\ell_{\beta}(\beta^{\ast},\Lambda^{\ast})^{\top}\right\} \beta_{n}^{\ast}+o(1),\label{eq.semi.interpret.3}\\
n^{1/2}\left\{ P_{\beta_{n}^{\ast},\Lambda^{\ast}}\ell_{\Lambda}(\beta_{M}^{\ast},\Lambda^{\ast})(h_{M}^{\ast})-P_{\beta^{\ast},\Lambda^{\ast}}\ell_{\Lambda}(\beta_{M}^{\ast},\Lambda^{\ast})(h_{M}^{\ast})\right\}  & =n^{1/2}P_{\beta^{\ast},\Lambda^{\ast}}\left\{ \ell_{\Lambda}(\beta_{M}^{\ast},\Lambda^{\ast})(h_{M}^{\ast})\ell_{\beta}(\beta^{\ast},\Lambda^{\ast})^{\top}\right\} \beta_{n}^{\ast}+o(1). \label{eq.semi.interpret.4}
\end{align}
Subtracting (\ref{eq.semi.interpret.4}) from (\ref{eq.semi.interpret.3}) gives us
\begin{align}
&n^{1/2}(P_{\beta_{n}^{\ast},\Lambda^{\ast}}-P_{\beta^{\ast},\Lambda^{\ast}})\{\ell_{\beta_{M}}  (\beta_{M}^{\ast},\Lambda^{\ast})-\ell_{\Lambda}  (\beta_{M}^{\ast},\Lambda^{\ast})(h_{M}^{\ast})\} \nonumber \\
 =& n^{1/2}P_{\beta^{\ast},\Lambda^{\ast}}\left\{ \tilde{\ell}_{M}(\beta_{M}^{\ast})\ell_{\beta}(\beta^{\ast},\Lambda^{\ast})^{\top}\right\} \beta_{n}^{\ast}+o(1)\nonumber \\
 =& n^{1/2}\left(\mathcal{I}_{M,M}\beta_{M,n}^{\ast}+\mathcal{I}_{M,-M}\beta_{-M,n}^{\ast}\right)+o(1)=n^{1/2}\mathcal{I}_{M,M}\widetilde{\beta}_{n}^{M}+o(1), \label{eq.semi.interpret.5}
\end{align}
where we have used a property of efficient score function $\tilde{\ell}(\beta^{\ast})$
that $P_{\beta^{\ast},\Lambda^{\ast}}\left\{ \tilde{\ell}(\beta^{\ast})\ell_{\beta}(\beta^{\ast},\Lambda^{\ast})^{\top}\right\} =P_{\beta^{\ast},\Lambda^{\ast}}\left\{ \tilde{\ell}(\beta^{\ast})^{\otimes2}\right\} =\mathcal{I}$.

Further subtracting \eqref{eq.semi.interpret.5} from \eqref{asymp_rep_beta_n_M_o1} yields
\begin{align}
&n^{1/2}(P_{\beta_{n}^{\ast},\Lambda^{\ast}}-P_{\beta^{\ast},\Lambda^{\ast}})\{\ell_{\beta_{M}}  (\beta_{n}^{M},\Lambda_{n}^{M})-\ell_{\Lambda}  (\beta_{n}^{M},\Lambda_{n}^{M})(h_{M}^{\ast})-\ell_{\beta_{M}}  (\beta_{M}^{\ast},\Lambda^{\ast})+\ell_{\Lambda}  (\beta_{M}^{\ast},\Lambda^{\ast})(h_{M}^{\ast})\} \nonumber \\
=&n^{1/2}\mathcal{I}_{M,M}(\beta_{n}^{M}-\widetilde{\beta}_{n}^{M})+o(1) \label{distance_bt_two_betas}.
\end{align}
By similar arguments to those below \eqref{eq.glm.interpret.expand} and the ``consistency" of $(\beta_{n}^{M},\Lambda_{n}^{M})$, equation \eqref{distance_bt_two_betas} implies
\[
n^{1/2}\mathcal{I}_{M,M}(\widetilde{\beta}_{n}^{M}-\beta_{n}^{M})=o(1),
\]
which in turn implies $\|\widetilde{\beta}_{n}^{M}-\beta_{n}^{M}\|_{2}=o(n^{-1/2})$
as $\mathcal{I}_{M,M}$ is positive definite.
\\ \\
\begin{remark}
The rate of convergence of $(\beta_{n}^{M},\Lambda_{n}^{M})$ to $(\beta^{\ast},\Lambda^{\ast})$ is $O(n^{-1/2})$ instead of the classic semiparametric rate $O(n^{-1/3})$ in the Cox model with interval-censored data \citep[][Page 548]{huang1996efficient}, possibly because   $(\beta_{n}^{M},\Lambda_{n}^{M})$ stems from $P_{\beta_{n}^{\ast},\Lambda^{\ast}}$ instead of $\mathbb{P}_{n}$. Also due to this reason, neither the uniform entropy integral nor the bracketing entropy integral can be used to bound the modulus of continuity of the processes $\sqrt{n}(\mathbb{M}_n-\mathbb{M})(\beta_M,\Lambda)$ . Our proof takes advantage of the fact
that the Kullback--Leibler divergence, total variation distance and Hellinger distance are all special cases of $f$-divergence \citep{sason2016f}, while enforcing the upper and lower bounds on $p(\beta_M,\Lambda)/p(\beta_{M}^{\ast},\Lambda^{\ast})$ is inspired by the proof of Theorem
3.4.4 in \cite{wellner2013weak}.

Our proof of rate of convergence and the arguments of ``deriving the limiting distribution'', i.e., showing $\|\widetilde{\beta}_{n}^{M}-\beta_{n}^{M}\|_{2}=o(n^{-1/2})$, are general and can be applied to other semiparametric models, including the Cox model with right-censored data.

The proof of ``consistency'', namely $\|\beta_{n}^{M}-\beta_{M}^{\ast}\|_{2}=o(1)$,
appears to be more tricky than showing the same result in GLM. One reason is that the log likelihood $\ell(\beta,\Lambda)$ is unbounded. Here we follow \cite{zeng2017maximum} to use
\[
m(\beta,\Lambda)=\log\left\{\frac{e^{\ell(\beta,\Lambda)}+e^{\ell(\beta^{\ast},\Lambda^{\ast})}}{2}\right\}.
\]
See also \cite{huang1996efficient} on how the unboundedness issue in the Cox model with interval-censored can be tackled differently from the approach of \cite{zeng2017maximum}. The other issue is the difficulty of proving that $(\beta^{\ast},\Lambda^{\ast})$ is a ``well-separated" maximum of $\mathbb{M}(\beta,\Lambda)$ in order to use Corollary 3.2.3 (i) of \citet{wellner2013weak}. In the GLM case, this is easy to show under the assumption that $\beta^{\ast}$ belongs to a known compact set \citep[see Page 46 of][]{van2000asymptotic}. However, the compactness of the space of $\Lambda$ is nontrivial, as the unit ball in $C[\zeta,\tau]$ (the space of continuous functions defined over $[\zeta,\tau]$) relative to the $L_{\infty}$ norm (metric of uniform convergence) is not compact \citep[][Theorem 6.5]{brezis2011functional}. One way to bypass the difficulty is following  \cite{zeng2017maximum} to apply Helly\textquoteright s selection theorem \citep[][Lemma 2.5]{van2000asymptotic} to show the pointwise convergence of the nonparametric maximum likelihood estimator for $\Lambda$ and then strengthens the result to the uniform convergence given that $\Lambda^{\ast}(t)$ is continuous. This is what we did.
 Another way is to take advantage of vague topology and the Banach--Alaoglu theorem \citep[][Theorem 3.16]{brezis2011functional}, as done in \citet{kiefer1956consistency}, \citet{van1992existence} and \citet{huang1996efficient}.  For proving $\|\beta_{n}^{M}-\beta_{M}^{\ast}\|_{2}=o(1)$ in the Cox model with right-censored data, we can benefit from the explicit formulation of the profile likelihood and adapt our proof for GLM. In general, proving the ``consistency" for a semiparametric model is a case-by-case task,  but we believe it is closely related to proving the consistency of maximum likelihood estimator.

\end{remark}

\FloatBarrier
\subsection{Additional simulation results} \label{section.additional.sim}
\begin{table}
\centering
  \caption{Simulation results on the coverage of confidence intervals for $\beta=1$, AIC}
  \begin{threeparttable}
    \begin{tabular}{cccccccccc}
    \hline
    \multirow{2}[4]{*}{Coverage} & \multicolumn{4}{c}{$n = 200$} &       & \multicolumn{4}{c}{$n = 400$} \\
\cmidrule{2-5} \cmidrule{7-10}
 & $\beta_1$ & $\beta_2$ & $\beta_9$ & $\beta_{10}$ &       & $\beta_1$ & $\beta_2$ & $\beta_9$ & $\beta_{10}$ \\ \hline
    \multicolumn{4}{l}{\textit{PRES+LS}}      &       &       &       &       &  \\
    $\epsilon=10^{-2}$ & 0.99 & 0.98 & 0.975 & 0.975 &       & 0.965 & 0.955 & 0.97 & 0.97 \\
    $\epsilon=10^{-5}$ & 0.995 & 0.98 & 0.97 & 0.975 &       & 0.97 & 0.945 & 0.97 & 0.965 \\
    $\epsilon=10^{-7}$ & 0.99 & 0.98 & 0.97 & 0.975 &       & 0.965 & 0.945 & 0.965 & 0.955 \\[5pt]
    \multicolumn{4}{l}{\textit{sPRES+LS}} &       &       &       &       &  \\
    $\epsilon=10^{-2}$ & 0.995 & 0.98 & 0.975 & 0.975 &      & 0.97 & 0.955 & 0.975 & 0.97 \\
    $\epsilon=10^{-5}$ & 0.995 & 0.98 & 0.975 & 0.975 &      & 0.97 & 0.955 & 0.975 & 0.97 \\
    $\epsilon=10^{-7}$ & 0.995 & 0.98 & 0.975 & 0.975 &       & 0.97 & 0.955 & 0.97 & 0.97 \\[5pt]
    \multicolumn{4}{l}{\textit{PRES}} &       &       &       &       &  \\
    $\epsilon=10^{-2}$ & 0.955 & 0.924 & 0.889 & 0.934 &       & 0.93 & 0.91 & 0.955 & 0.945 \\
    $\epsilon=10^{-5}$ & 0.958 & 0.927 & 0.901 & 0.937 &       & 0.949 & 0.933 & 0.964 & 0.959 \\
    $\epsilon=10^{-7}$ & 0.967 & 0.918 & 0.896 & 0.94 &       & 0.958 & 0.942 & 0.974 & 0.948 \\[5pt]
    \multicolumn{4}{l}{\textit{sPRES}} &       &       &       &       &  \\
    $\epsilon=10^{-2}$ & 0.94 & 0.91 & 0.88 & 0.925 &        & 0.92 & 0.905 & 0.95 & 0.945 \\
    $\epsilon=10^{-5}$ & 0.94 & 0.92 & 0.88 & 0.925 &        & 0.92 & 0.905 & 0.945 & 0.945 \\
    $\epsilon=10^{-7}$ & 0.939 & 0.909 & 0.873 & 0.924 &        & 0.925 & 0.915 & 0.95 & 0.94 \\
    \hline
    \end{tabular}%
    \begin{tablenotes}%
    \item For $n=200$, $1$, $9$ and $17$ estimates from PRES with $\epsilon=10^{-2}$, $\epsilon=10^{-5}$ and $\epsilon=10^{-7}$ and $1$ and $3$ estimates from sPRES with $\epsilon=10^{-2}$ and $\epsilon=10^{-7}$ are not positive definite respectively out of the $200$ replications. For $n=400$, $5$ and $9$ estimates from PRES with $\epsilon=10^{-5}$ and $\epsilon=10^{-7}$ and $1$ estimate from sPRES with $\epsilon=10^{-7}$  are not positive definite respectively out of the $200$ replications. The corresponding replications are removed from the calculation of coverage of the post-selection confidence intervals based on the PRES and sPRES methods. All numbers are rounded to three demical places.
  \end{tablenotes}%
  \end{threeparttable}
  \label{tab:1.aic}%
\end{table}

\begin{table}
\centering
  \caption{Simulation results on the coverage of confidence intervals for $\beta=0.5$, AIC}
  \begin{threeparttable}
    \begin{tabular}{cccccccccc}
    \hline
    \multirow{2}[4]{*}{Coverage} & \multicolumn{4}{c}{$n = 200$} &       & \multicolumn{4}{c}{$n = 400$} \\
\cmidrule{2-5} \cmidrule{7-10}
 & $\beta_1$ & $\beta_2$ & $\beta_9$ & $\beta_{10}$ &       & $\beta_1$ & $\beta_2$ & $\beta_9$ & $\beta_{10}$ \\ \hline
    \multicolumn{4}{l}{\textit{PRES+LS}}      &       &       &       &       &  \\
    $\epsilon=10^{-2}$ & 0.96 & 0.955 & 0.97 & 0.98 &       & 0.96 & 0.97 & 0.955 & 0.975 \\
    $\epsilon=10^{-5}$ & 0.96 & 0.955 & 0.97 & 0.975 &        & 0.955 & 0.965 & 0.945 & 0.975 \\
    $\epsilon=10^{-7}$ & 0.96 & 0.955 & 0.965 & 0.975 &        & 0.955 & 0.965 & 0.945 & 0.975 \\[5pt]
    \multicolumn{4}{l}{\textit{sPRES+LS}} &       &       &       &       &  \\
    $\epsilon=10^{-2}$ & 0.96 & 0.955 & 0.97 & 0.98 &      & 0.96 & 0.97 & 0.955 & 0.975 \\
    $\epsilon=10^{-5}$ & 0.955 & 0.955 & 0.97 & 0.98 &      & 0.96 & 0.97 & 0.96 & 0.965 \\
    $\epsilon=10^{-7}$ & 0.955 & 0.955 & 0.965 & 0.98 &        & 0.96 & 0.97 & 0.955 & 0.965 \\[5pt]
    \multicolumn{4}{l}{\textit{PRES}} &       &       &       &       &  \\
    $\epsilon=10^{-2}$ & 0.905 & 0.89 & 0.94 & 0.94 &       & 0.935 & 0.945 & 0.925 & 0.955 \\
    $\epsilon=10^{-5}$ & 0.895 & 0.9 & 0.94 & 0.94 &      & 0.935 & 0.945 & 0.925 & 0.955 \\
    $\epsilon=10^{-7}$ & 0.895 & 0.9 & 0.935 & 0.94 &        & 0.945 & 0.945 & 0.93 & 0.955 \\[5pt]
    \multicolumn{4}{l}{\textit{sPRES}} &       &       &       &       &  \\
    $\epsilon=10^{-2}$ & 0.905 & 0.89 & 0.94 & 0.94 &         & 0.935 & 0.945 & 0.925 & 0.955 \\
    $\epsilon=10^{-5}$ & 0.905 & 0.895 & 0.94 & 0.935 &       & 0.935 & 0.945 & 0.915 & 0.955 \\
    $\epsilon=10^{-7}$ & 0.905 & 0.895 & 0.93 & 0.935 &         & 0.935 & 0.945 & 0.91 & 0.955 \\
    \hline
    \end{tabular}%
  \end{threeparttable}
  \label{tab:05.aic}%
\end{table}

\begin{figure}[htb]
\centering
  \subfloat[$n=200$]{%
    \includegraphics[scale = 0.6]{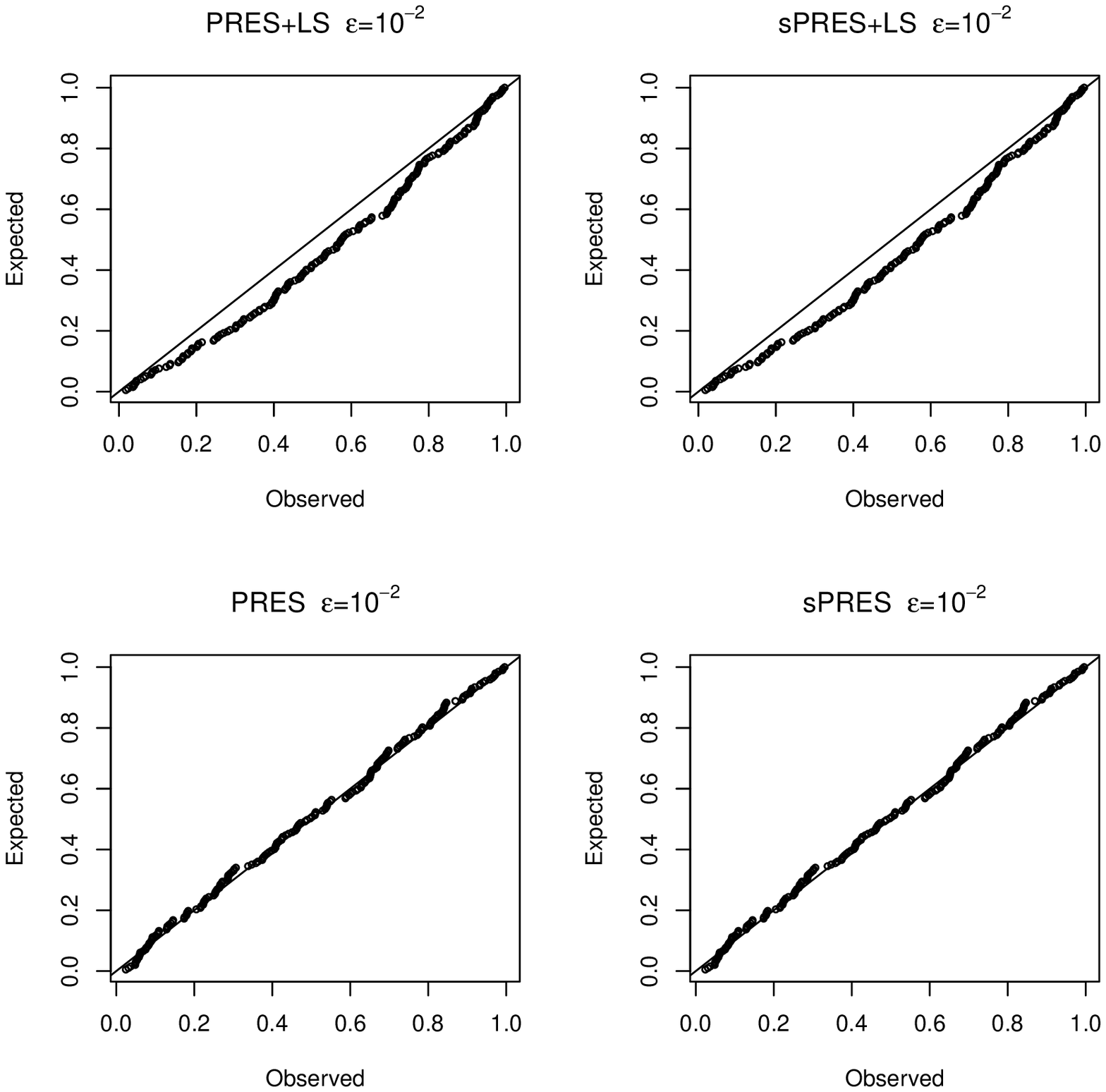}}\\
  \subfloat[$n=400$]{%
    \includegraphics[scale = 0.6]{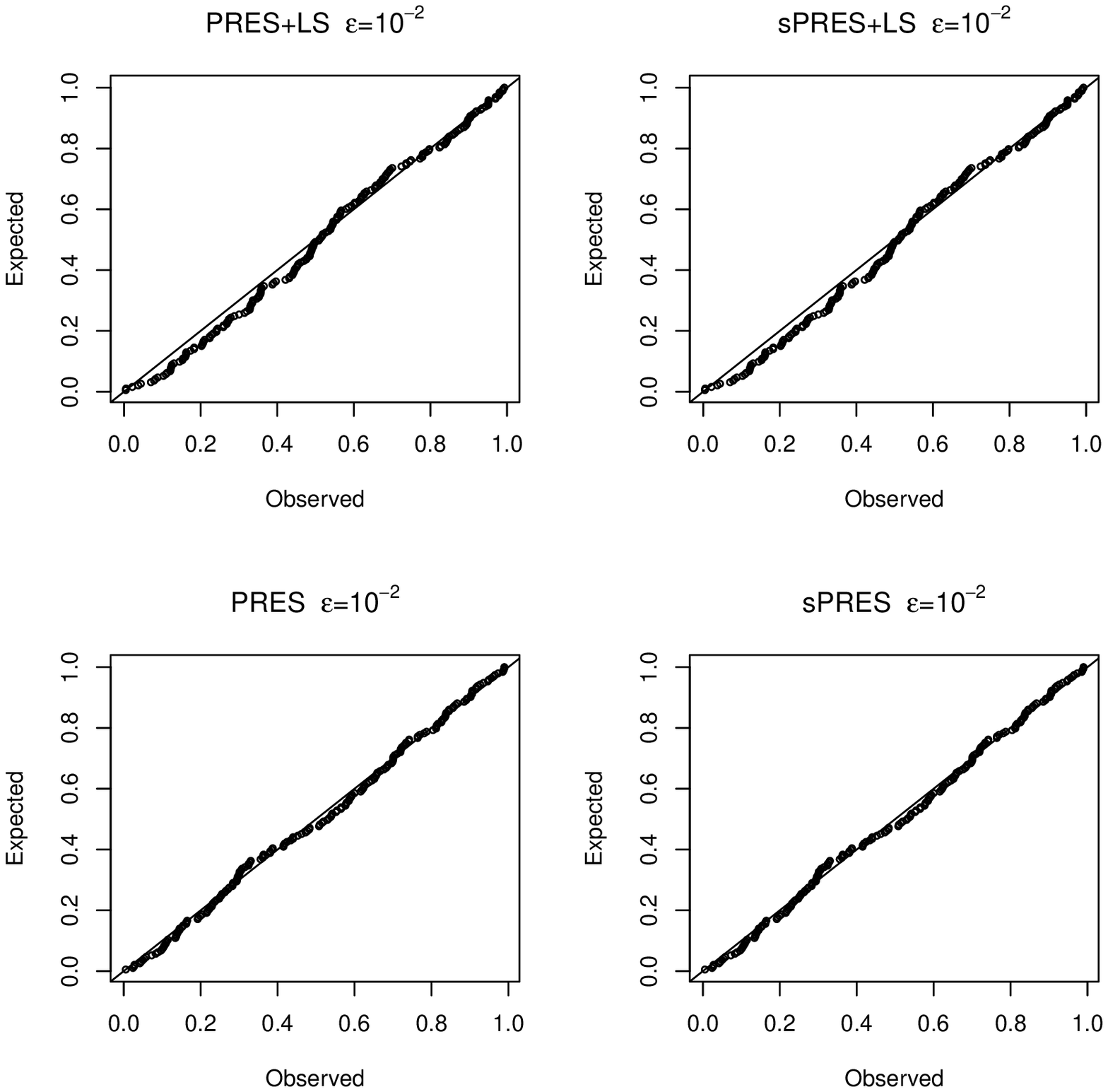}}
  \caption{Uniform Q-Q plots of the null $p$-values for $\lambda=Cn^{1/2}$ and $\epsilon=10^{-2}$ under $\beta=1$.}\label{fig:1.2.fixed}
\end{figure}

\begin{figure}[htb]
\centering
  \subfloat[$n=200$]{%
    \includegraphics[scale = 0.6]{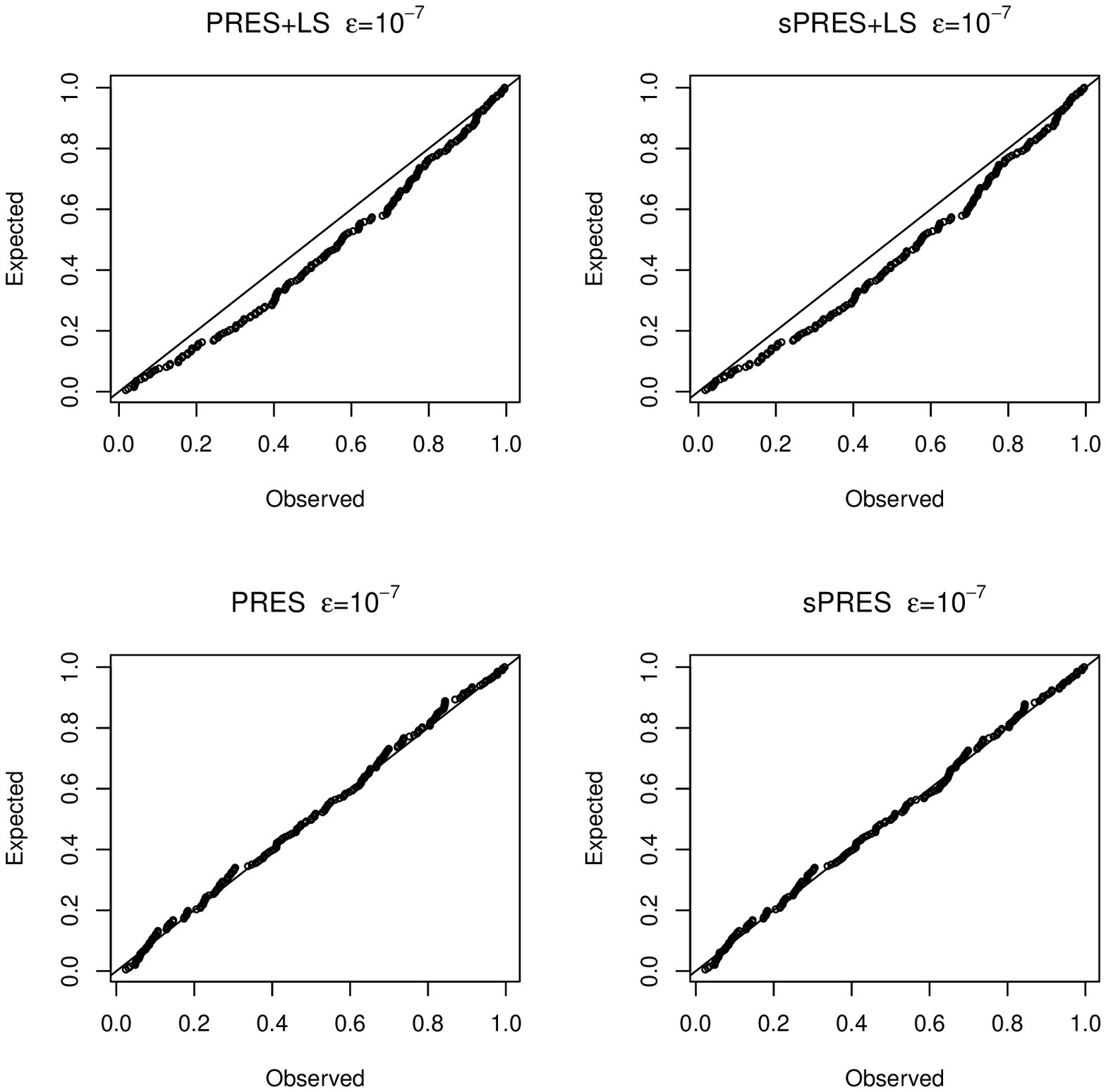}}\\
  \subfloat[$n=400$]{%
    \includegraphics[scale = 0.6]{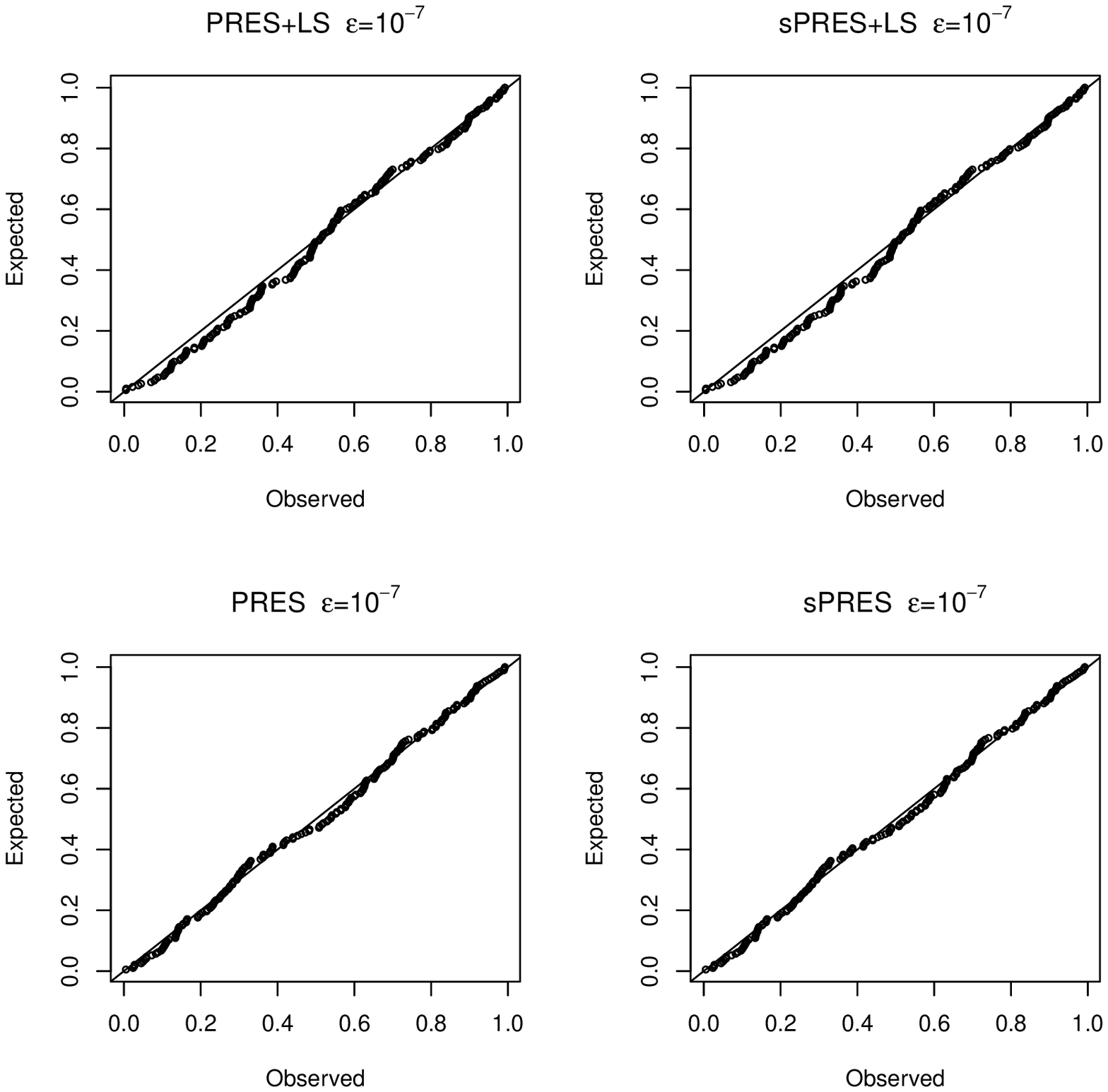}}
  \caption{Uniform Q-Q plots of the null $p$-values for $\lambda=Cn^{1/2}$ and $\epsilon=10^{-7}$ under $\beta=1$.}\label{fig:1.7.fixed}
\end{figure}

\begin{figure}[htb]
\centering
  \subfloat[$n=200$]{%
    \includegraphics[scale = 0.6]{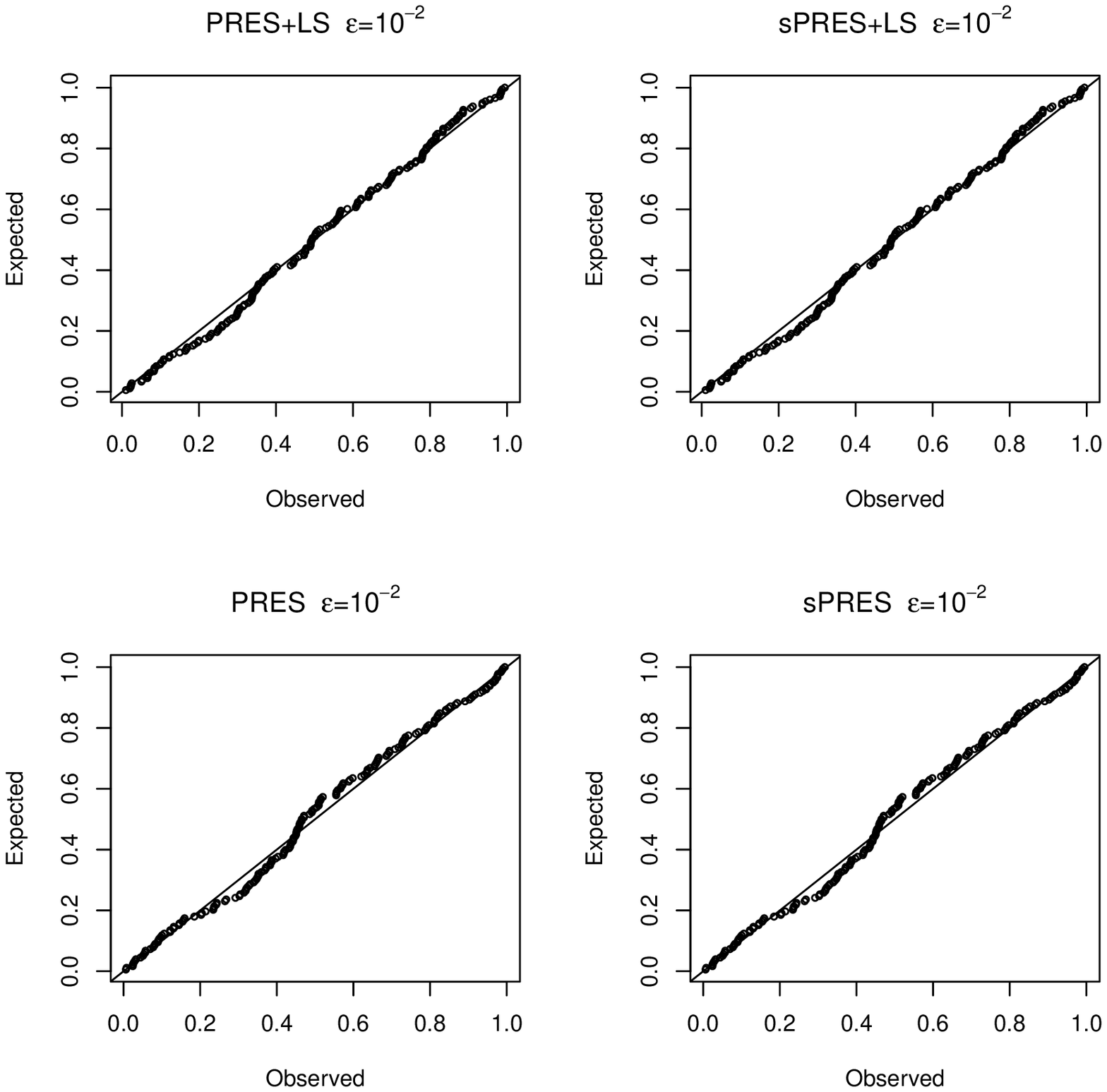}}\\
  \subfloat[$n=400$]{%
    \includegraphics[scale = 0.6]{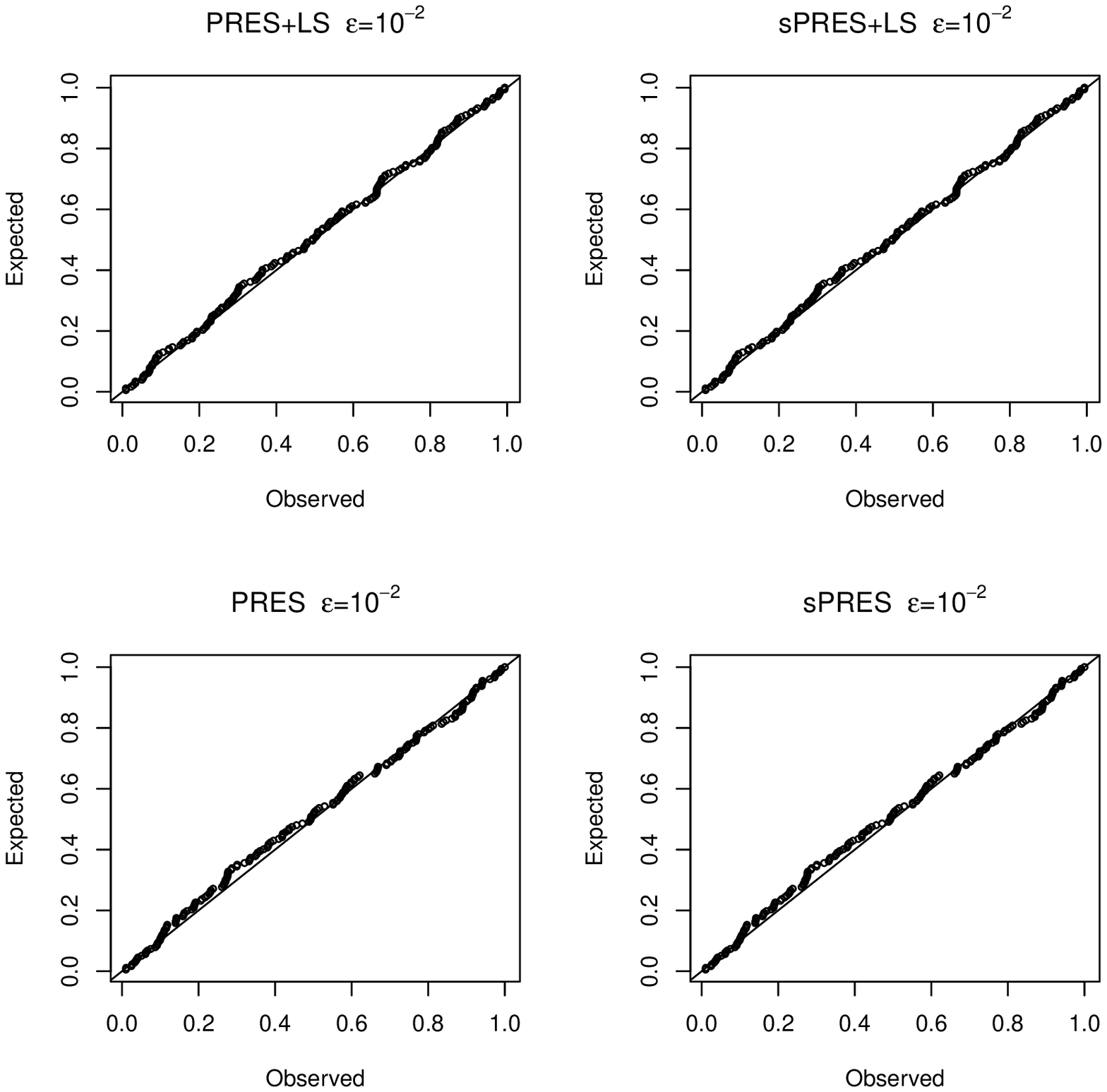}}
  \caption{Uniform Q-Q plots of the null $p$-values for $\lambda=Cn^{1/2}$ and $\epsilon=10^{-2}$ under $\beta=0.5$.}\label{fig:05.2.fixed}
\end{figure}

\begin{figure}[htb]
\centering
  \subfloat[$n=200$]{%
    \includegraphics[scale = 0.6]{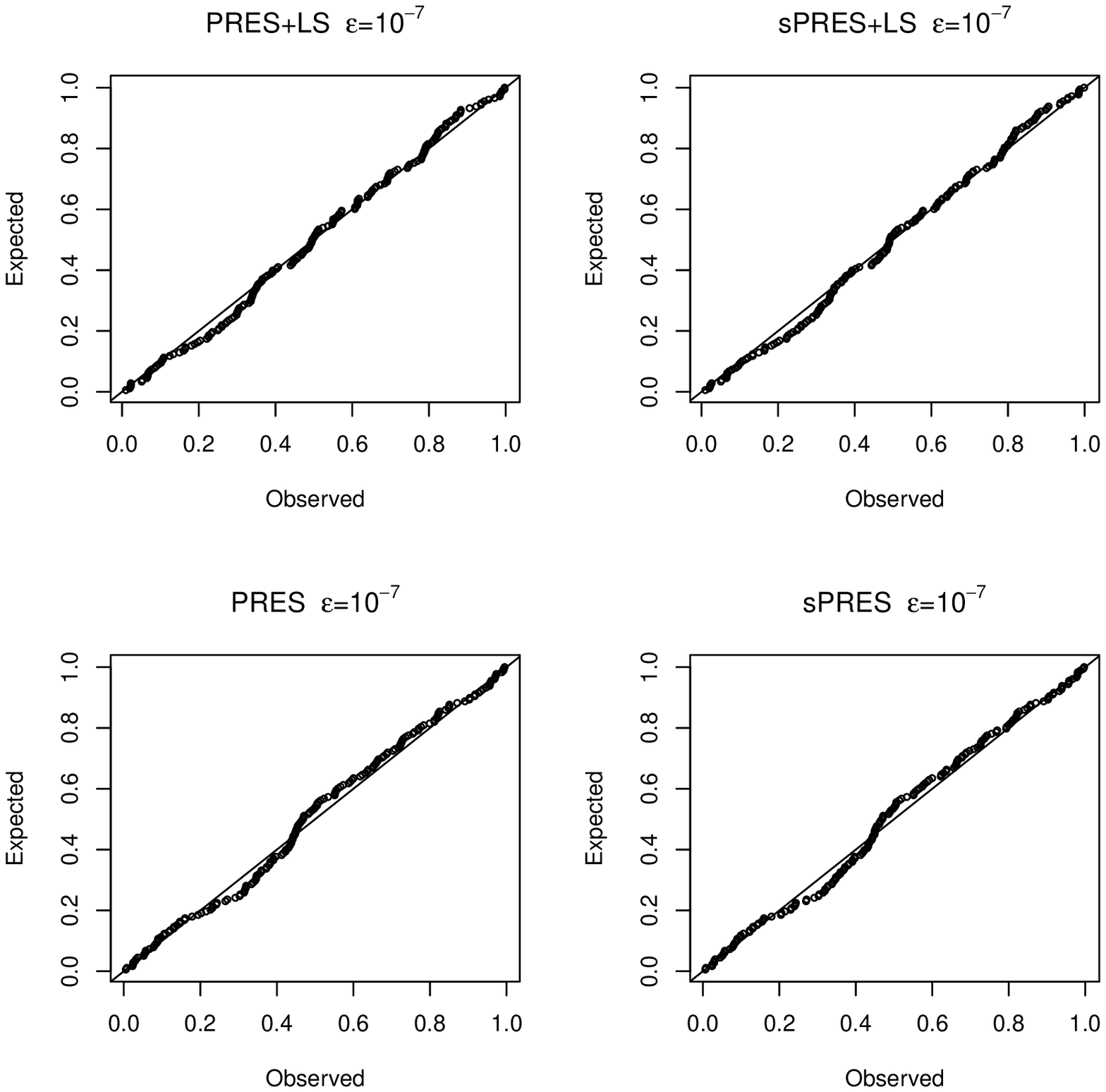}}\\
  \subfloat[$n=400$]{%
    \includegraphics[scale = 0.6]{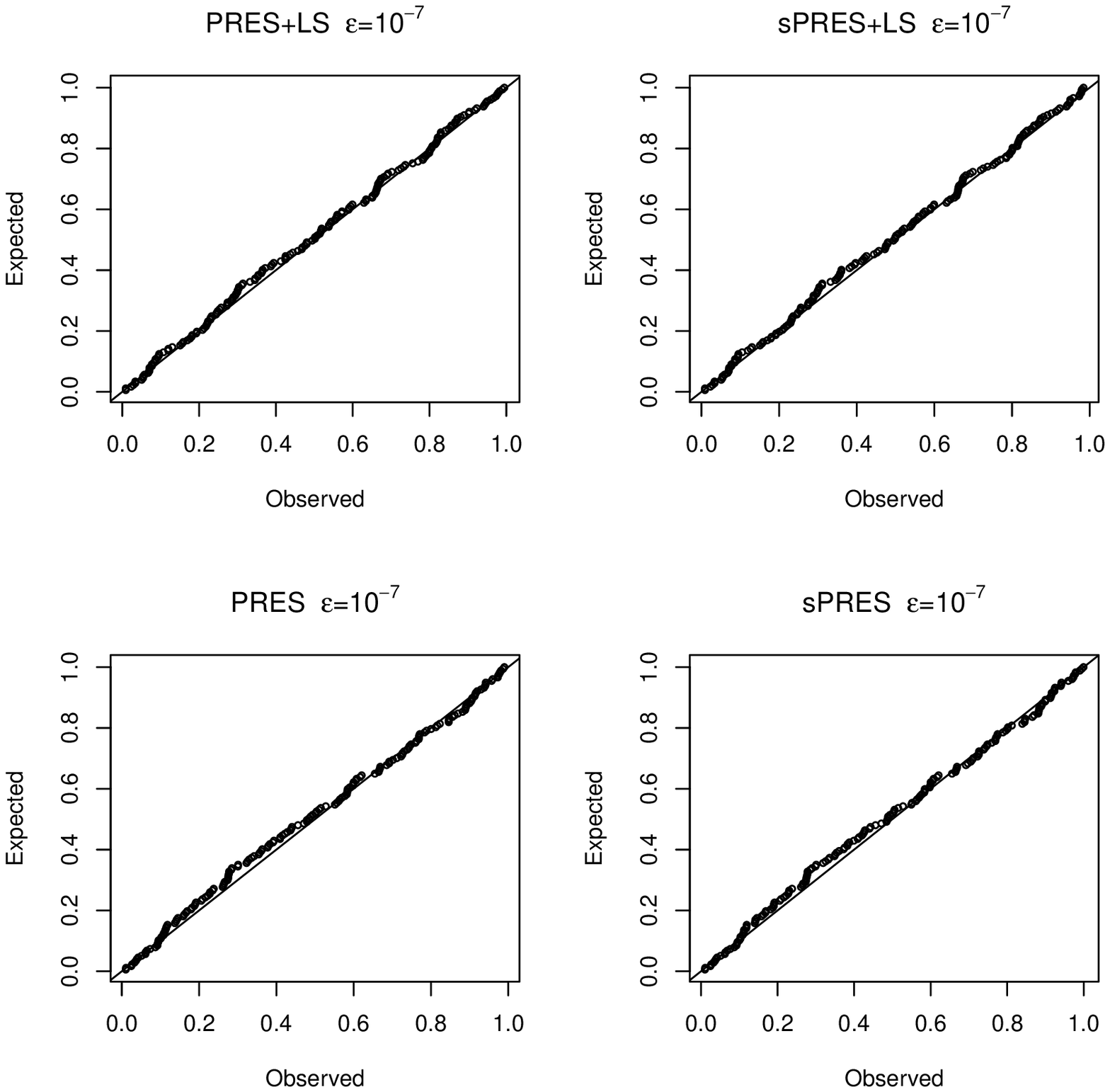}}
  \caption{Uniform Q-Q plots of the null $p$-values for $\lambda=Cn^{1/2}$ and $\epsilon=10^{-7}$ under $\beta=0.5$.}\label{fig:05.7.fixed}
\end{figure}

\begin{figure}[htb]
\centering
  \subfloat[$n=200$]{%
    \includegraphics[scale = 0.6]{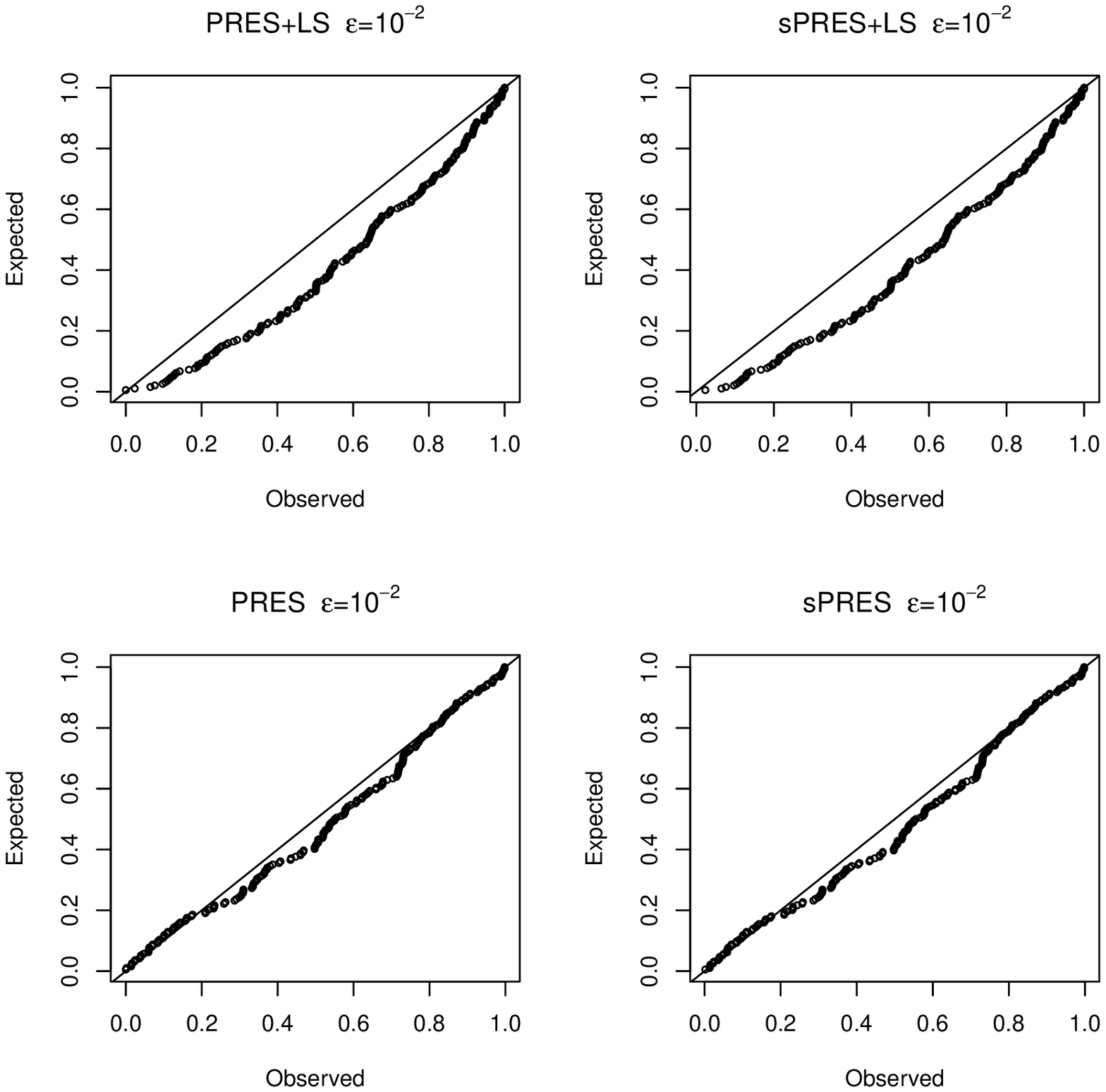}}\\
  \subfloat[$n=400$]{%
    \includegraphics[scale = 0.6]{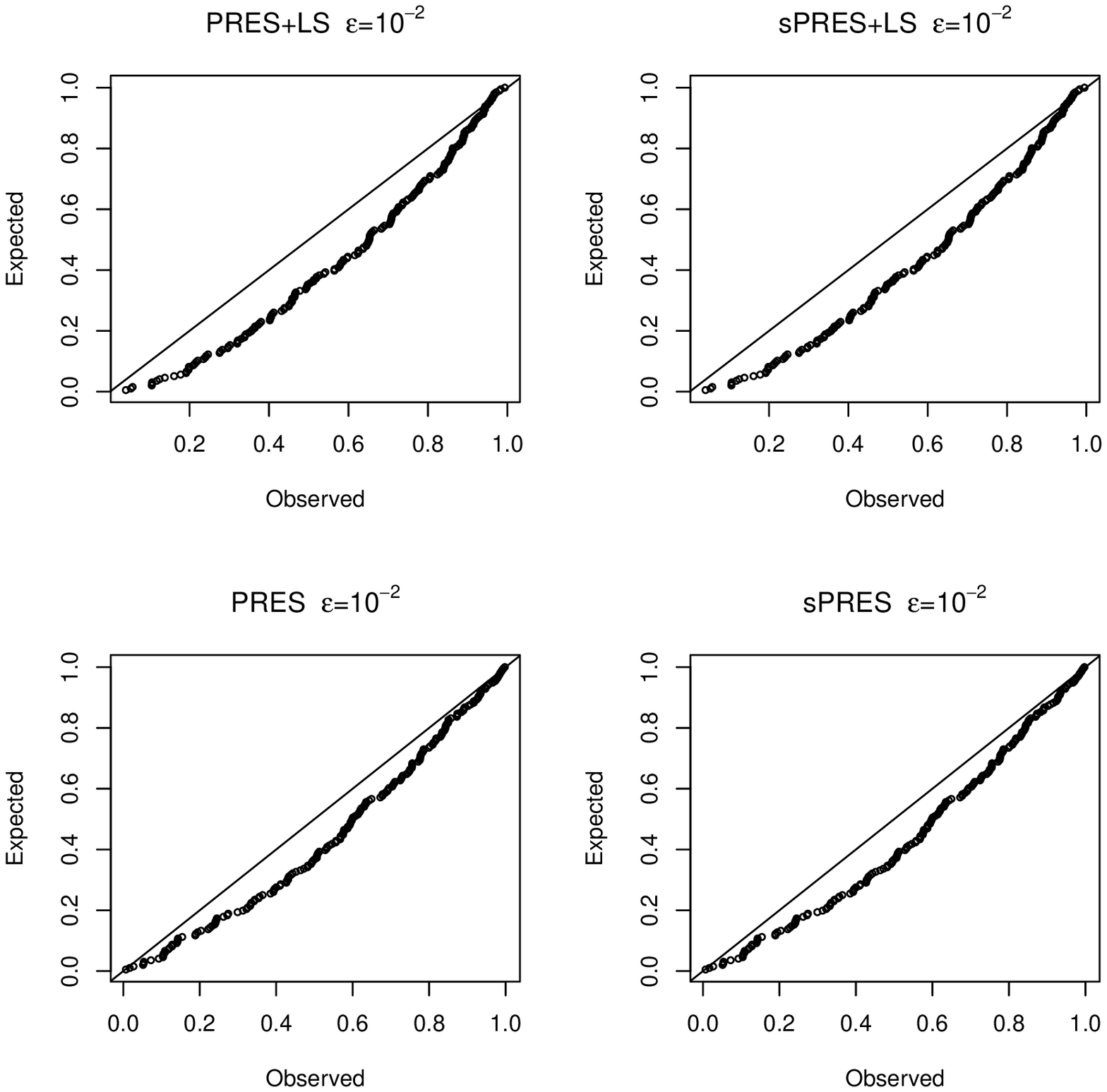}}
  \caption{Uniform Q-Q plots of the null $p$-values with $\lambda$ selected via AIC and $\epsilon=10^{-2}$ under $\beta=1$.}\label{fig:1.2.aic}
\end{figure}

\begin{figure}[htb]
\centering
  \subfloat[$n=200$]{%
    \includegraphics[scale = 0.6]{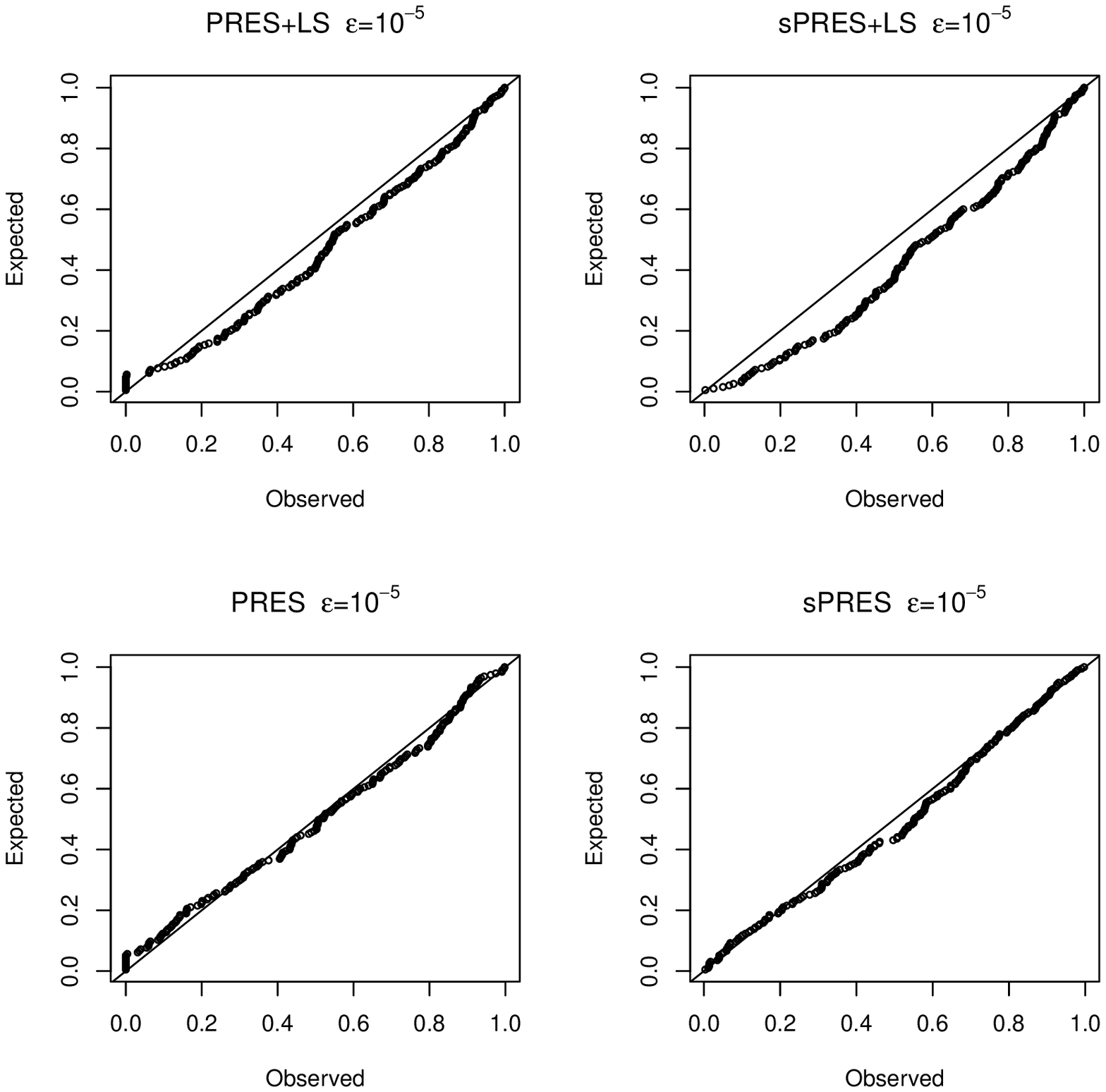}}\\
  \subfloat[$n=400$]{%
    \includegraphics[scale = 0.6]{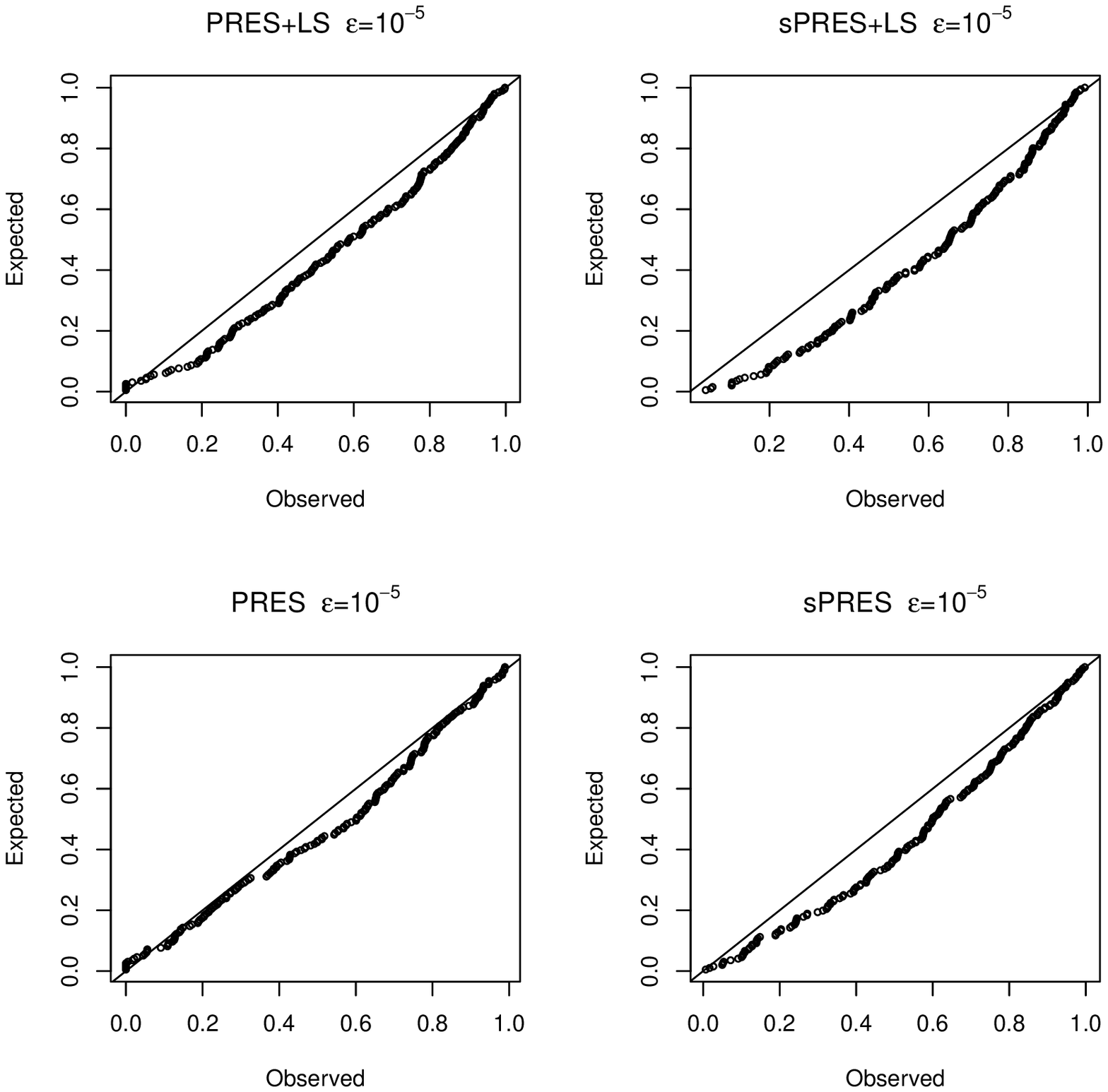}}
  \caption{Uniform Q-Q plots of the null $p$-values with $\lambda$ selected via AIC and $\epsilon=10^{-5}$ under $\beta=1$.}\label{fig:1.5.aic}
\end{figure}

\begin{figure}[htb]
\centering
  \subfloat[$n=200$]{%
    \includegraphics[scale = 0.6]{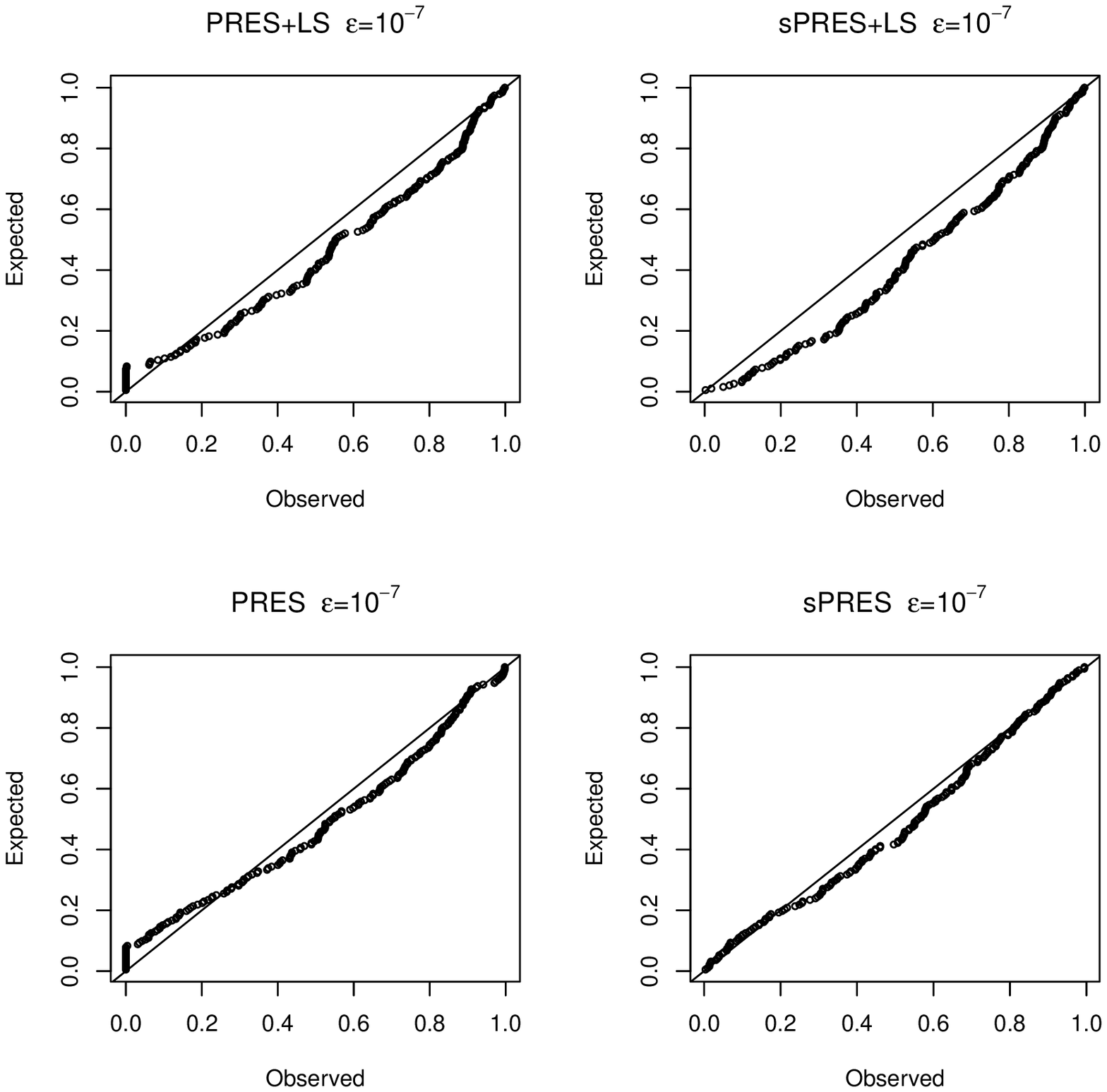}}\\
  \subfloat[$n=400$]{%
    \includegraphics[scale = 0.6]{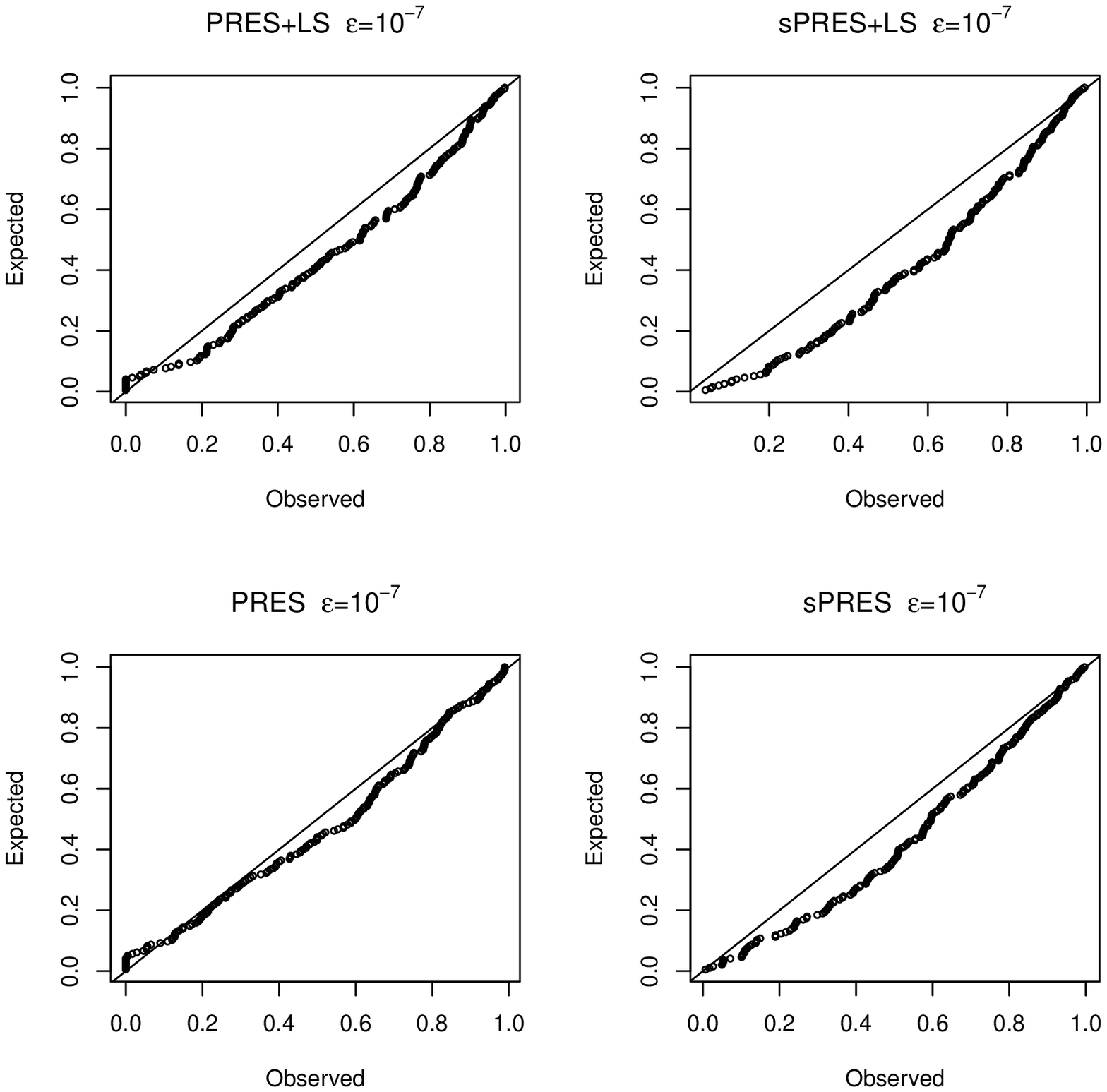}}
  \caption{Uniform Q-Q plots of the null $p$-values with $\lambda$ selected via AIC and $\epsilon=10^{-7}$ under $\beta=1$.}\label{fig:1.7.aic}
\end{figure}

\begin{figure}[htb]
\centering
  \subfloat[$n=200$]{%
    \includegraphics[scale = 0.6]{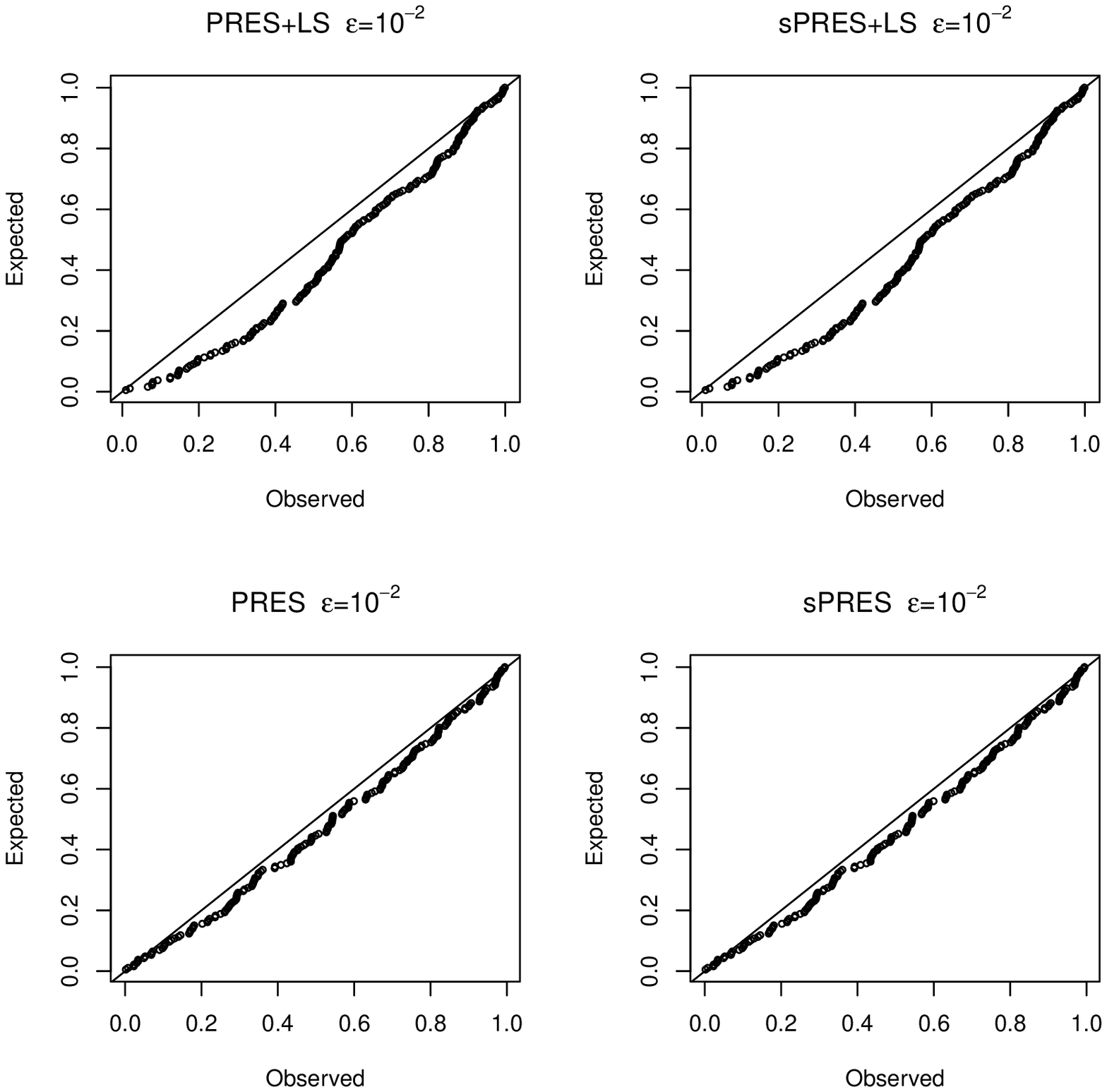}}\\
  \subfloat[$n=400$]{%
    \includegraphics[scale = 0.6]{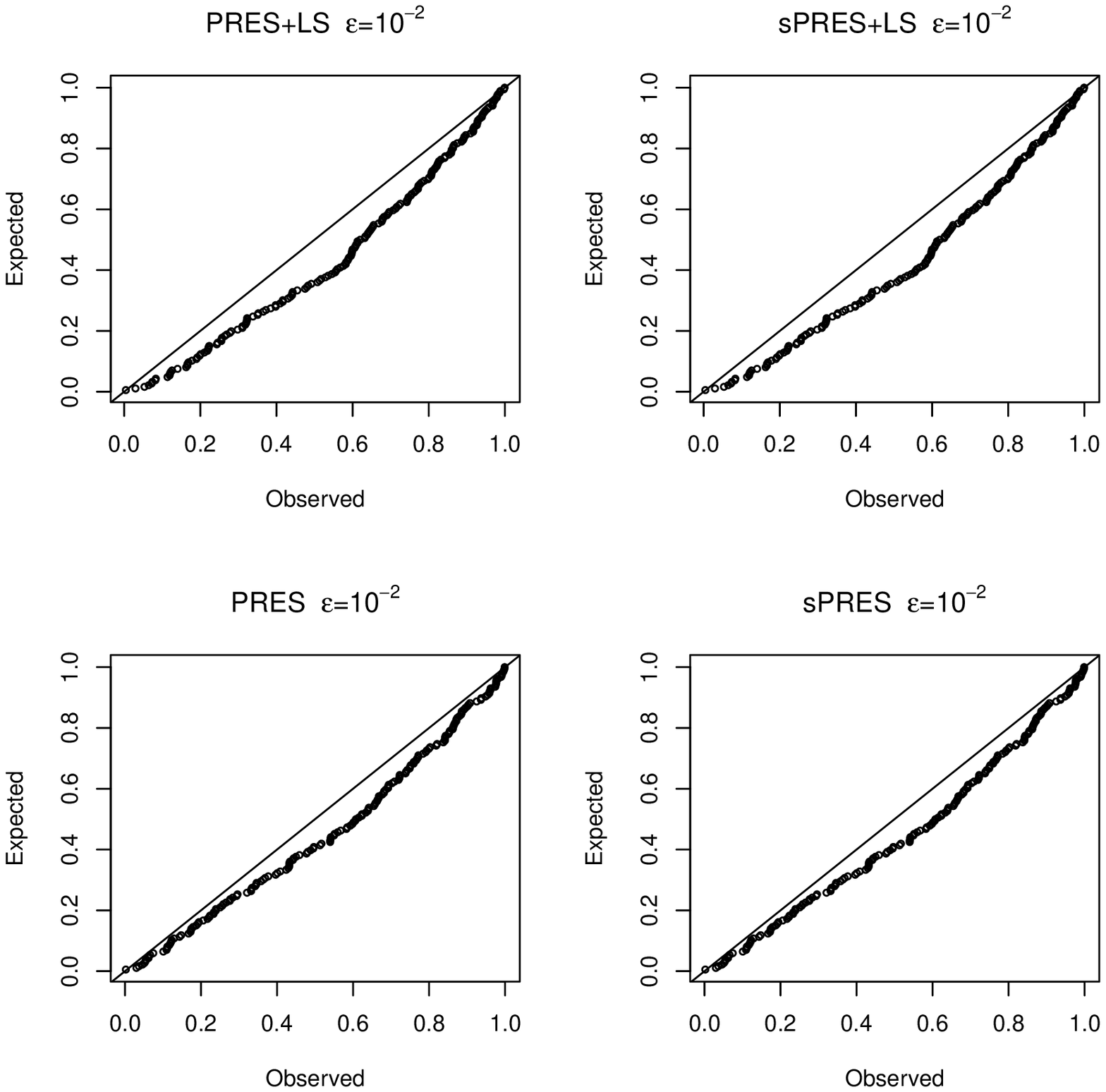}}
  \caption{Uniform Q-Q plots of the null $p$-values with $\lambda$ selected via AIC and $\epsilon=10^{-2}$ under $\beta=0.5$.}\label{fig:05.2.aic}
\end{figure}

\begin{figure}[htb]
\centering
  \subfloat[$n=200$]{%
    \includegraphics[scale = 0.6]{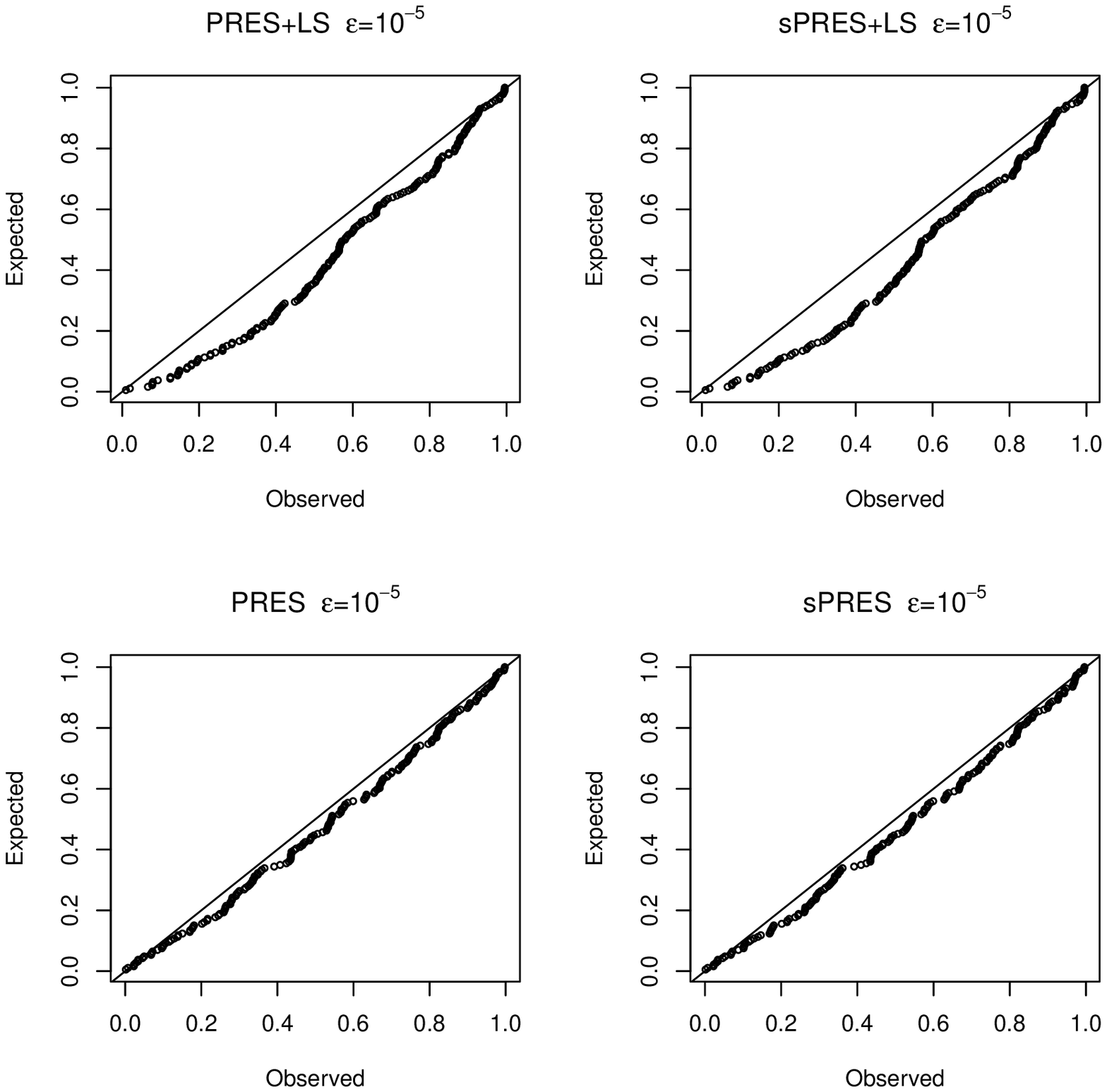}}\\
  \subfloat[$n=400$]{%
    \includegraphics[scale = 0.6]{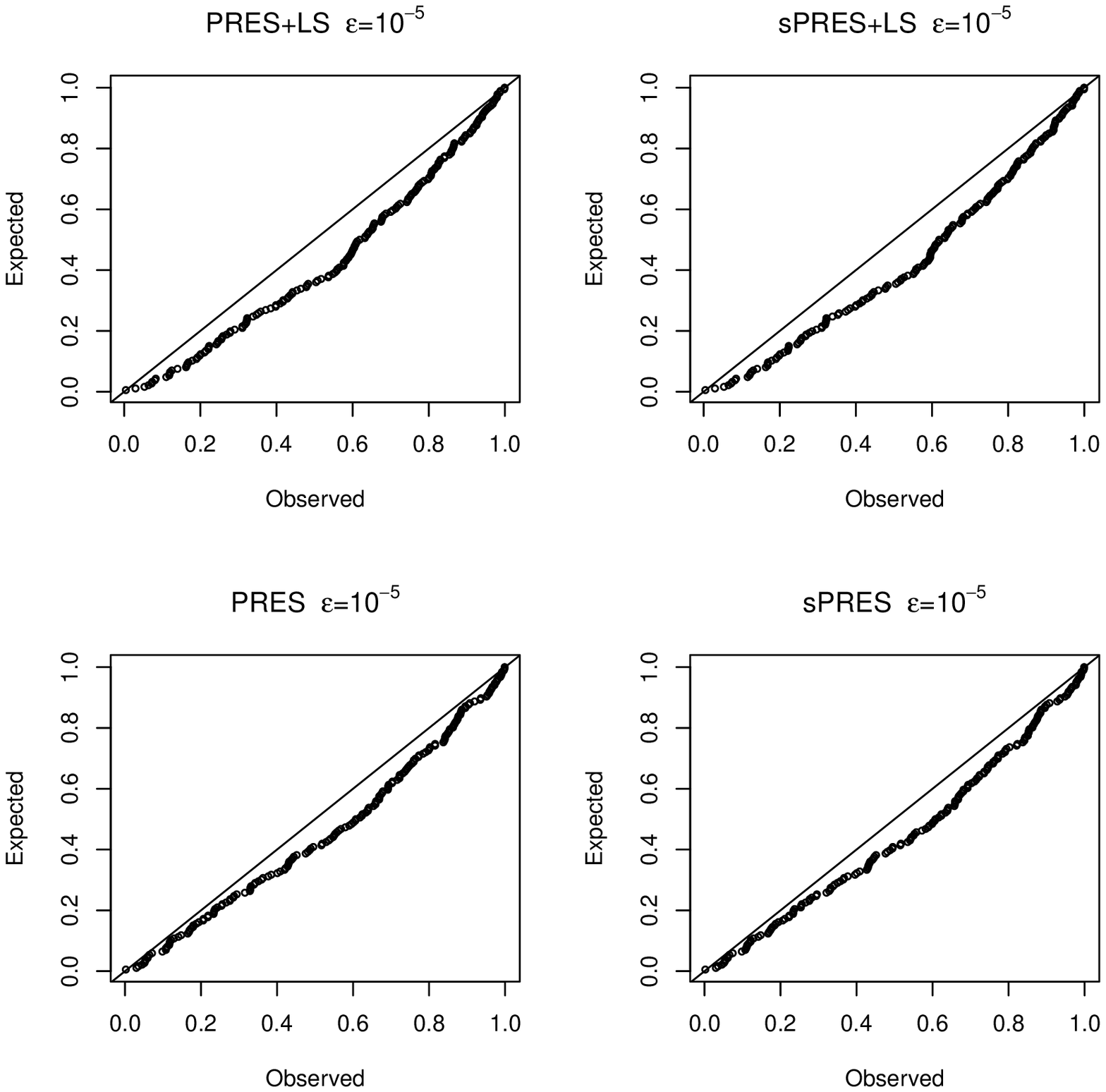}}
  \caption{Uniform Q-Q plots of the null $p$-values with $\lambda$ selected via AIC and $\epsilon=10^{-5}$ under $\beta=0.5$.}\label{fig:05.5.aic}
\end{figure}

\begin{figure}[htb]
\centering
  \subfloat[$n=200$]{%
    \includegraphics[scale = 0.6]{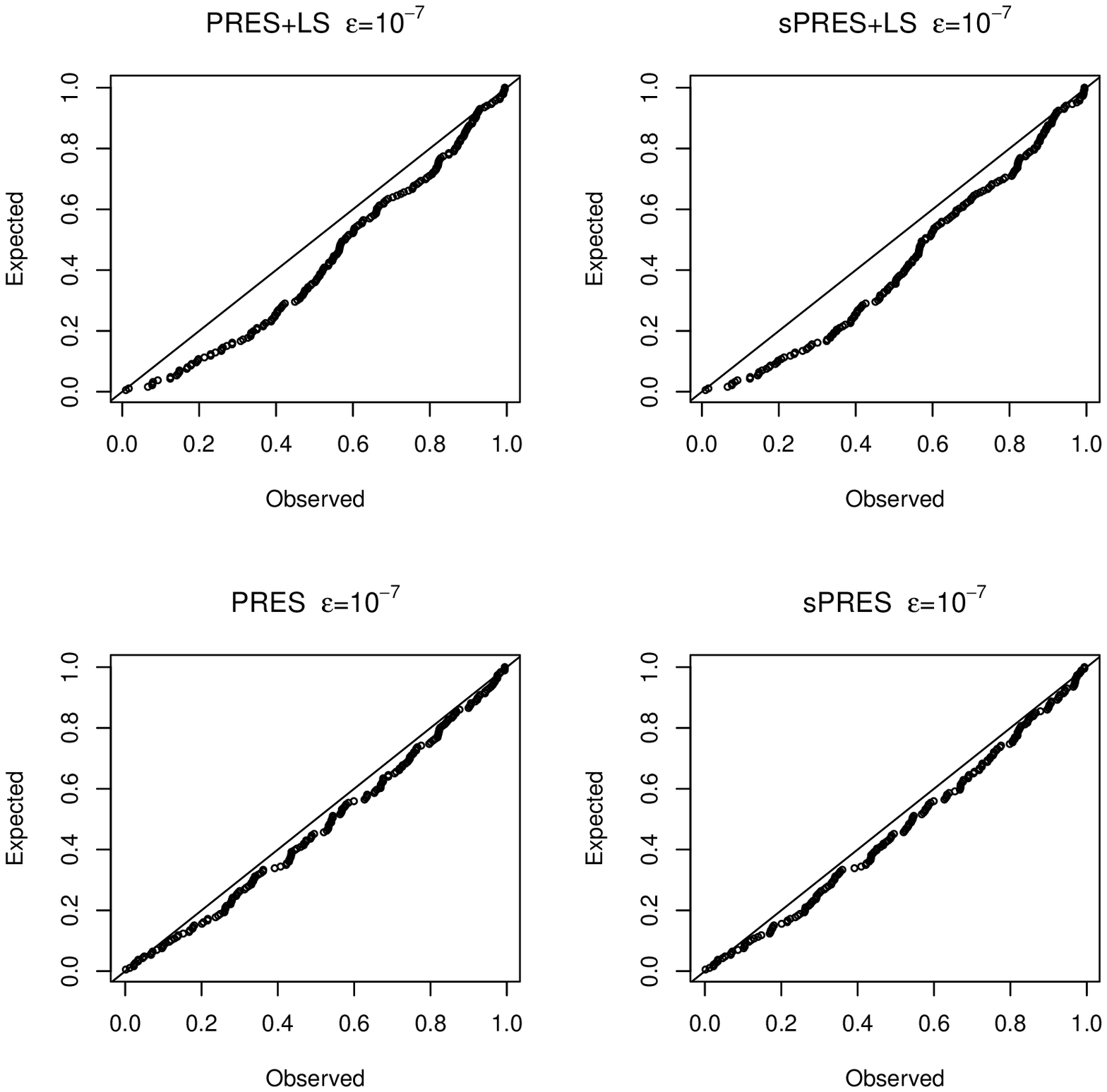}}\\
  \subfloat[$n=400$]{%
    \includegraphics[scale = 0.6]{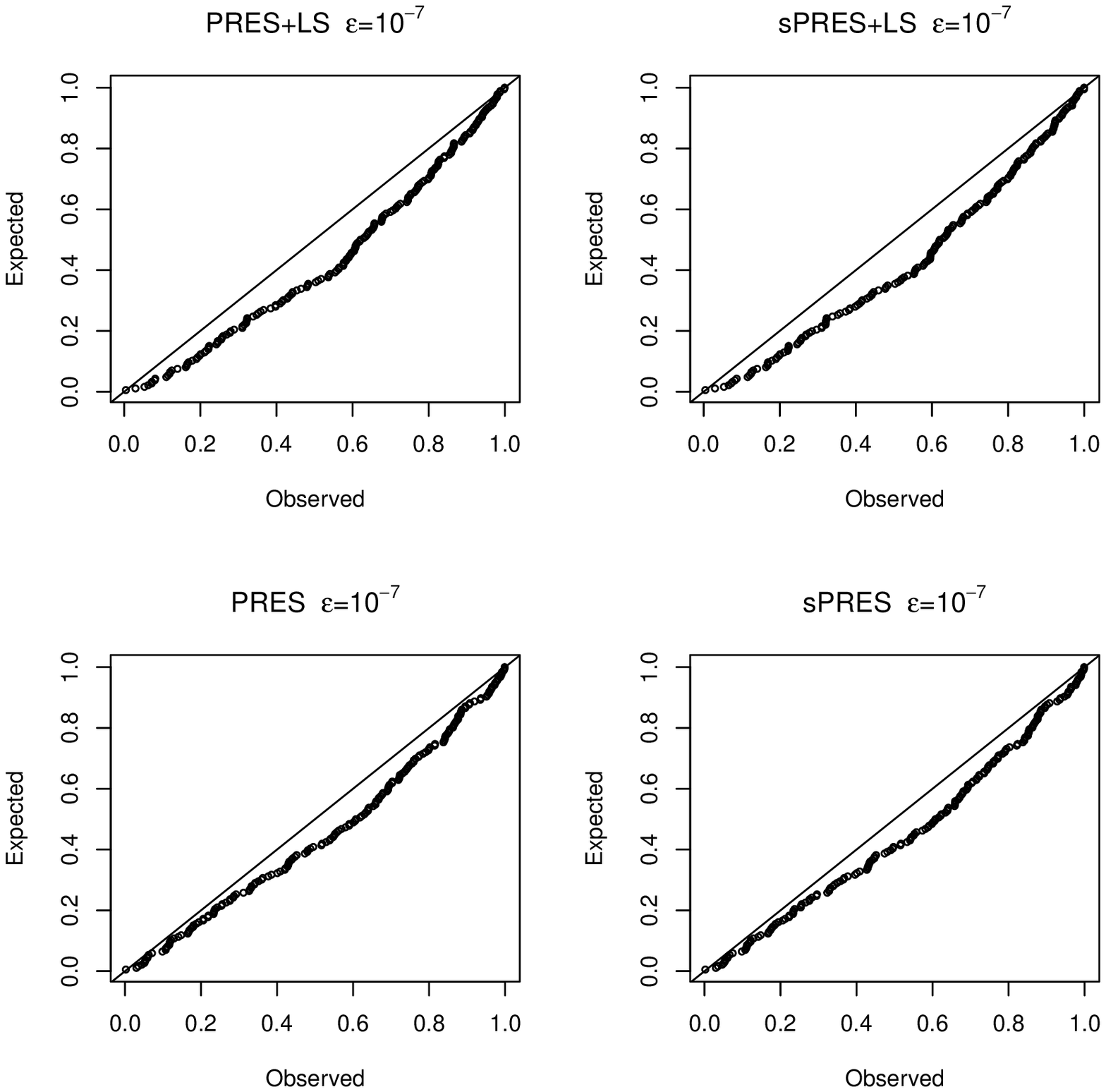}}
  \caption{Uniform Q-Q plots of the null $p$-values with $\lambda$ selected via AIC and $\epsilon=10^{-7}$ under $\beta=0.5$.}\label{fig:05.7.aic}
\end{figure}

\FloatBarrier

\bibliographystyle{biometrika}
\bibliography{mainNov18}
\end{document}